\documentclass[times,preprint,3p]{elsarticle}

\usepackage{framed,multirow}

\usepackage{amssymb}
\usepackage{latexsym}

\usepackage{url}
\usepackage{xcolor}
\usepackage{bm}
\definecolor{newcolor}{rgb}{.8,.349,.1}

\usepackage{bm}

\usepackage[utf8]{inputenc}
\usepackage[T1]{fontenc}
\usepackage{mathptmx}
\usepackage{etoolbox}
\usepackage{amsmath}
\usepackage[standard]{ntheorem} 
\usepackage{amssymb} 
\usepackage{hyperref}

\usepackage{csquotes}
\usepackage{subfigure}

\theoremstyle{plain}
\newcommand{\cqfd}{\hfill \ensuremath{\Box}}

\newcommand{\real}{\mathbb{R}}
\newcommand{\comp}{\mathbb{C}}

\newcommand{\HS}{{\cal H}}

\newcommand{\alf}{\scriptscriptstyle{1/2}}

\newcommand{\mbs}[1]{\ensuremath{\boldsymbol{#1}}}
\newcommand{\bcdot}{\bm{\cdot}}
\newcommand{\Mop}{{\cal M}}
\newcommand{\XX}{X}
\newcommand{\YY}{Y}
\newcommand{\dif}{{\mathrm{d}}}
\newcommand{\XXi}{\XX_t^{(i)}}
\newcommand{\XXj}{\XX_t^{(j)}}
\newcommand{\XXiz}{\XX_0^{(i)}}
\newcommand{\XXjz}{\XX_0^{(j)}}

\newcommand{\AM}[1]{\textcolor{black}{#1}}
\newcommand{\AD}[1]{\textcolor{black}{#1}}
\newcommand{\GT}[1]{\textcolor{blue}{#1}}

\begin{document}
\begin{frontmatter}
\title{Ensemble forecasts in reproducing kernel Hilbert space family }
\author[1]{Benjamin Duf\'ee}
\author[1]{B\'erenger Hug}
\author[1]{\'{E}tienne M\'emin}%
\cortext[cor1]{Corresponding author: 
 Etienne.Memin@inria.fr}
\author[1]{Gilles Tissot}
\address[1]{INRIA/IRMAR Campus universitaire de Beaulieu, Rennes, 35042 Cedex, France}
\date{\today}
\begin{abstract}
A methodological framework for ensemble-based estimation and simulation of high dimensional dynamical systems such as the oceanic or atmospheric flows is proposed. To that end, the dynamical system is embedded in a family of reproducing kernel Hilbert spaces \AD{(RKHS)} with kernel functions driven by the dynamics. 
In \AD{the RKHS family},  the Koopman and Perron-Frobenius operators are unitary and uniformly continuous. This property warrants they can be  expressed in exponential series of diagonalizable bounded \AD{evolution operators defined from their} infinitesimal generators. Access to Lyapunov exponents and to exact ensemble based expressions of the  tangent linear dynamics are directly available as well.  \AD{The RKHS family} enables us the devise of strikingly simple ensemble data assimilation methods for trajectory reconstructions in terms of  constant-in-time linear combinations of trajectory samples. Such an embarrassingly simple  strategy is made possible through a fully justified superposition principle ensuing from several fundamental theorems.  
\end{abstract}
\end{frontmatter}

\section{Introduction}
Forecasting the state of geophysical fluids has become of paramount importance for our daily life  either at short time scale, for weather prediction, or at long time scale, for climate study. There is in particular a strong need to provide reliable likely scenarios for probabilistic forecasts and uncertainty quantification. To that end, methods based on linear combinations of an ensemble of numerical simulations are becoming more and more ubiquitous.
Ensemble forecasting and data assimilation techniques, where the linear coefficients are further constrained through partial observations of the system and Gaussian stochastic filtering strategies, are good examples  developed in meteorological centers for weather forecast operational routine use. 
However, although efficient in practice, this linear superposition principle of solutions remains theoretically questionable for nonlinear dynamics.
Very recently, such techniques have been further coupled with machine learning to characterize surrogate dynamical models from long time-series of observations.
This includes neural networks, analog forecasting, reservoir computing, kernel methods  to name only a few of the works of this wide emerging research effort \citep{Bocquet-et-al-20, Brunton-Pnas-16,Fablet-JAMES-21,Gottwald-Reich-21,Hamzi-Phys-D-I-21,Owhadi-JCP-19,Pathak-17,Pathak-PRL-18,Raissi-19, Zhao-Giannakis-16}.
To characterize a surrogate model of the dynamics, these techniques are built upon principles such as delay coordinates embedding, empirical basis functions, or sparse linear representations.
The difficulty here is to get sufficiently long data series to represent the manifold associated to a given observable of the system as well as to exhibit an intrinsic representation of the whole system that does not depend on parameters such as the discrete time-stepping, the integration time window, or the initial conditions. Looking for such intrinsic representation often leads to focus (at least theoretically) on reversible ergodic systems.

Spectral representation of the Koopman operator \citep{Koopman-31} or of its infinitesimal generator constitutes one of the most versatile approaches to achieve this objective as they enable in theory to extract  intrinsic eigenfunctions of the system dynamics. Koopman operator or its dual, the Perron-Frobenious operator, have attracted a great deal of attention in the operator-theoretic branch of ergodic theory \citep{EFHN-15} as well as in data science for the study of data driven modeling of dynamical systems.  Several techniques, starting from \citep{Dellnitz-99} and \cite{Mezic-05}, have been proposed to estimate such spectral representation. Most of them rely either on longtime harmonic averages \cite{Mezic-05,budisic2012,Das-Giannakis-19, Das-Giannakis-20} or on finite-dimensional approximations of the Koopman operator such as the dynamic modes representation and its extensions \citep{Rowley09,Schmid10,tu2014,Williams-et-al-15,williams2016,Kutz-16,degennaro2019,buza2021,pan2021,colbrok2021,baddoo2022} or on  Galerkin approximation and delay embedding \citep{Brunton17,Giannakis-15,Giannakis17}. \AD{See also the link between singular spectrum analysis (SSA) method \cite{Vautard-89}, data-adaptive harmonic decomposition \cite{Kondrashov-et-al-20},  Hankel alternative \cite{Brunton17} and Koopman analysis \cite{Zhen-et-al-2022, Zhen-et-al-2023} for ergodic systems.}

These  methods enable performing spectral projection from long time series of measured observables of the system. The practical linear nature of the {\it compositional} Koopman operator, defined for any observable, $f$, of a dynamical system $\Phi_t(x)$ with initial condition $x$, as  $(U_t f) (x):= f\circ \Phi_t (x)$, is nevertheless hindered by its infinite dimension. Due to this, it often exhibits a spectrum with a continuous component \citep{mezic2005} and is not amenable to diagonalization nor to any sparsity assumption. Finite dimensional approximations and the direct use of matrix algebra numerical methods raises immediate questions on the pertinence of the associated estimation. To tackle these difficulties several techniques combining regularization, compactification of the generators, and reproducing kernel Hilbert space (RKHS) to work in spaces of smooth functions have been recently proposed \citep{Das-Giannakis-19, Das-Giannakis-20,Das-Giannakis-21}. They are possibly associated with forecast strategies \citep{gilani2021,alexander2020,burov2021} and transfer operator estimations \citep{klus2020,santitissadeekorn2020}. {For data driven specification of dynamical systems RKHS or kernel methods have been also explored in a more direct way in terms of advanced Gaussian process regression techniques \citep{Lee-Phys-D-23,Hamzi-Phys-D-I-21,Hamzi-Proc-RSA-21}. Kernel methods have been found particularly efficient in this context when coupled with a learning strategy for the kernel \citep{Owhadi-JCP-19}}.

Working with RKHS brings convenient regular basis functions together with approximation convergence guaranties. However, the dynamical systems need still to be ergodic, and invertible {to ensure that the Koopman operator has a simple point spectrum and an associated system of orthogonal basis eigenfunctions \citep{EFHN-15}}. Measure invariance and invertibility warrants that the associated Koopman/Perron-Frobenius operators are dual unitary operators  \citep{EFHN-15}. Note nevertheless that state-of-the-art estimation techniques  based on Dynamic Mode Decomposition (DMD) or Hankel matrix \cite{Arbabi-Mezic17,Brunton17} require very long time series to converge to actual Koopman eigenvalues \AD{\cite{Zhen-et-al-2022}}. 
This is all the more  detrimental for high-dimensional state-spaces. \AD{An extended dynamical modes decomposition has also been proposed to approximate the Koopman infinitesimal generator of deterministic or stochastic systems \cite{Klus-et-al-Physica-D-2021}. Like DMD, this approach corresponds to a direct discretization of the Koopman operator.} 

Recent reviews on modern Koopman analysis applications and their related algorithmic developments from  measurements time series can be found in \citep{Brunton-et-al-22,BEVANDA2021197,Otto-Rowley-21}.  {Another limitation is related to the definition of a kernel (and hence of a  RKHS)  independent of the dynamics to describe the system's observable, which { de facto} boils down to the strong assumption of an invariance of the RKHS under the Koopman groups. This condition requires to define the RKHS from the subspace spanned by the Koopman operator eigenfunctions, which excludes in general classical parametric kernels.} Note that in particular bias in the estimation of the Koopman eigenvalues have been observed for kernel-based methods based on finite dimensional approximation and sparsity assumptions \cite{Kostic-22-NIPS,Kostic-23}. \AD{A technique coupling RKHS and the extended  DMD approach of \cite{Klus-et-al-Physica-D-2021}, enabling to provide an approximation of the Koopman generator, was proposed in \cite{Klus-et-al-2020}. In the same way as the previous techniques, an invariance of the RKHS is also assumed.}

In this study we will explore the problem from a different angle. Instead of working with long time series of data, we will work with a set of realizations of a given dynamical system. The objective will be to learn this dynamics locally in the phase space to synthesize in a fast way new trajectories of the system through a theoretically  well justified superposition principle. {This superposition principle will in particular enable us to fully justify the linear combination assumption employed in many ensemble data assimilation filters even though  the dynamics is nonlinear.}
To perform an efficient estimation inherently bound to the dynamics, we propose to embed the ensemble members in a set of time-evolving RKHS defining a family of spaces. This setting is designed to deal with large-scale systems such as oceanic or meteorological flows, for which it is out of the question to explore the whole attractor (if any), neither to run very long time simulations calibrated on real-world data of a likely transient system due to climate change. The present work relies on the fact that the kernel functions, usually called feature maps, between the native space and the RKHS family are transported by the dynamical system. This creates, at any time, an isometry between the evolving RKHS  and the RKHS at a given initial time. Instead of being attached to a system state, the kernel is now fully associated with the dynamics and new ensemble members embedded in the RKHS family at a given time can be very simply forecast at a further time. 
For invertible dynamical systems, the Koopman and Perron-Frobenius operators in \AD{the RKHS family} are unitary, like in $L^2$. However, they are in addition uniformly continuous, with bounded generators, and diagonalizable. As such they can be rigorously expanded in exponential form and their eigenfunctions provide an intrinsic orthonormal basis system related to the dynamics.
This set of analytical properties immediately brings practical techniques to estimate Koopman eigenfunctions or essential features of the dynamics. Evaluations of these Koopman eigenfunctions at the ensemble members are directly available. Finite-time Lyapunov exponents associated with each Koopman eigenfunction are easily accessible on the RKHS family as well and can be used to get practical predictability time of the system. Expressions relating Koopman eigenfunctions to the tangent linear dynamics and its adjoint are also immediately available for data assimilation. These expressions, ensuing from the fundamental results listed previously, correspond in their simpler forms to the ensemble-based finite difference approximations intensively used in ensemble methods. Beyond that, \AD{the RKHS family} finally leads to a theoretically well grounded superposition principle (for nonlinear dynamics) enabling the devise of embarrassingly  simple methods for trajectory reconstructions in terms of  constant-in-time linear combinations of trajectory samples.

The embedding of the trajectories in a time-evolving RKHS which follows the dynamical system can be related to concepts of local kernels \citep{berry2016}, embedding in dynamical coordinates \citep{banish2017} or embedding of changing data over the parameter space as described in \citet{coifman2014}. In our case, we further exploit the relations with the Koopman operator together with an associated isometry, which allows us to define simple practical schemes for the setting of the time-evolving kernels and the modal information that can be extracted from them.

\AM{The main objective of this study is to propose a numerical framework for ensemble-based estimation and simulation of geophysical dynamics driven by systems of partial differential equations (PDE). In this study, we present theoretical elements for analyzing this framework. The mathematical assertions herein are based on idealized assumptions such as an infinite number of ensemble members, invertible dynamical systems and the availability of an ensemble of initial conditions. These theoretical statements will be further refined with more realistic assumptions in future works.}

This paper is organized as follows. We start briefly recalling the definition and some properties of RKHS in the next section. These definitions are complemented by \ref{RKHS} which details and adapts several known results to our functional setting. The embedding of the Koopman and Perron-Frobenius operators in \AD{the RKHS family} are then presented in sections \ref{sec-sys-dyn-rkhs} and \ref{Koopman-W}.
This is followed by the enunciation of our main results.  The main theorems and properties are demonstrated in section \ref{Koopman-W} and in the appendixes. Numerical results for a quasi-geostrophic barotropic ocean model in an idealistic north Atlantic configuration will be presented in section \ref{Sec-Numerical-Results} to support the pertinence of the proposed approach. \AM{This work corresponds to a presentation that has been given to the online seminar on Machine Learning for Dynamical Systems organized by the Turing institute (\href{https://www.youtube.com/channel/UCetvKhuAbnuU1tidgCz9g0g/videos}{link}).}

\section{Reproducing kernel Hilbert spaces}
A RKHS, $\mathcal{H}$, is a Hilbert space of \AM{complex functions} $f:E\mapsto\comp$ defined over a non empty abstract set $E$ on which a positive definite kernel and a kernel-based inner product, $\langle \bcdot \, ,\, \bcdot\rangle_{\HS}$ can be defined. Throughout this work, we will consider that \AM{$E$ is a compact subset of a function vector space as we will work with an integral compact operator from which a convenient functional description can be set. In addition to enable the RKHS functions to be Gateaux differentiable, we will assume as in \cite{Zhou-2008} that $E$ is defined as the closure of its nonempty interior, and we will work with RKHS defined from smooth enough kernels.} RKHS possess remarkable properties, which make their use very appealing in statistical machine learning applications and  interpolation problems \citep{Berlinet-Thomas-Agnan,Cucker-Smale-01}. The kernels from which they are defined have a so called \enquote{reproducing property}. 
\begin{Definition}[Reproducing kernel]
\label{def1}
Let $\HS$ be a Hilbert space of $\mathbb{C}$-valued functions defined on a non-empty compact topological space $E$. A map $k:E\times E \mapsto \mathbb{C}$ is called a reproducing kernel of $\HS$ if it satisfies the following principal properties: $\forall x \in E,$
\begin{itemize}  
\item[$\bullet$] membership of the evaluation function $ \;k(\bcdot, x) \in \HS$,
\item[$\bullet$] reproducing property  $ \forall f \in \HS, \; \bigl\langle f, k(\bcdot,x)\bigr\rangle_{\HS} = f(x)$.
\end{itemize}
\end{Definition}
The last property, provides an expression of the kernel as $k(x,y)= \overline{k(y,x)}=\bigl\langle k(\bcdot,y), k(\bcdot,x)\bigr\rangle_{\HS}$, which is Hermitian -- with $\overline\bullet$ denoting complex conjugate --, positive definite and associated with a continuous evaluation function $\delta_x f= \bigl\langle f,k(\bcdot,x)\bigr\rangle_{\HS}=f(x)$. The continuity of the Dirac evaluation operator is indeed sometimes taken as a definition of RKHS. By the Moore-Aronszajn theorem \citep{Aronszajn-50}, the kernel $k$ defines uniquely the RKHS, $\HS$, and vice versa. The set spanned by the feature maps ${\rm {Span}}\{k(\bcdot,x), x \in E \}$, is dense in $(\HS, \|\bcdot\|_{\HS})$.  We note also that useful kernel closure properties enable to define kernels through operations such as addition, Schur product, and function composition \citep{Berlinet-Thomas-Agnan}. 
Besides, RKHS can be meaningfully characterized through integral operators, leading to an isometry with $L^{2}_{\comp}(E,\nu)$ the space of square integrable functions defined on a compact metric space $E$ with finite measure  $\nu$  \cite{Cucker-Smale-01}.  
 \subsection*{Integral kernel operators}
 Let $k : E\times E \mapsto \comp$ be {$C^{(1,1)}(E \times E)$} \AD{(i.e. one time differentiable with respect to each argument)}, Hermitian, and positive definite,  and let the map ${\cal L}_{k}: L_{\comp}^{2}(E,\nu)\mapsto L_{\comp}^{2}(E,\nu)$ be defined as: 
\begin{equation}
\label{def-L_k}
\bigl( {\cal L}_{k} f\bigr)(x)= \int_E k(x,y) f(y) \,\nu({\mathrm{d}}y).
\end{equation}
This operator, which must be understood within the composition with the continuous inclusion $i : C^{0}(E, \comp ) \hookrightarrow L_{\comp}^2(E,\nu)$, can be shown to be well defined, {positive}, compact and self-adjoint \citep{Cucker-Smale-01}. \AD{We will always assume that the space $L_{\comp}^2(E,\nu)$ is  infinite dimensional.}The range of this operator is assumed to be dense in $L_{\comp}^2(E,\nu)$ \AD{and hence infinite dimensional}. 
 From Mercer's theorem \citep{Konig86}, the feature maps $k(\bcdot,x)$ span a RKHS defined through the eigenpairs $(\mu_i , \varphi_i)_{i\in\mathbb{N}}$ of the kernel $k$ :
\begin{equation} 
\HS : = \bigl\{ f\in L_{\comp}^{2}(E,\nu), \; f=\displaystyle\sum_{i=0}^{\infty} a_i \varphi_i \; : \; \displaystyle\sum_{i} \dfrac{\, |a_i|^{2}}{\mu_{i}} < \infty\bigr\},
\end{equation}
with no null eigenvalues \AD{and an inifinite sequence of eigenvalues} since we have assumed that the kernel range  is dense in $L_{\comp}^2(E,\nu)$. The rank of the kernel (number of -- non-zero -- eigenvalues)   corresponds to  the dimension of $\HS$, \AD{and will always be infinite in this work, which excludes the cases where $E$ is a finite set }.
 The RKHS $\HS$ is a space of smooth functions with decreasing {high frequency} coefficients. In fact, for all $u\in E$, there exists a constant $C>0$ for which $\| \partial_{u} f \|_{{L_{\comp}^2(E , \nu)}} \leq C \, \left\| f \right\|_{{\HS}}$ (Theorem \ref{ContKoopmTHEOREM} -- \ref{sec2-RKHS}), {where the derivative denotes the Gateau directional derivative of function $f$ in the direction $u$ defined as 
 \[
\forall x,u \in E, \partial_u f{(x)} = \lim_{\epsilon\to 0}  \frac{1}{\epsilon} \bigl(f(x + \epsilon u) -f(x)\bigr).
 \]}
 In order to properly define the Gateaux derivative,  a sufficient condition is to embed $E$ with a local vector space structure, which is for instance the case of differentiable manifolds. To enable such a differentiability property of the RKHS functions, we will in particular assume a weaker condition, namely \AM{that $E$ is the closure of its non empty interior \cite{Zhou-2008}.}
 We may also define uniquely a square-root symmetric isometric bijective operator ${\cal L}_{k}^{\alf}$ between $L_{\comp}^2(E,\nu)$ and $\HS$ (\ref{sec3-RKHS}). This operator enables to go from $L_{\comp}^2(E,\nu)$ to $\HS$ by increasing the functions regularity while its inverse lowers the function regularity by bringing them back to $L_{\comp}^2(E,\nu)$. Both operators are bounded. The injection $j : \HS \mapsto L^{2}_{\comp}(E,\nu)$ is continuous and compact (Theorem \ref{InjectionContinueCompacte} -- \ref{sec3-RKHS}) and  {$j(\HS)$ is assumed to be dense} in $L_{\comp}^2(E,\nu)$ (Prop. \ref{DensiteResultat}). These statements are precisely recalled in \ref{RKHS}.
 
 \begin{remark}\label{inf_sets}
 As mentioned above, the spaces $\mathcal{H}$ and $L_{\comp}^2(E,\nu)$ are  infinite dimensional spaces, preventing $E$ to be a finite set. Note  that the denseness and the infinite dimensionality assumptions can easily be relaxed -- see Remark  \ref{ValPropreNulles} -- \ref{sec3-RKHS}.
 \end{remark} 

\section{Dynamical systems on a RKHS family}
\label{sec-sys-dyn-rkhs}
We consider an invertible nonsingular  dynamical system $X(t)  = \Phi_t(X_0)$, defined from a continuous flow, $\Phi_t$ \AD{(meaning that, for any $X\in\Omega$, the mapping $t\mapsto\Phi_t(X)$ is continuous)}, on a compact invariant metric phase space \AD{differentiable} manifold, $\Omega$, (i.e. $\Phi^{-1}_t(\Omega)=\Omega$, $\forall t\in\real^+$) of time evolving vector functions over a spatial support $\Omega_x$. The functions $X:\real^{+}\times\Omega_x\mapsto\real^d$ with $X\in C^{p}$, $p\geq 1$, are solutions of the following $d$-dimensional differential system:
\begin{equation} 
\left\{
\begin{aligned}
\label{dyn-syst}
\partial_t X(t,\AD{\bcdot}) &=\Mop \bigl(\XX(t,\AD{\bcdot})\bigr), \text{ with } X(t,\AD{\bcdot}) \in \Omega,\forall t> 0,\\
\XX(0,\AD{\bcdot})&= \XX_0(\AD{\bcdot}).
\end{aligned}\right.
\end{equation}
The {nonlinear} differential operator  $\Mop{: \Omega\to \Omega}$ is assumed $C^1{(\Omega)}$, and in particular {its linear tangent expression defined as the Gateaux derivative:
\[
\partial_X \Mop(\XX)\delta \XX = \lim_{\beta\to 0} \frac{1}{\beta}\biggl(\Mop \bigl(\XX(t,\AD{\bcdot})+ \beta{\delta}\XX(t,\AD{\bcdot})\bigr) -\Mop \bigl(\XX(t,\AD{\bcdot})\bigr)\biggr)
\]
is such that} $\sup_{\XX\in\Omega} \partial_X \Mop(\XX) < \infty$ (since $\Omega$ is compact).   \AD{Let us note that such a differentiablity assumption is quite common in geophysical flows, as otherwise so-called 4DVar variational assimilation strategies \cite{Ledimet86} derived from optimal control theory \cite{Lions71} and that are routinely used in weather and climate centers would make no sense.} In the following, we consider the measure space $(\Omega_x , \mathrm{d})$ where $\mathrm{d}$ is the Lebesgue measure on $\Omega_x$. We note   $L^2(\Omega_x, \real^{d}) := \{ f = (f_1 , .. , f_d) : \Omega_x \mapsto \real^{d} \;  :  \; \text{for all} \; 1 \leq i \leq d \quad  f_i \in L^{2}(\Omega_x , \real) \}$ the vector space of square integrable functions on $\Omega_x$. We note $\|  \bcdot \|_{{L^2}}$  the norm associated with  $L^2(\Omega_x, \real^{d})$ and given by  $\|  f \|^{2}_{{L^2}} := \sum_{i=1}^{d} 
 \|f_i\|^{2}_{L^{2}(\Omega_x , \real)}$ for all $f \in L^2(\Omega_x , \real^{d})$. The set $\Omega$ is included in $L^{2}(\Omega_x , \real^{d})$. The system \eqref{dyn-syst} is assumed to admit a finite invariant measure $\nu$ on $\Omega$ (note that from invertibility the measure is also nonsingular with $\nu(\Phi^{-1}_t)(A) =0$, $\forall A{\subset} \Omega$ such that $\nu(A) =0$).
The system's observables are square integrable measurable complex functions with respect to measure $\nu$. They belong to the Hilbert space $L^{2}_{\comp}(\Omega , \nu)$   with the inner product $\langle \bcdot \, , \, \bcdot \rangle_{{L^{2}_{\mathbb{C}}(\Omega , \nu)}}$ given for $f$ and $g \in L^{2}_{\mathbb{C}}(\Omega , \nu)$ by 
\[
\langle f \, , \, g \rangle_{{L^{2}_{\mathbb{C}}(\Omega , \nu)}} := \int_{\Omega} f(y) \, \overline{\,g(y)\,} \; \nu(\mathrm{d}y). 
\]
 Depending on the context $X(t,\AD{\bcdot})$ or $X_t$ will denote either an element of $\Omega$ or the function $\AD{X(t,\bcdot)} \colon \Omega_x \mapsto \mathbb{R}^{d}$ such that $x \mapsto X(t,x)$.

 In this work, \AD{the set of different initial conditions $\Omega_0$ is an infinite compact subset of $L^2(\Omega_x, \mathbb{R}^{d})$}. For all time $t\geq0$, we denote by {$\Omega_t \subset\Omega \subset L^2(\Omega_x, \mathbb{R}^{d})$} the space defined as  $\Omega_t:= \Phi_t(\Omega_0)$. Furthermore, the set $\Omega_0$ of initial conditions will be assumed \AM{to be given -- as it is usual in ensemble data assimilation or forecasting applications --} and rich enough so that $\bigcup_{t\geq 0}( \Omega_t) =\Omega$. Hence, by this, each point of the manifold, $\Omega$, is  assumed to be uniquely characterized by a \AM{given} initial condition and the integration of the \AM{invertible} dynamical system over a given time $t$. In other words, for any $\XX\in\Omega$, there exists a unique initial condition $\XX_0\in\Omega_0$ and a unique time $t\in\mathbb{R}_+$ such that $\XX=\Phi_t(\XX_0)$. \AM{All the points of the phase space are tagged with a reaching time. A stationary point belongs to all sets $\Omega_t$ and a recurrent point belongs also to several sets.} All the sets $\Omega_t$ will be assumed to be \AM{compact subset of a functions vector space, defined as the closure of their non-empty interior. Compactness enables the use of a classical Mercer theorem while non-empty interior allows differentiability of the associated RKHS feature maps. Compactness could be relaxed with generalized form of the Mercer theorem \cite{Carmelli-et-al-2006, Steinward-Scovel-2012}. This is however outside of the scope of this paper and we will remain in the case of compactness assumption. Let us note that an assumption of finite dimensional vector spaces for the set $\Omega_t$ could have been done (which comes immediately by the Riesz theorem, if we additionally impose a vectorial space structure together with compactness and non-empty interior). While many geophysical systems are conjectured to have finite-dimensional attractors (like the 2D Navier-Stokes equations \cite{ladyzhenskaya1982finite, Constantin-Fois-Temam-85}), we prefer to adhere to a more general assumption here.} 
 
 \begin{remark}\label{dyn-sys-assumption}
\AM{For measure-preserving systems, the assumption $\Omega=\bigcup_{t\geq 0}\Omega_t$ is weaker than an ergodicity assumption, which is equivalent (together with Poincar\'{e} recurrence associated with measure preserving dynamical system) to the following statement \cite{EFHN-15} (Lemma 6.19): for any (non-empty) measurable set $A\subset \Omega$, then the set $\bigcup_{t\geq 0} \Phi^{-1}_t(A)=\Omega$.  Ergodicity is very commonly assumed in the literature on Koopman spectral analysis \cite{Brunton-et-al-22}. We have here a weaker assumption.}
 \end{remark}

 Defining at each time $t$, from ensemble $\Omega_t$, a positive Hermitian  kernel  $k_{\Omega_t}:\Omega_t\times\Omega_t$, there exists a unique associated RKHS $\HS_t$. In the following for the sake of concision the kernels $k_{\Omega_t}$, will be denoted by $k_{t}$ to refer to the dependence on the set $\Omega_t$. For all $t\geq 0$, we will use the notation $\XX_t  =\Phi_t(\XX_0)$ and $\left( \HS_t , \langle \, \bcdot \, , \, \bcdot \, \rangle_{\HS_t}  , \| \bcdot \|_{_{\HS_t}} \right)$ for the RKHS associated with the kernel $k_t$ defined on $\Omega_t \times \Omega_t$. The kernels are assumed to be $C^{\AD{(1,1)}}(\Omega_t \times \Omega_t)$ and as a consequence as shown in appendix A, the associated feature maps  have derivatives in $\HS_t$. The  RKHS $\HS_t$ for all time $t\geq 0$ forms a family of Hilbert spaces of complex functions, each of them  equipped with their own inner product $k_t(\YY_t,\XX_t)$, for all functions $\XX_t, \YY_t \in$ {$ \Omega_t$}. At any time, a measurable function of the system state, {usually} referred to as an observable, {$f$, belonging to the RKHS $\HS_t$}  can be  described  \AM{as the limit of} a linear combination of the feature maps \AD{ $\{k_t(\bcdot,\XX_t), \XX_t\in\Omega_t \}$}. 
 As it will be shown, the features maps of this RKHS family can be expressed on a time evolving orthonormal systems of basis functions, connected to each other through an exponential form and given by the eigenfunctions of the infinitesimal \AD{evolution operator} of a ``Koopman-like'' operator defined on the RKHS family. The RKHS family \AD{is defined by  \AM{${\mathcal W}= (\HS_t)_{t\geq 0}$}}. 
 In the following, we shall present a summary of the mathematical results associated to the Koopman operator in \AD{the RKHS family} .
 
 Numerically, in practice, the setup is based on $n$ realizations (called ensemble members) of this dynamical system, $\{\XX^{(i)}_t,i=1,\ldots,n \}$, generated from a finite set of different initial conditions $\{\XX^{(i)}_0,i=1,\ldots,n \}\subset\Omega_0$, are available up to time $T$.  Still, we underline that, in the following development,  the time horizon can be infinite, the sets $\Omega_t$ are infinite and the corresponding RKHS $\mathcal{H}_t$ are infinite dimensional.  This setting (both practically and theoretically) is quite common for ensemble methods applied to geophysical systems. The ensemble size is small in general, while the phase space is in theory infinite dimensional  (or at least of very high dimension). \AM{In the same way as it is usual in data assimilation or climate/weather forecasting, the dynamics on which we will experiment our setting will be assume to be invertible at least on small time range. Although often derived from Euler equation, geophysical dynamics introduce in one way or the other some viscosity and are not stricto-sensu invertible. Additionally, for most of these systems, only local-in-time solutions have been demonstrated thus far. As demonstrated in the numerical section, the proposed framework yields good results for a simple geophysical PDEs system within a short-time range. As already mentioned in the introduction, the primary goal of this work is to establish a numerical framework for ensemble-based estimation and simulation.  
We provide below some theoretical elements for analyzing this framework. These mathematical statements rely in particular 
on the assumptions of an infinite number of members in the ensemble and on invertible dynamical systems, which is obviously not the case in practice. These theoretical statements will be refined with more realistic assumptions in subsequent works.}

\section{Koopman operator in \AD{the RKHS family} }
\label{Koopman-W}
{So far{,} we did not give any precise definition of the kernels associated to the RKHS family $\HS_t$ {yet}. These kernels are defined {from} a known \textit{a priori} initial kernel, $k_0: \Omega_0\times\Omega_0$, as:
\begin{definition}[$\HS_t$ kernel]
\label{def-kt-direct}
The kernel $k_t$ associated to the RKHS $\HS_t$ are defined as
     \begin{equation}
    \forall X_t, Y_t\in \Omega_t,   \quad k_t \bigl(Y_t\,,\, X_t\bigr) =  k_0 \bigl(\Phi^{-1}_t(Y_t)\,,\, \Phi^{-1}_t(X_t)\bigr),
    \end{equation}
    where $k_0: \Omega_0\times\Omega_0$ is a given kernel.
\end{definition}
These kernels can also be equivalently defined introducing the operators, ${\cal U}_t$, operating on the feature maps. An isometric property of this operator on the RKHS family, enable{s} us to fully define the kernels along time in the same way as the previous definition. The operator ${\cal U}_t: \HS_0\mapsto \HS_t$ defined such that 
\begin{equation}
\label{def-kt}
    {\cal U}_t k_0 (\bcdot, X_0) = k_t \bigl(\bcdot, \Phi_t(X_0)\bigr),
\end{equation}
transports the kernel feature maps on the RKHS family by composition with the system's dynamics. This operator, and more specifically \AD{ an associated infinitesimal evolution operator}, will enable us to define the feature maps of $\HS_t$ from an initial {set of} feature maps on $\HS_0$. As will be fully detailed in section \ref{Proof-TH}, we will see that the operator ${\cal U}_t$  is indeed directly related to the restriction on $\HS_t$ of the {adjoint of the} Koopman operator $U_t$ on a bigger RKHS space, $\HS$, associated to a fixed kernel  defined on the whole phase space, $\Omega$.  {As  ${\cal U}_t$ propagates forward the second argument of the feature maps, it} is referred in the following {to} as the Koopman kernel operator in \AD{the RKHS family}.} Indeed, it will be pointed out that, for any $f\in\HS$, and any $X_0\in\Omega_0$, 
\begin{align*}
U_t f(X_0)=f(X_t) &= \bigl\langle R_t f\,,\, k_t(\bcdot\,,\,X_t) \bigr\rangle_{\HS_t}\\ 
&= \bigl\langle R_t f\,,\,{\cal U}_t k_0(\bcdot\,,\,X_0) \bigr\rangle_{\HS_t},
\end{align*}
where $U_t$ denotes the Koopman operator operating on $L^{2}_{\mathbb{C}}(\Omega , \nu)${,} and $R_t f$ denotes the restriction of $f$ to $\Omega_t\subset\Omega$. This expression exhibits a kernel expression of the Koopman operator definition and formally{,} at this point{, the} operator ${\cal U}_t$ can {be} thought as a kernel expression of the Koopman operator.

{The global kernel $k$ (respectively the associated RKHS) is tightly bound to the time evolving kernels $k_t$ (respectively $(\HS_t)_{t\geq 0}$).  
The restriction on $\HS$ of the Koopman operator $U_t$ {and} its adjoint the Perron-Frobenius $P_t$ {exhibit} some remarkable properties. As classically, the operators $U_t$  and  $P_t$ are unitary in $L^{2}_{\comp}(\Omega , \nu)$ (Prop. \ref{KoopmanUnitaire}), but they have the striking property to be uniformly continuous in $L^{2}_{\comp}(\Omega , \nu)$ (i.e. with bounded generators -- Theorem \ref{GenerateurKoopBorne}). As such, they can be {expanded} in an uniformly converging exponential series. 
Nevertheless, it must be outlined that the fixed kernel $k(x,y)$ and consequently $\HS$ are in practice only partially accessible as they are defined on the whole manifold of the dynamics and such expansion cannot be directly used. A local representation of the RKHS family  ${\mathcal W}$ is on the other hand much easier to infer in practice through the time evolution of ensembles $\Omega_t$ and kernels $k_t$. As we shall see{,} the operator ${\cal U}_t$ {on} ${\mathcal W}$ inherits {many} properties {from} $U_t$ and in particular, a related form of exponential series expansion (Theorem\;\ref{W-spectral-representation}).}

The Koopman {kernel} operator in \AD{the RKHS family}, {$ {\cal U}_t$}, defines an isometry from $\HS_0$ to $\HS_t$ (Theorem \ref{KoopVarieteIsometry})
\begin{align}
\label{isometry-kt}
\biggl\langle k_t\bigl(\bcdot,\!\Phi_t(\YY_0)\bigr),k_t\bigl(\bcdot,\!\Phi_t(\XX_0)\bigr) \biggr\rangle_{\!\!\HS_t}\!\!&=\! \biggl\langle k_0\bigl(\bcdot,\!\YY_0\bigr), k_0\bigl(\bcdot,\!\XX_0\bigr) \biggr\rangle_{\!\!\HS_0},\\
k_t\bigl( \Phi_t(\XX_0) , \Phi_t(\YY_0)\bigr) \!&=\! k_0(\XX_0,\YY_0).\nonumber
\end{align}
This isometry ensues obviously directly from definition \ref{def-kt-direct}. But, it can be also guessed from definition \eqref{def-kt} and the unitarity  of ${\cal U}_t$  (Theorem\;\ref{KoopVarieteIsometry}, Prop.\ref{Koopman-Perron-Frobenius-duallity}), inherited from the unitarity in $L^2_{\comp}(\Omega,\nu)$ of the Koopman operator and of its adjoint, the Perron-Frobenious operator.  This property is of major practical interest as it allows us to define the kernels of the  RKHS family from a given initial kernel fixed by the user. The kernels remain constant along the system trajectories. Alternatively, an explicit form of the feature maps can be obtained from an adjoint transport equation associated to  \AD{ an evolution operator in the RKHS family}.  Nevertheless, the isometry is far more straightforward to use to set the kernel evolution. This isometry will reveal also very useful for the  data assimilation and trajectories reconstruction problems investigated {in} section \ref{Sec-Numerical-Results}.
Strikingly, we have even more than this kernel isometry. An evolution operator $A_{U\!,\,t}: \HS_t  \to  L_{\comp}^2(\Omega_t,\nu)$, associated to the infinitesimal generator of Koopman operator $U_t$ can also be defined as
\begin{equation}
    A_{U\!,\,t}\;k_t(\bcdot,X_t) := \partial_{\Mop(\bcdot)}\, k_t (\bcdot, \XX_t),
    \label{eq:Au}
\end{equation}
where $\partial_u k_t(\bcdot, \XX_t)$ stands for the {Gateaux} directional derivative along $u\in \Omega$ of function $k(\bcdot,X_t)$.
This operator that will be shown to be bounded (Prop.\ref{A-U-t-Continuity}) and skew-symmetric (Prop.\ref{AntisymAUt}) for the inner product of $L^{2}_{\comp}(\Omega , \nu)$ plays the role of an infinitesimal generator on ${\cal W}$ and enables us expressing an exponential expansion of ${\cal U}_t$.
\begin{theorem}[\AD{The RKHS family} spectral representation]\label{W-spectral-representation}
\AD{Let $\Omega$ be a compact metric differentiable manifold that is invariant by an invertible  continuous flow, $\Phi_t$, defined from a dynamical system \eqref{dyn-syst} admitting  an invariant finite measure $\nu$ on $\Omega$. Let $\Omega_0\subset\Omega$ be an infinite compact subset of initial conditions and define $\Omega_t=\Phi_t(\Omega_0)$, with the assumption that $\Omega=\bigcup_{t\geq 0}\Omega_t$. Consider  a $C^{(1,1)}$ initial kernel $k_0:\Omega_0 \times \Omega_0 \rightarrow \mathbb{C}$ that uniquely defines an infinite dimensional initial RKHS $\mathcal{H}_0$, and, the time dependent kernels $k_t$ defined in (4) together with its associated RKHS family ${\mathcal W}= (\HS_t)_{t\geq 0}$. }Then, the evolution operator  $A_{U\!,\,t}: \HS_t \mapsto L_{\comp}^2(\Omega_t,\nu)$ given in \eqref{eq:Au}, and which \AD{is defined from the restriction of the infinitesimal generator $A_{U}$ of the Koopman operator in $L_{\comp}^2(\Omega,\nu)$},  can be diagonalized, at any time $t\geq 0$, by an orthonormal basis $(\psi^{t}_\ell)_\ell$ of $\HS_t$ such that 
\[
A_{U\!,\,t}\psi^{\,t}_\ell( X_t) = \lambda_\ell {{j\circ}} \psi^{\,t}_\ell( X_t),
\]
where $j$ is the injection $j : \HS_t \mapsto L^{2}_{\comp}(\Omega_t,\nu)$. We have in addition the following relation between the orthonormal basis systems along time:
\begin{equation}
\label{RelationVectpropreTemps}
\forall t\geq 0\;\;\psi^{\,t}_\ell( X_t) = \exp (t\; \lambda_\ell) \psi^{0}_\ell( X_0),
\end{equation}
with $X_t = \Phi_t (\XX_0)$. Furthermore, the purely imaginary eigenvalues $(\lambda_\ell)_\ell$ { do not depend on time}.
\end{theorem}
This theorem, which constitutes our main result, provides us a time-evolving system of orthonormal bases of the  RKHS family. It brings us the capability to express any observable of the system in terms of  bases that are intrinsically  linked to the dynamics and  related to each other by an exponential relation. The eigenvalues of $A_{U\!,\,t}$ are purely imaginary since this operator is skew-symmetric in $L^{2}_{\comp}(\Omega_t , \nu)$. Remarkably, the eigenvalues of each $A_{U\!,\,t}$  do not depend on time and are connected with the same covariant eigenfunctions (in the sense of \eqref{RelationVectpropreTemps}). These eigenfunctions correspond to restrictions of eigenfunctions of the infinitesimal generator of Koopman/Perron-Frobenius operators defined on $\HS$. The proofs of these results are fully detailed in the next section. 

{The frozen in time spectrum, also referred to as  iso-spectral property in the litterature, connects directly \AD{the RKHS family} representation with Lax pair theory of integrable system \cite{Lax-68}. Here the Lax pair is provided by the operator $A_{U\!,\,t}$ at a given time and the evolution equation \eqref{RelationVectpropreTemps}. The expression of the dynamical system in \AD{the RKHS family} (i.e. through the feature maps) provides a compatility condition for the Lax pair. As the existence of a Lax pair is directly related to integrable systems, it means that the \AD{the RKHS family} representation formally retains only the integrable part of a (non-necessary integrable) dynamical system. This relation between Koopman operator and integrable systems has been already put forward in several works \cite{Brunton-et-al-22, Li-PAMS-05}. Here an ensemble characterization of such a relation is exhibited.} 

The { kernel isometry \eqref{isometry-kt} (Theorem \ref{KoopVarieteIsometry})} and the Koopman spectral representation in \AD{the RKHS family} (Theorem \ref{W-spectral-representation}) constitute fundamental results enabling us to build very simple ensemble-based trajectory reconstructions for new initial conditions without the requirement of resimulating the dynamical system. Amazingly, the family of RKHS together with the Koopman isometry allows to define a system's trajectory as a constant-in-time linear combination of the time varying RKHS feature maps. Several of such reconstruction techniques, based on this fully justified superposition principle, will be shown in  section \ref{Sec-Numerical-Results}. We will in particular demonstrate the potential of this framework with a simple and very cheap data assimilation technique in the context  of very scarce data in space and time. In the next section we demonstrate  the different properties related to \AD{the RKHS family}.  
\subsection{Proofs on the properties of the Koopman operator in \AD{the RKHS family} and of the \AD{the RKHS family} spectral theorem}\label{Proof-TH}
As explained in the previous section the RKHS family, $\bigl(\HS_t\bigr)_{t\geq 0}$, does not have a good topology to work with. We need first to start defining a  ``big''  set with a better topology and which encompasses all the RKHS $\HS_t$. On this big encompassing set, we shall then define a Koopman operator enabling us to study properly the Koopman operator in \AD{the RKHS family}. 

\subsubsection{Construction of the \enquote{big} encompassing set $\HS$}
The phase space $\Omega$ corresponds to the set generated by the value of the dynamical system at a given time $t$. We note hence that $\Omega_t$ is a subset of $\Omega$. Each point of $\Omega$ designates a phase-space point $\XX=\Phi_t(\XX_0)$ uniquely defined from time $t$ and initial condition $\XX_0\in\Omega_0$. In order to define the RKHS $\HS$, let us 
 define from the kernel $k_t:\Omega_t\times\Omega_t$ a  symmetric positive definite  map $k : \Omega \times \Omega$. 
\begin{definition}[$\HS$ kernel]
\label{def_global_kernel}
    For all $X=\Phi_r(\XX_0),\, Y= \Phi_s(\YY_0) \in \Omega$, with $\XX_0, \, \YY_0\in \Omega_0$  we define
    \begin{equation}
\label{def-k-1}
k(X\,,\, Y) = k_0\bigl(X_0\,,\,  Y_0\bigr)\ell(r,s) =
        k_{t}\bigl(\Phi_{t}(X_0) \, , \, \Phi_{t}(Y_0)\bigr)\ell(r,s)\quad \forall t\geq 0,
\end{equation}
where $\ell:\mathbb{R}_+\times\mathbb{R}_+\to \mathbb{R}$ is a symmetric kernel defined, for all $r,s\in\mathbb{R}{^+}$ by
\begin{equation}
\label{def_time_kernel}
\ell(r,s)=\varphi(r-s),
\end{equation}
where $\varphi$ is a twice-\AD{differentiable} even \AM{positive definite} function such that $\varphi(0)=1$.
\end{definition}
The positivity and symmetry of kernel $k$ ensues from the property of { kernels $k_0$ and $\ell$, which are assumed to be valid kernels. Kernel $k$ inherits the regularity conditions of $k_0$ and $\ell$ and is $C^{\AD{(1,1)}}(\Omega\times\Omega)$ as well.  In the trivial case where $\varphi=1$, then comparing any pair of points on two trajectories would be equivalent to compare the initial conditions, which would result in a quite poor kernel{, and degeneracy issues}. In order to enrich the kernel structure, one can think of $\varphi$ as a regularized Dirac distribution, or a time Gaussian distribution, that will discriminate the points of the phase space that are reached at different times.}

By the Moore-Aronszajn theorem, there exists a unique RKHS called $\left( \HS , \langle \, \bcdot \, , \, \bcdot \, \rangle_{\HS}  , \| \bcdot \|_{_{\HS}} \right)$ with kernel $k$. We note that in practice the full knowledge of the phase-space {is} completely unreachable. Again, we therefore stress the fact that the setting of this encompassing RKHS $\HS$ has only a theoretical purpose. 
The RKHS $\HS$ can be connected to each RKHS of the family  ${\mathcal W}$ through extension and restriction operators denoted $E_{t}$ and $R_{t}$ respectively, and defined as follows. For all $t\geq 0$, let 
\begin{equation}\label{ExtensionDefRKHS}
 \begin{array}{ccccc}
E_{t}  : \; &  & \HS_t & \to & \HS \\
 & & k_t(\, \bcdot\, , \XX_t) \; & \mapsto & \; k(\, \bcdot \, , \, \XX_t).
\end{array}
\end{equation}
\AM{As shown in  \ref{EXT-REST} this operator is an isometry (hence continuous), and we} extend this definition by linearity on ${\rm Span}\{ k_t(\,\bcdot\, , \XX_t) \, : \, \XX_t\in \Omega_t \}$. Then, by density the function $E_t(f)$ is defined for all $f\in \HS_t$. The restriction,  {
\begin{equation}\label{RestrictDefRKHS}
 \begin{array}{ccccc}
R_{t}  : \; &  & \HS & \to &  \HS_t \\
 & & k(\, \bcdot\, ,\, \XX )=k\bigl(\, \bcdot\, ,\, \Phi_r(X_0) \bigr) \; & \mapsto & \;  k(\,\bcdot\, ,\, \XX)\restriction_{_{\Omega_t}}  = k_t\bigl(\, \bcdot \, , \, \Phi_t(\XX_0)\bigr)\ell(t,r),  \\
\end{array}
\end{equation}
\AM{which is also an isometry (\ref{EXT-REST}),} is defined similarly  for $g \in \HS^{\AD{\mathrm{sp}}} = {\rm Span}\{ k\bigl(\,\bcdot\, , \XX\bigr) \; : \; \XX\in \Omega \}$ } and extended by density in $\HS$ by the Moore-Aronszajn theorem \cite{Aronszajn-50}. The expression of $R_t(f)$ for $f\in \HS$ is further specified in remark \ref{ExpressionRestrictionFonctionRKHS} (\ref{EXT-REST}). The extension map is built in such a way that each RKHS $\HS_t$ of the family is included in the \enquote{big} encompassing RKHS $\HS$. 


In  \ref{EXT-REST}, we list several useful properties of the restriction and extension operators. Namely, $E_t$ and $R_t$  are both isometries (Prop.\ref{RestrProlongementIsometrie}); they form an adjoint pair (Prop.\ref{RestrictionProlongelementAdjoint}); the restriction is continuous in $L^{2}_{\comp}(\Omega_t, \nu)$ (Prop.\ref{ContRestrictionL2}). We define now the Koopman operator on the encompassing RKHS $\HS$.

\subsubsection{Koopman operators on $\HS$}
For all $t\geq 0$, we consider the  Koopman operator \AM{$U_t$} defined by
\begin{equation}
\label{def-Koppman-Gamma}
{U}_t(f)(X) := f \circ {\Phi}_t({\XX})= f\bigl(\Phi_t(\XX)\bigr),   \text{ for all } f \in \HS.
\end{equation}
Since $\HS$ is dense in $(L^{2}_{\comp}(\Omega , \nu) , \| \bcdot \|_{{L^{2}_{\comp}(\Omega , \nu)}})$ (by  proposition \ref{DensiteResultat} (\ref{sec3-RKHS}), the operator $U_t$ can be continuously extended on $L^{2}_{\comp}(\Omega , \nu)$ and to avoid notations inflation we keep noting  $U_t$ this extension. We first study $U_t :  L^{2}_{\comp}(\Omega , \nu) \mapsto L^{2}_{\comp}(\Omega , \nu) $ with the $L^{2}_{\comp}(\Omega, \nu)$ topology. The family $(U_t)_{t \geq0}$ is a strongly continuous semi-group on $(L^{2}_{\comp}(\Omega , \nu), \| \bcdot \|_{L^{2}_{\comp}(\Omega , \nu)})$ since $t\mapsto \Phi_t(\bcdot)$ is continuous on $\real^{+}$. 
As the feature maps are functions of $\HS$, it can be noticed that for all $\XX_r \in \Omega$
\begin{equation}\label{DefKoopman1}
U_t\, [ \, k(\, \bcdot \, , \, {\XX}_r) \, ] \; = \; k\bigl( \, {\Phi}_t(\bcdot) \, , \, {\XX}_r \bigr),
\end{equation}
\AD{which justifies the stability of $\mathcal{H}$ by the operator $U_t$.}
This corresponds to a natural expression of the Koopman operator for any function $g\in \HS^{\AD{\mathrm{sp}}}$, extended then by density to $\HS$. Yet another useful equivalent expression of the Koopman operator  is available for the feature maps.  {We have, for any points $\XX=\Phi_r(X_0), \YY=\Phi_s(Y_0) \in \Omega$
\begin{equation}
    U_t \left[ \, k(\, \bcdot \, , \, \XX) \, \right] (\YY) = k \big({\Phi}_t(Y)\,,\, \XX \big) =  k \big(\Phi_{t+s}(Y_0)\,,\, \Phi_r (X_0)\bigr)=k_0(Y_0,X_0)\ell(t+s,r).
\end{equation}
From the properties of the time kernel $\ell$, we get
\begin{equation}
\ell(t+s,r)=\varphi(t+s-r)=\varphi(s-(r-t))=\ell(s,r-t),
\end{equation}
which leads to 
\begin{equation}
    k \big({\Phi}_t(Y)\,,\, \XX \big)=k_0(Y_0,X_0)\ell(t+s,r)=k_0(Y_0,X_0)\ell(s,r-t)=k \big(\Phi_{s}(Y_0)\,,\, \Phi_{r-t} (X_0)\big)=k\big(Y,\Phi_t^{-1}(X)\big).
\end{equation}
and hence
\begin{equation}
U_t \left[ \, k(\, \bcdot \, , \, \XX) \, \right] (\YY) = k\big(Y,\Phi_t^{-1}(X)\big).
\end{equation}}

{As the previous equality is true for all $\YY\in \Omega$, this implies that for all $\XX \in \Omega$ 
\begin{equation}\label{DefKoopman2}
U_t\, [ \, k(\, \bcdot \, , \, \XX) \, ] \; = \; k\bigl(\,\bcdot \, , \, {\Phi}_t^{-1}({\XX}) \bigr).
\end{equation}
This dual formulation of the kernel expression of the Koopman operator, \AM{which further highlights the stability of $\HS$ by $U_t$} is intrinsically linked to the definition of the kernel $k$. }

This dual expression will be of central interest in the following as it enables us to formulate the time evolution of the feature maps in terms of the Koopman operator $U_t$ and its adjoint at any time $t\geq 0$. 
\begin{remark}[Transport of the kernel $k$]\label{TransportKernel}
For all $\XX=\Phi_r(X_0)\in \Omega$ and $t \geq 0$, by definition of the Koopman operator on the feature maps and \eqref{DefKoopman2}, $U_t \left[ \, k(\,\bcdot\, , \,{\XX})\right]$ has two expressions and we obtain
\[
U_t \left[ \, k(\,\bcdot\, , \,{\XX})\right] = k\bigl(\, {\Phi}_t(\bcdot) \, , \, {\XX}\,\bigr) \; = \; k\bigl(\, \bcdot \, , \, {\Phi}_t^{-1}({\XX}) \bigr) \; . 
\]
\end{remark}

The next remark provides a useful commutation property {between $U_t$ and the kernel integral operator $\mathcal{L}_{k}$ defined in equation~\eqref{def-L_k} or of its unique symmetric square-root $\mathcal{L}_{k}^{\alf}$ defined from the square-root of the kernel eigenvalues (see \ref{sec3-RKHS} for a precise definition in the general case). Note that in the case of $\HS$, the kernel integral operator is indeed a complex object which hides a time dependency.} 

\begin{remark}[Commutation between $\mathcal{L}_k^{\alf}$ (or $\mathcal{L}_k$) and $U_t$ ]\label{CommutationL_k-1/2KoopmanU_t}
For all ${\XX}\in \Omega$, we have 
\[ \mathcal{L}_{k}^{\alf} \circ U_t \, [k(\, \bcdot \, , \, {\XX})] \, = \, U_t \circ \mathcal{L}_{k}^{\alf} \, [k(\, \bcdot \, , \, {\XX})] .\]
\end{remark}
\begin{proof}
By \eqref{DefKoopman1}, we have $\mathcal{L}_{k}^{\alf} \circ U_t \, [k(\, \bcdot \, , \, {\XX})] \, = \, \mathcal{L}_{k}^{\alf} \, [k(\, {\Phi}_t(\bcdot) \, , \, {\XX})]$. {Then we obtain,  for all ${\YY} \in \Omega$
\[
\left(\mathcal{L}_{k}^{\alf}  \circ  U_t \, [k(\, \bcdot \, , \,{\XX})]\,\right)({\YY})  = \left(\, \mathcal{L}_{k}^{\alf} \, [k(\, {\Phi}_t(\bcdot) \, , \, {\XX})]\right)(Y)=  \, U_t \, [ \mathcal{L}_{k}^{\alf} \,  k(\, \bcdot \, , \, {\XX})]({\YY}) \,,
\]}
by definition of the Koopman operator since $\mathcal{L}_{k}^{\alf} \,  k(\, \bcdot \, , \, {\XX}) \in \HS$. 
\cqfd
\end{proof}
This commutation property that ensues directly from the compositional nature of the Koopman operator allows us to write immediately the equality 
\begin{equation}
\label{Commutation-Lk}
\mathcal{L}_{k}\circ U_t \, [k(\, \bcdot \, , \, \bm{\XX})] \, = \, U_t \circ \mathcal{L}_{k} \, [k(\, \bcdot \, , \, \bm{\XX})].
\end{equation}
By linearity these properties extend to all functions of $\HS$. Let us note that  the function $\mathcal{L}_{k}^{-\alf} \, k(\, \bcdot \, \, ,{\XX})$, for all ${\XX} \in \Omega$, does not necessarily belong  to $\HS$, therefore this proof cannot be applied to $\mathcal{L}_{k}^{-\alf}$.

The next proposition shows the Koopman operator defined on RKHS $\HS$ is unitary in $L^{2}_{\comp}(\Omega , \nu)$, which is a classical property of the Koopman operator in $L^2$ for measure preserving invertible systems. \AD{The proof we give for this result relies on the transport kernel expression given in Remark \ref{TransportKernel}. } 

\begin{proposition}[Unitarity of the Koopman operator in $L^{2}_{\comp}(\Omega , \nu)$]\label{KoopmanUnitaire}
The map $U_t$ : $(L^{2}_{\comp}(\Omega , \nu) ,  \| \bcdot \|_{L^{2}_{\comp}(\Omega, \nu)} )$  $\mapsto$\\ $(L^{2}_{\comp}(\Omega , \nu) , \| \bcdot \|_{L^{2}_{\comp}(\Omega , \nu)} )$ is unitary for all $t\geq 0$.
\end{proposition}
\begin{proof}
 Since  the dynamical system   $(\Phi_t)_{t\geq0}$ is measure preserving, (i.e. $\forall A\in \Omega, \nu(A)= \nu\bigl(\Phi_t^{-1}(A)\bigr)$, with $\Phi_t^{-1}(A)$ denoting the pre-image of $\Phi_t$)  and left invariant (i.e. $\Phi_t(\Omega)=\Omega$, {with $Ran(\Phi_t)= \Omega$}), the Koopman operator $U_t$ is an isometry: {for all $f$ and $g\in L^2_\mathbb{C}(\Omega,\nu)$,} 
\[
\langle U_t(f) \, , \, U_t(g) \rangle_{{L^{2}_{\comp}(\Omega , \nu)}} = \langle f \, , \, g \rangle_{{L^{2}_{\comp}(\Omega , \nu )}} \, .
\]
Besides, it can be noticed that $U_t( \HS^{\AD{\mathrm{sp}}})$, with {$\HS^{\AD{\mathrm{sp}}} = {\rm Span} \{ \, k(\,\bcdot \, , \, \XX) \, : \, {\XX} \in \Omega \}$},  is a subset of the range of $U_t : L^{2}_{\comp}(\Omega , \nu) \mapsto L^{2}_{\comp}(\Omega , \nu) $, { $U_t( \HS^{\AD{\mathrm{sp}}})=  {\rm Span} \{ \, k\bigl(\,\bcdot \, , \, {\Phi}_t^{\AD{-1}}({\XX})\bigr) \, : \, {{\XX} \in \Omega \}}$} and we have the following inclusions
\[
 U_t(\HS^{\AD{\mathrm{sp}}}) \subset Ran(U_t) \subset L^{2}_{\comp}(\Omega , \nu)\, .
\]
We show now that $U_t(\HS^{\AD{\mathrm{sp}}}) $ is dense in $\bigl( L^{2}_{\comp}(\Omega , \nu) , \| \bcdot \|_{L^{2}_{\comp}(\Omega , \nu)}\bigr)$. Let $f\in L^{2}_{\comp}(\Omega , \nu)$ be such that { $\bigl\langle  f , k\bigl(\,\bcdot , {\Phi}_t^{\AD{-1}}({\XX})\bigr) \bigr\rangle_{L^{2}_{\comp}(\Omega , \nu)}\!=0$ for all ${\XX}\in \Omega$. Then, from the definition of $\mathcal{L}_k$, it means that  $({\cal L}_{k} f) [{\Phi}_t^{\AD{-1}}({\XX})] = 0$.  The injectivity of $\mathcal{L}_k$ ensures that $f\circ\Phi_t^{\AD{-1}}=0$, so $U_{\AD{-t}} f=0$. As $U_{\AD{-t}}$ is an isometry, it is in particular injective, which eventually leads to $f=0$. }
By a consequence of Hahn-Banach theorem this result shows that the range of $U_t$ is dense in $L^{2}_{\comp}(\Omega , \nu)$ and Proposition \ref{KoopmanUnitaire} is proved. \cqfd
\end{proof}
\noindent
Proposition \ref{KoopmanUnitaire}, shows that the Koopman operator $U_t$ is invertible and that its inverse in $L^{2}_{\comp}(\Omega , \nu )$ is $U^{*}_t$. Denoting by $P_t$ the operator defined by $P_t:=U^{*}_t$, for all $f$ and $g \in L^{2}_{\comp}(\Omega, \nu)$ we have $\langle \, U_t(f) \, , \, g \rangle_{{L^{2}_{\comp}(\Omega, \nu)}} = \langle f \, , \, P_t(g) \,  \rangle_{{L^{2}_{\comp}(\Omega, \nu)}}$. 
The family $(P_t)_t$ is a strongly continuous  semi-group on $\bigl(L^{2}_{\comp}(\Omega , \nu), \|\bcdot\|_{L^{2}_{\comp}(\Omega, \nu)}\bigr)$. This operator is referred to as the Perron-Frobenius operator. 
For all $t\geq 0$, the Perron-Frobenius operator $P_t$ verifies {for all $\XX \in \Omega$
\begin{equation}\label{PerronFrobeniusKernel}
P_t [ k\bigl( \Phi_t(\bcdot)  ,  \XX)\bigr ] = k(\, \bcdot \, , \XX) .
\end{equation}
From Remark 1, we can write a more explicit expression of $P_t$ on the featutre maps: for all $X\in\Omega$,
\begin{equation}\label{PerronFrobeniusRightflow}
P_t\big[k(\bcdot,X)\big]=k\big(\bcdot,\Phi_t(X)\big).
\end{equation}}
{Informally, if we see the function $k(\bcdot,X)$ as an atom of the measure, then its expression at a future time is provided  by \eqref{PerronFrobeniusRightflow}, which corresponds well to the idea that the   Perron-Frobenius operator advances in time the density.}
\noindent

As previously stated the Koopman operator $U_t$ is an isometry in $L^{2}_{\comp}(\Omega , \nu)$. The next result asserts that $U_t$ is also an isometry in $\HS$.  
\begin{proposition}[Isometric relation of the kernel]\label{RelationKernel}
{For all $\XX$ and $\YY \in \Omega$, we have $k\big({\Phi}_t(\XX) \, , {\Phi}_t(\YY)\big) = k( X\, , \, Y)$.}
\end{proposition}
\begin{proof}
{This results follows immediately from the kernel definition: for all $X=\Phi_r(X_0),Y=\Phi_s(Y_0)\in\Omega$, we have \begin{equation}
k\big({\Phi}_t(\XX) \, , {\Phi}_t(\YY)\big)= k\big(\Phi_{t+r}(X_0), \Phi_{t+s}(Y_0)\big) = k_0(X_0,Y_0)\ell(t+r,t+s) = k_0(X_0,Y_0)\ell(r,s)=k(X,Y).
\end{equation} } 
\end{proof}
Remark \ref{TransportKernel} on the application of the operator $U_t$ to  $k(\, \bcdot \, , \, \XX)$, with $\XX\in \Omega$, {shows that applying the flow, $\Phi_t$, on the first variable of $k(\, \bcdot \, , \, \bcdot \,)$. is equivalent to applying $\Phi_t^{-1}$ to the second variable and vice-versa.  Consequently, for all $f\in L^2_\mathbb{C}(\Omega,\nu)$ and $\XX \in \Omega$, we have $\langle f \, , \, k\bigl(\, \bcdot \, , \, \Phi_t(\XX)\bigr) \,  \rangle_{{L^{2}_{\comp}(\Omega, \nu)}} \; = \langle f \, , \, k\bigl(\, \Phi_t^{-1}(\bcdot) \, , \, \XX\bigr) \,  \rangle_{{L^{2}_{\comp}(\Omega, \nu)}}   \; = \; \langle f \; , \; P_t \,k(\, \bcdot \, , \, \XX) \rangle_{{L^{2}_{\comp}(\Omega, \nu)}}$ and by definition of the adjoint, we obtain 
\begin{equation}\label{PerronFrobeniusFlotEg}
\bigl\langle f \, , \, k\bigl(\, \bcdot \, , \, \Phi_t(\XX)\bigr) \,  \bigr\rangle_{{L^{2}_{\comp}(\Omega, \nu)}} = \; \bigl\langle U_t(f) \; , \;  \,k(\, \bcdot \, , \, \XX) \bigr\rangle_{{L^{2}_{\comp}(\Omega, \nu)}}.
\end{equation}}
{
We already knew that this equality was right for $f\in\mathcal{H}$ and for the inner product in $\mathcal{H}$, namely
\begin{equation}
\big\langle f,k\big(\bcdot,\Phi_t(X)\big)\big\rangle_\mathcal{H}=f\big(\Phi_t(X)\big)=U_t f(X)=\big\langle U_t(f),k(\bcdot,X)\big\rangle_\mathcal{H}.
\end{equation} 
Equation (\ref{PerronFrobeniusFlotEg}) provides a weak (in the sense that the $L^2_\mathbb{C}(\Omega,\nu)$-inner product against a feature map is no longer the evaluation function)  formulation of the transport of any observable in $L^2_\mathbb{C}(\Omega,\nu)$ by the flow.}

In order to further exhibit  several analytical results on the Koopman operator $(U_t)_{t\geq0}$ in $L^{2}_{\comp}(\Omega , \nu)$, we introduce in the following its infinitesimal generator.

\subsubsection{Koopman infinitesimal generator}\label{K-Generator}
We will note as $\bigl(A_{_U},  \mathcal{D}(A_{_U})\bigr)$ the infinitesimal generator  of the strongly continuous semigroup $(U_t)_{t\geq 0}$ on $L^{2}_{\comp}(\Omega , \nu)$ and its domain.
As the Koopman and Perron-Frobenius operators are adjoint  in $L^{2}_{\comp}(\Omega, \nu)$ their infinitesimal generators are also adjoint of each other {with possibly their own domain}. The following lemma characterizes first the Perron-Frobenius infinitesimal generator { $A_{_P}$ and its domain $\mathcal{D}\left(A_{_P}\right)$}.
\begin{lemma}[Perron-Frobenius infinitesimal generator]
The Perron-Frobenius infinitesimal generator is the un\-boun\-ded operator, $\left( A_{_P} , \mathcal{D}\left(A_{_P}\right)\right)$ defined by 
{\[
\mathcal{D}\left(A_{_P}\right) = \left\{ f\in L^{2}_{\comp}(\Omega , \nu) \; : \; x\mapsto \partial_{\Mop(x)} f(x)\; \in L^{2}_{\comp}(\Omega, \nu) \right\} \; \text{and }\; A_{_P} f := \, - \partial_{\Mop(\bcdot)}\, f(\bcdot) \quad \text{for} \quad  f\in \mathcal{D}\left(A_{_P}\right),
\]}
{where $\Mop$ denotes the  differential \AD{operator} of the system dynamics \eqref{dyn-syst} and $\partial_u f$ stands for the directional derivative of $f$ along $u\in \Omega$.}
\end{lemma}
\begin{proof} See \ref{Ap-lemma1}. \end{proof}
For all $t\geq 0$, $P_{t}$ is the adjoint of $U_{t}$ in $L^{2}_{\comp}(\Omega, \nu)$. The Koopman infinitesimal generator in $L^{2}_{\comp}(\Omega , \nu)$ is consequently given by 
\begin{equation}\label{KoopPFL1}
\mathcal{D}(A_{_U}) = \mathcal{D}(A^{*}_{_P}) \quad \text{and} \quad  A_{_U} = A_{_P}^{*} .
\end{equation}
The Koopman and Frobenius-Perron operators are unitary in $L^{2}_{\comp}(\Omega, \nu)$, as $\mathcal{D}\left(A_{_U}\right)$ is dense in $L^{2}_{\comp}(\Omega ,\nu)$, by Stone's theorem, their infinitesimal generators are therefore skew symmetric for $\langle\,\bcdot\,,\,\bcdot\,\rangle_{{L^{2}_{\comp}(\Omega, \nu)}}$ and we have
\begin{equation}\label{KoopPFL2}
 \mathcal{D}\left(A^{*}_{_{P}}\right) = \mathcal{D}\left(A_{_{P}}\right)  \quad \text{and} \quad A_{_{P}}^* =  - A_{_{P}},
\end{equation} 
\begin{equation}\label{KoopPFL3}
 \mathcal{D}\left(A^{*}_{_{U}}\right) = \mathcal{D}\left(A_{_{U}}\right)  \quad \text{and} \quad A_{_{U}}^* =  - A_{_{U}}. 
\end{equation} 
The adjoint has to  be understood in the topology of $L^{2}_{\comp}(\Omega, \nu)$. The two next propositions {-- well known in $L^{2}_{\comp}(\Omega, \nu)$ for invertible measure preserving systems --}, summarize \eqref{KoopPFL1}, \eqref{KoopPFL2} and \eqref{KoopPFL3}.  

\begin{proposition}[Koopman infinitesimal generator]\label{GenerateurKoopman}
The Koopman infinitesimal generator of $(U_{t})_{t}$ in $L^{2}_{\comp}(\Omega, \nu)$ is the unbounded operator, {$\bigl(A_{_U}, \mathcal{D}(A_{_U})\bigr)$,} defined by 
{\[
\mathcal{D}\left(A_{_U}\right) \, = \, \left\{ f\in L^{2}_{\comp}(\Omega, \nu ) \; : \; x\mapsto \partial_{\Mop(x)} f(x)\; \in L^{2}_{\comp}(\Omega , \nu)  \, \right\} \; \text{and }\; A_{_U} f= \partial_{\Mop(\bcdot)} f(\bcdot)  \quad \text{for} \quad f\in \mathcal{D}\left(A_{_U}\right).
\]}
\end{proposition}

\begin{proposition}[Skew symmetry of the generators]\label{KoopAntisym}
The Koopman infinitesimal generator $A_U$ is skew-symmetric in $L^{2}_{\comp}(\Omega, \nu)$ . 
\end{proposition}
\noindent
We set now a continuity property of the Koopman infinitesimal generator on a subspace of $\mathcal{D}(A_{_U})$. Let $f\in \HS$, by theorem \ref{RKHSC1} (\ref{sec2-RKHS}), ${x\mapsto \partial_{\Mop(x)} f(x)} \in L^{2}_{\comp}(\Omega , \nu)$, and through Cauchy-Schwarz inequality we get that $f\in \mathcal{D}(A_{_U})$ since the real valued function $\Mop$ belongs to $L^{2}_{\comp}(\Omega, \nu)$ and consequently $\HS \subset \mathcal{D}(A_{_U})$.
\begin{theorem}[Continuity of the Koopman  generator on $\HS$]\label{GenerateurKoopBorne}
The restriction of the Koopman infinitesimal generator 
\[
A_{_U} : \left( \HS \, , \, \| \bcdot \|_{_\HS} \right) \mapsto \left( L^{2}_{\comp}(\Omega, \nu) \, , \, \| \bcdot \|_{L^{2}_{\comp}(\Omega, \nu)} \right)
\]
is continuous. 
\end{theorem}
\begin{proof}
For all $f\in \HS$ {using} the Cauchy-Schwarz inequality {as above} and theorem \ref{ContKoopmTHEOREM}  (\ref{sec2-RKHS}) we have
\[
{\| A_{_U} f \|_{{L^{2}_{\comp}(\Omega , \nu)}} = \|\partial_{\Mop(\bcdot)}\, f(\bcdot)\|_{{L^{2}_{\comp}(\Omega , \nu)}} \, \leq C \; \| f\|_{_{\HS}}.}
\]
\cqfd
\end{proof}
Through the dual expression of the Koopman operator (remark \ref{TransportKernel}) the infinitesimal generator $A_{_U}$ provides a dynamical system specifying  the time evolution of the feature maps as, for all $\XX \in \Omega$, the feature map $k(\, \bcdot \, , \Phi_t(\XX))$ verifies 
\begin{equation}\label{EquationFeatureMap}
\partial_t \,U_t\,  k(\, \bcdot \, , \, \XX)  ={ \partial_t \, k\bigl(\, \Phi_t(\bcdot )\, , \, \XX\bigr)  =  A_{_U} \, k\bigl(\,\Phi_t( \bcdot) \, , \, \XX\bigr) \,} .
\end{equation}

We have the following useful differentiation formulae.
\begin{proposition}[Differentiation formulae]\label{DerivativeRule}
For all $f\in \HS$ {(with $f = E_t g,\,g\in\HS_t$)}  and $\XX_t \in \Omega_t$ we have
\[
\langle {\cal L}_{k} \circ \partial_t f \; , \; k(\,\bcdot\, , \, \XX_t) \rangle_{_\HS} = - \langle f \; , \; {\cal L}_{k} \circ \partial_t \, k(\,\bcdot\, , \, \XX_t) \rangle_{_\HS}.
\]
\end{proposition}
\begin{proof}
As ${\cal L}_{k} \circ A_{_U}f$ belongs to $\HS$, we have ${\cal L}_{k} \circ A_{_U}f(\XX_t) = {\cal L}_{k} \circ \partial_t f(\XX_t)$ and 
\[
\langle {\cal L}_{k} \circ \partial_t f \; , \; k(\,\bcdot\, , \, \XX_t) \rangle_{_\HS}  = \langle {\cal L}_{k} \circ A_{_U} f \; , \; k(\,\bcdot\, , \, \XX_t) \rangle_{_\HS} = \;\langle\AD{j\left( {\cal L}^{\alf}_{k} \circ A_{_U} f\right)} \; , \; {\cal L}^{-\alf}_{k} \, k(\,\bcdot\, , \, \XX_t) \rangle_{L^{2}_{\comp}(\Omega , \nu)}.
\]
As ${\cal L}^{\alf}_{k}$ is self-adjoint in $L^{2}_{\comp}(\Omega , \nu)$ and $A_{_U}$ is skew-symmetric in $L^{2}_{\comp}(\Omega ,\nu)$, we obtain 
\[
\langle {\cal L}_{k} \circ \partial_t f \; , \; k(\,\bcdot\, , \, \XX_t) \rangle_{_\HS}  = \langle A_{_U}f \, , \, k(\,\bcdot\, , \, \XX_t) \rangle_{L^{2}_{\comp}(\Omega , \nu)} =  - \langle f \; , \; A_{_U}\, k(\,\bcdot\, , \, \XX_t) \rangle_{L^{2}_{\comp}(\Omega , \nu)} = - \langle f \; , \; \partial_{t}\, k(\,\bcdot\, , \, \XX_t) \rangle_{L^{2}_{\comp}(\Omega , \nu)}.
\]
As $\langle f \; , \; \partial_{t}\, k(\,\bcdot\, , \, \XX_t) \rangle_{L^{2}_{\comp}(\Omega , \nu)} = \langle f \; , \; \mathcal{L}_{k} \circ \partial_t \, k(\, \bcdot \, , \XX_t) \rangle_{_\HS}$, this complete the proof.
\cqfd
\end{proof}

As already outlined, the whole phase space, $\Omega$, and the global embedding RKHS, $\HS$,  defined on it are both completely inaccessible for high dimensional state spaces.  Instead of seeking to reconstruct this global RKHS, we will work in the following with time-varying ``local'' RKHS  spaces built from a small (w.r.t. the phase space dimension) ensemble of initial conditions. The set of these time-varying spaces forms the RKHS family.  To express the time evolution of the features maps associated to the RKHS family we now define an appropriate expression of the Koopman operator on this family.

\subsubsection{Derivation of Koopman operator expression in ${\mathcal W}$} 
To fully specify the Koopman operators in \AD{the RKHS family}, we rely on the family of extension, restriction mappings $(E_t)_{t\geq0}$ and $(R_t)_{t\geq0}$ (\ref{ExtensionDefRKHS}, \ref{RestrictDefRKHS}) relating the ``big'' encompassing RKHS $\HS$ to the family of time-evolving RKHS $\HS_t$  and use the Koopman operators $U_t : L^{2}_{\comp}(\Omega, \nu) \mapsto L^{2}_{\comp}(\Omega, \nu)$ for $t\geq 0$.  {The adjoints $P_t$ will also be very helpful as well as Remark 1} {on the dual relation of the application of the flow on the global kernel $k$. }

From now on, we note $\XX_t := \Phi_t(\XX_0)$ for all $\XX_0 \in \Omega_0$ with $\XX_t$ belonging to $\Omega_t$.
We define the Koopman operator in \AD{the RKHS family} by {${\cal U}_t  := R_t  \circ P_{t} \circ  E_0$ for all $t\geq 0$. For all $X_0 \in \Omega_0$, we have 
\begin{equation}
\label{Koopvariete}
\mathcal{U}_t \left[\, k_0(\, \bcdot \, ,\, X_0)\,\right] \, = \,  R_t \, \circ  P_{t}  \left[ \, k(\,\bcdot\,,X_0) \, \right]  \, = \, R_t \, k(\,\bcdot\,,X_t) \,  \, = \, k_t(\bcdot\, , \,X_t)\, ,
\end{equation}
where the second equality is due to (\ref{PerronFrobeniusRightflow}), and the third equality holds true from the definition of $R_t$ and the fact that $\ell(t,t)=1$.}
The Koopman operator in \AD{the RKHS family},  $\mathcal{U}_t : \HS_{0} \mapsto \HS_t$, transports the  kernel feature maps on the RKHS family by composition with the system's dynamics. It  inherits some of the nice properties of the {Perron-Frobenius} operator defined on the encompassing global RKHS. As shown by the following theorem,  proposition \ref{RelationKernel} remains valid for the family of kernels $(k_t)_{t\geq 0}$ and $\mathcal{U}_t$  is still  unitary in the sense of the following theorem.
\begin{theorem}[Koopman RKHS isometry]\label{KoopVarieteIsometry}
The Koopman operator on the RKHS family defines an isometry from $\HS_0$ to $\HS_t$: for all $\XX_0$ and $\YY_0 \in \Omega_0$
\begin{equation*}
\!\!\!\bigl\langle \mathcal{U}_t \, k_0(\, \bcdot \, , \, \!\YY_0) \; , \; \mathcal{U}_t \, k_0(\, \bcdot \, , \, \!\XX_0) \bigr\rangle_{\HS_t}\!\! \; = \; \! \bigl\langle k_0( \, \bcdot \, , \, \!\YY_0) \; , \; k_0(\, \bcdot \, , \, \!\XX_0) \bigr\rangle_{\HS_0}.
\end{equation*}
The range of $\mathcal{U}_t : \HS_0 \mapsto \HS_t$ is dense in $\HS_t$.
\end{theorem}
\begin{proof}
{Upon applying equations (\ref{R_tIsometrie} of Prop. \ref{RestrProlongementIsometrie}, \ref{EXT-REST}), {which states} that $R_t:\mathcal{H}\rightarrow\mathcal{H}_t$ is an isometry, we have
\[
\bigl\langle \mathcal{U}_t \, k_0(\, \bcdot \, , \, \YY_0) \; , \; \mathcal{U}_t \, k_0(\, \bcdot \, , \, \XX_0) \bigr\rangle_{\HS_t} = \bigl\langle P_t k(\, \bcdot \, , \, \YY_0) \; , \; P_t \, k(\, \bcdot \, , \, \XX_0) \bigr\rangle_{_\HS} \; = \; \bigl\langle  k\big(\, \bcdot \, , \, \Phi_t(Y_0)\big) \; , \;  k\big(\, \bcdot \, , \, \Phi_t(X_0)\big) \bigr\rangle_{_\HS} = k(\Phi_t(X_0)\, , \, \Phi_t(Y_0)).
\]
With proposition \ref{RelationKernel} we obtain 
\[
\bigl\langle \mathcal{U}_t \, k_0(\, \bcdot \, , \, \YY_0) \; , \; \mathcal{U}_t \, k_0(\, \bcdot \, , \, \XX_0) \bigr\rangle_{\HS_t}
=k(\XX_0 \, , \, \YY_0)
= k_0(X_0,Y_0)
{=} \langle k_0(\,\bcdot\, , \, \YY_0) \, , \, k_0(\, \bcdot \, , \, \XX_0) \rangle_{\HS_0}.
\] }
Besides, the image of ${\rm Span} \{ k_0( \, \bcdot \, , \, \YY_0) \, : \, \YY_0 \in \Omega_0 \}$  under $\mathcal{U}_t$  is the set  ${\rm Span} \{ k_t( \, \bcdot \, , \, \YY_t) \, : \, \YY_t \in \Omega_t \}$ which is dense in $(\HS_t , \| \bcdot \|_{_{\HS_t}})$. The range of $\mathcal{U}_t : \HS_0 \mapsto \HS_t$ is thus dense in $\HS_t$ and the operator $\mathcal{U}_t$ is, in that sense, unitary. \cqfd
\end{proof}
\noindent
Let us now determine the adjoint of the Koopman operator in \AD{the RKHS family}. {Let $\mathcal{P}_t := \; R_0 \circ U_t \circ  E_t$ \; for all $t\geq0$ and let $\XX_t=\Phi_t(X_0) \in \Omega_t$. With the same arguments as for $\mathcal{U}_t$, we have
\begin{equation}\label{PFvariete}
\mathcal{P}_t \, \left[\,  k_t(\,\bcdot\, , X_t) \, \right] \, = \, k_0( \bcdot \, , \, X_0) \; .
\end{equation}}
The mapping $\mathcal{P}_t : \HS_t \mapsto \HS_0$ with $t\geq 0$ constitutes the Perron-Frobenius family of operators in \AD{the RKHS family}. The mapping $\mathcal{P}_t : \HS_t \mapsto \HS_0$ is unitary for the RKHS family topology (isometry from $\HS_t$ to $\HS_0$ and the range of $\mathcal{P}_t$ is dense in $\HS_0$). The next proposition justifies that $\mathcal{U}_t$ and $\mathcal{P}_t$ have inverse action on the feature maps. 
\begin{proposition}[Koopman Perron-Frobenius duality]\label{Koopman-Perron-Frobenius-duallity}
For all $\XX_0 \in \Omega_{0}$ and {$\YY_t=\Phi_t(Y_0) \in \Omega_t$}, we have 
\[
\left\langle \, \mathcal{U}_t \, k_{0}(\,\bcdot\, , \, \XX_0 ) \; , \; k_t(\, \bcdot \, , \YY_t) \, \right\rangle_{\HS_{t}} \;  = \; \left\langle \, k_{0}(\,\bcdot\, , \, \XX_{0}) \; , \; \mathcal{P}_t \left[ k_{t}(\, \bcdot \, , \YY_{t}) \right]\, \right\rangle_{\HS_{0}} \, .  
\]
\end{proposition}
\begin{proof}
By the equations \eqref{Koopvariete} and \eqref{PFvariete}, we have
\[
\left\langle \, \mathcal{U}_t \, k_{0}(\,\bcdot\, , \, \XX_0 ) \; , \; k_t(\, \bcdot \, , \YY_t) \, \right\rangle_{\HS_{t}} ={ k_t\big( \YY_t , \Phi_t(X_0)\big)} \; \text{and }\; \left\langle \, k_{0}(\,\bcdot\, , \, \XX_{0}) \; , \; \mathcal{P}_t \left[ k_{t}(\, \bcdot \, , \YY_{t}) \right]\, \right\rangle_{\HS_{0}} = k_0(\YY_0 , \XX_0)\, , 
\]
and { both terms are equal by construction because of Definition \ref{def-kt-direct}. } \cqfd
\end{proof}
In order to derive the Koopman and Perron-Frobenius operators' spectral representation in \AD{the RKHS family}, we exhibit now two propagation operators that will allow us to express the evolution of the feature maps \AD{in $\mathcal{W}$}. 
\subsubsection{\AD{The RKHS family} spectral representation} 
 We specify hereafter  a family of operator $A_{U\!,\,t}$, related to the Koopman infinitesimal generator $A_{_U}$,  and that give rise to an evolution equation on $\Omega_t$ akin to  \eqref{EquationFeatureMap}. They will play, in that sense, the role of infinitesimal generators on the RKHS family. 

For all $t\geq 0$, let $A_{U\!,\,t}$ be defined by $A_{U\!,\,t}:= R_t   \circ A_{_U} \circ E_t$ with $E_t : \HS_t \mapsto \HS$ and $R_t : L^{2}_{\comp}(\Omega , \nu) \mapsto L^{2}_{\comp}(\Omega_t , \nu)$ (\ref{RestrictDefL2} \ref{EXT-REST}). 

\begin{proposition}[Continuity of $A_{U\!,\,t}$]\label{A-U-t-Continuity}
The mapping $A_{U\!,\,t} : (\HS_t , \| \bcdot \|_{_{\HS_t}})  \mapsto (L^{2}_{\comp}(\Omega_t , \nu) , \| \bcdot \|_{{L^{2}_{\comp}(\Omega_t , \nu) }})$ is well-defined and continuous. 
\end{proposition}
\begin{proof}
Upon applying (\ref{E_tIsometrie} of Prop. \ref{RestrProlongementIsometrie}, \ref{EXT-REST}), the continuity of $A_{_U}$ (Theorem \ref{GenerateurKoopBorne}) and of $R_t$ (Proposition \ref{ContRestrictionL2}, \ref{EXT-REST}) we obtain the proof of the proposition. 
\cqfd 
\end{proof}
In a very similar way as the infinitesimal generator on $\HS$, the operator $A_{U\!,\,t}$ can be understood as an evolution equation of the feature maps defined on $\Omega_t$ and associated to $\HS_t$. For $\XX_t \in \Omega_t$, by proposition \ref{GenerateurKoopman} it can be noticed that 
\[
A_{U\!,\,t} \; k_t(\, \bcdot \, , \, \XX_t) = R_t \circ A_{_U} \left[ k(\, \bcdot \, , \, \XX_t) \right] = R_t \left[ \, {\partial_{\Mop(\bcdot)}}\, k(\, \bcdot \, , \, \XX_t)  \, \right]
\]
and in particular on $\Omega_t$
\begin{equation}\label{Evolution}
A_{U\!,\,t} \; k_t(\, \bcdot \, , \, \XX_t)  =  {\partial_{\Mop(\bcdot)}}\,  k_t ( \, \bcdot \, , \,  \XX_t).
\end{equation}
\begin{remark}\label{LienGenerateurs} Operator  $A_{U\!,\,t}$ is alike the Koopman infinitesimal generator $A_{_U}$.
As a matter of fact by proposition \ref{GenerateurKoopman} and theorem  \ref{GenerateurKoopBorne} we have for all $\XX\in \Omega$
\[
A_{_U} k(\, \bcdot \, , \, \XX) = {\partial_{\Mop(\bcdot)}}\, k (\, \bcdot \, , \,\XX),
\]
and for all $\XX_t \in \Omega_t$
\[
    A_{U\!,\,t}\;k_t( \,\bcdot \, , \, \XX_t) = {\partial_{\Mop(\bcdot)}}\, k_t ( \, \bcdot \, , \,  \XX_t).
\]
\AD{Notice that these two operators are different and act on different domains.}
\end{remark}

\begin{remark}\label{feature-maps-evolution} Operators $A_{U\!,\,t}$ can be understood as an evolution operator for the feature maps. As a matter of fact through \eqref{EquationFeatureMap} and the definition of $A_{U\!,\,t}$, we have, for all \AM{$\XX \in \Omega_{0}$}
\begin{equation}
 \partial_t \, k\bigl(\, \bcdot \, , \, \Phi_t(\XX)\bigr)_{{|_{\Omega_t}}}  =  A_{U\!,\,t} \, k_{t}\bigl(\, \bcdot \, , \, \Phi_t (\XX)\bigr) \, .   
\end{equation}
\end{remark}

Through the above remark the operator $A_{U\!,\,t}$ 
inherits the properties of operator $A_{_U}$ defined on the global encompassing RKHS $\HS$.  As shown in the next proposition it remains in particular skew-symmetric.
\begin{proposition}[Skew-symmetry of $A_{U\!,\, t}$]\label{AntisymAUt}
The operator $A_{U\!,\, t}$ is skew-symmetric in $L^{2}_{\comp}(\Omega_{t} , \nu)$: for all $f$ and $g\in \HS_t$
\[
\langle A_{U\!,\,t} \, f \; , \; \AM{j\circ g} \rangle_{{L^{2}_{\comp}(\Omega_t , \nu)}} \; = \; - \, \langle f \; , \; A_{U\!,\,t} \, g \rangle_{{L^{2}_{\comp}(\Omega_t , \nu)}} .
\]
\AM{where $j$ is the injection $j : \HS_t \mapsto L^{2}_{\comp}(\Omega_t,\nu)$.}
\end{proposition}
For the need of this proof only, we introduce the following notation. For a function $h$ defined on $\Omega_t$, we note  $\widetilde{h}$ the extension on $\Omega$ defined by 
$$
\widetilde{h} := \left\{
    \begin{array}{ll}
        h & \mbox{on } \Omega_t \\
        0 & \mbox{else.}
    \end{array}
\right.
$$
\begin{proof}
As $g=E_t(g)$ holds on $\Omega_t$, we have 
\[
\langle A_{U\!,\,t} \, f \; , \, \AM{j\circ g} \rangle_{{L^{2}_{\comp}(\Omega_t , \nu)}} \; = \; \langle A_{_U} \circ E_t (f)\restriction_{_{\Omega_t}} \, , \, E_t(g) \, \rangle_{{L^{2}_{\comp}(\Omega_t , \nu)}} \; =\; \langle \widetilde{A_{_U} \circ E_t(f)}\, , \, E_t(g) \, \rangle_{{L^{2}_{\comp}(\Omega , \nu)}}{.}
\]
By  proposition \ref{KoopAntisym}  and since $E_t(f) \restriction{_{\Omega_t}} = f$ on $\Omega_t$ we infer that 
\[
\langle A_{U\!,\,t} \, f \; , \, \AM{j\circ g} \, \rangle_{{L^{2}_{\comp}(\Omega_t , \nu)}} \; = \; - \, \langle  E_t (f) \, , \, \widetilde{A_{_U} \circ E_t(g)} \, \rangle_{{L^{2}_{\comp}(\Omega , \nu)}} \; = \; - \langle  f \, , \, A_{_U} \circ E_t(g) \, \rangle_{{L^{2}_{\comp}(\Omega_t , \nu)}} \; = \; - \,  \langle  f \; , \; A_{U\!,\,t} (g) \, \rangle_{{L^{2}_{\comp}(\Omega_t , \nu)}} . 
\]
\cqfd
\end{proof}

We are now ready to prove the \AD{RKHS family} spectral representation theorem, which states  that the bounded operator $A_{U\!,\,t} : \HS_t \mapsto L^{2}_{\comp}(\Omega_t , \nu)$ is diagonalizable for all $t\geq 0$.

\subsubsection{Proof of Theorem \ref{W-spectral-representation} [\AD{The RKHS family} spectral representation]}
The full proof, organized in two main steps,  is thoroughly detailed in \ref{Step-1-2}.  For this proof, we consider the \AD{restriction of the Koopman infinitesimal generator on $\mathcal{H}$}, $A_{_U} : \HS \mapsto L^{2}_{\comp}(\Omega , \nu)$, which is connected to each $A_{U\!,\,t}$ for $t\geq 0$ through the restriction operator $R_t$. In the first step, the diagonalization of operator $A_{_U} : \HS \mapsto L^{2}_{\comp}(\Omega , \nu)$ in $L^{2}_{\comp}(\Omega , \nu)$ is first performed. To that end we introduce  an intermediate (approximating) operator denoted $\widetilde{A}_{_U}$, directly related to $A_{_U}$ and whose inverse is shown to be compact and self-adjoint.  The second step of the proof consists in deducing the diagonalization of each $A_{U\!,\,t}$ from the diagonalization of $A_{_U}$ obtained at  step 1. These two steps are thoroughly presented in \ref{Step-1-2}.

\subsection{Tangent linear dynamics}
A  result of practical interest  concerns the establishment of a rigorous ensemble expression of the tangent linear dynamics operator. Recall that we note $\XX_t = \Phi_t(\XX_0)$ for all $t\geq 0$.  We define  $\delta \XX(t,\AD{\bcdot}) := \Phi_t[\XX_0(\AD{\bcdot}) + \delta \XX_0(\AD{\bcdot})] - \Phi_t[\XX_0(\AD{\bcdot})]$ for all $t\geq 0$, and where $\delta \XX_0(x)$ is a perturbation of the initial condition at point $x$. The function $\delta \XX(t , \bcdot) \in L^{2}(\Omega_x , \real^{d})$ is the perturbation of the flow at time $t$ with respect to the initial condition $\XX_0$. We have 
\[
\Phi_t(\XX_0 + \delta \XX_0) - \Phi_t(\XX_0)= \int_{0}^{t} \Mop[\Phi_s(\XX_0 + \delta \XX_0)] - \Mop[\Phi_s(\XX_0)] \, \mathrm{d}s + o(\|\delta \XX_0\|) = \int_{0}^{t} \dif(\Mop \circ \Phi_s(\XX_0)) \delta\XX_0 \; \mathrm{d}s + o(\|\delta \XX_0\|)
\]
and we obtain also that
\[
\Phi_t(\XX_0 + \delta \XX_0) - \Phi_t(\XX_0)= \int_{0}^{t} \dif\Mop[\Phi_s(\XX_0)] \,  \,  [\dif\Phi_s(\XX_0) \delta \XX_0] \, \mathrm{d}s  + o(\|\delta \XX_0\|) = \int_{0}^{t} \partial_{X}\Mop[\Phi_s(\XX_0)] \,  \, \delta\XX(s , \bcdot) \; \mathrm{d}s + o(\|\delta \XX_0\|) .
\]
The variation of the flow verifies $\delta \XX(t ,\bcdot)= \int_{0}^{t} \partial_{X}\Mop[\XX_s]  \delta\XX(s, \bcdot)  \; \mathrm{d}s$ \, almost everywhere on $\Omega_x$. Recall  that $\Mop$ is assumed $C^1$, in particular $\sup_{\XX\in\Omega} \partial_X \Mop(\XX) < \infty$ (since $\Omega$ is compact).  The function $\delta \XX(t,\bcdot)$ belongs to $L^{2}(\Omega_x , \real^{d})$ and verifies for all $t\geq 0$
\begin{equation}\label{EqTgtLin1}
    \partial_t \, \delta \XX(t,\AD{\bcdot})\; = \; \partial_X \Mop [\XX_t] \delta\XX(t,\AD{\bcdot}).
\end{equation}
Each component $[\delta\XX(t,\bcdot)]_i$ of $\delta\XX(t,\bcdot)$ belongs to $L^{2}(\Omega_x , \real)$. 
Let $[g^\epsilon_x (\delta \XX_t)]_i: \Omega\mapsto \comp$ be the representation of $[\delta\XX(t,\bcdot)]_i$ in $L^{2}_{\comp}(\Omega , \nu)$ defined for all $x \in \Omega_x $ by
\[
[g^\epsilon_x(\delta \XX_t)]_i := [\delta\XX_t(x)]_i.
\]
For any $x\in\Omega_x$, the family of functions $g^\epsilon_x(\XX_t)$ correspond to vector observables of the dynamical system and  we may immediately write for all $t\geq 0$ 
\begin{equation}\label{EqTgtLin2-1}
\partial_t \, g^\epsilon_x (\delta \XX_t) \; = \; \partial_X \Mop (\XX_t) g^\epsilon_x (\delta \XX_t).
\end{equation}
We are now ready to exhibit a kernel expression of the tangent linear operator $\partial_X \Mop (\XX_t)$. For all $1\leq i \leq d$,  the function $\mathcal{L}_k \, [g^\epsilon_x ]_i \in \HS$ verifies for all $\delta \XX_t \in \Omega$
\begin{equation}
\label{tang-lin-k-1}
\langle {\cal L}_{k} \circ \partial_t \, {\cal L}_{k} \, [g^\epsilon_x ]_i \; , \, k(\, \bcdot \, , \, \delta \XX_t ) \rangle_{_\HS}= {\cal L}_{k} \; (  \partial_t \, {\cal L}_{k} \, [ g^\epsilon_x ]_i) (\delta\XX_t).
\end{equation}
Upon applying the differentiation formulae of proposition \ref{DerivativeRule} and the evolution equation \eqref{EquationFeatureMap}, we have
\begin{align*}
\langle{\cal L}_{k} \circ \partial_t \, {\cal L}_{k} \, [g^\epsilon_x  ]_i \, , \,  k(\, \bcdot \, , \, \delta \XX_t) \, \rangle_{_\HS} 
&=  - \langle {\cal L}_{k} [g^\epsilon_x ]_i \, , \, {\cal L}_{k} \circ  \partial_t \, k(\, \bcdot \, , \, \delta \XX_t) \rangle_{_\HS}   \\
&=  - \langle {\cal L}_{k}[g^\epsilon_x]_i \, , \, {\cal L}_{k}  \, [ A_{_U} \, k(\, \bcdot \, , \, \delta \XX_t ] \,\rangle_{_\HS}  \\
&=  - \langle \, {\cal L}^{\alf}_{k} \,[g^\epsilon_x]_i \; , \; {\cal L}^{\alf}_{k}  \, [ A_{_U} \, k(\, \bcdot \, , \, \delta \XX_t ] \,\rangle_{L^{2}_{\comp}(\Omega , \nu)} . 
\end{align*}
The operator $\mathcal{L}_{k}^{\alf}$ being self-adjoint in $L^{2}_{\comp}(\Omega , \nu)$ and $A_{_U}$  skew-symmetric for the inner product of $L^{2}_{\comp}(\Omega , \nu)$, we have
\begin{align*}
\langle{\cal L}_{k} \circ \partial_t \, {\cal L}_{k} \, [g^\epsilon_x ]_i\; , \;  k(\, \bcdot \, , \, \delta \XX_t) \,\rangle_{_\HS} 
&=\langle A_{_U} \, {\cal L}_{k}[g^\epsilon_x]_i\; , \; k(\, \bcdot \, , \, \delta \XX_t ) \,\rangle_{L^{2}_{\comp}(\Omega , \nu)}  \\
&= \langle {\cal L}^{\alf}_{k} \, A_{_U} \, {\cal L}_{k}[g^\epsilon_x ]_i \; , \; {\cal L}^{-\alf}_{k} \, k(\, \bcdot \, , \,\delta \XX_t ) \,\rangle_{L^{2}_{\comp}(\Omega , \nu)}  \\
&= \langle {\cal L}_{k} \, A_{_U} \, {\cal L}_{k}[g^\epsilon_x]_i \; , \; k(\, \bcdot \, , \,\delta \XX_t ) \,\rangle_{_\HS} \\
&=  {\cal L}_{k} \, A_{_U} \, {\cal L}_{k} \, [g^\epsilon_x ]_i (\delta \XX_t ) . 
\end{align*}
Combining the right-hand side of the above expression  with \eqref{tang-lin-k-1},  we have
\begin{equation}\label{TgtLin4-1}
{\cal L}_{k} ( \, \partial_t \, {\cal L}_{k} [g^\epsilon_x ]_i \, )(\delta \XX_t ) = {\cal L}_{k}  \, (  A_{_U} \, {\cal L}_{k}[g^\epsilon_x ]_i \, ) (\delta \XX_t ) . 
\end{equation}

\noindent 
As ${\cal L}_{k}$ is injective, the kernel of the tangent linear operator in the RKHS $\HS$ reads: 
\[
\left(\partial_t \, {\cal L}_{k} [g^\epsilon_x ]_i\right)(\delta \XX_t)  = A_{_U} \, {\cal L}_{k} \, [g^\epsilon_x ]_i (\delta \XX_t).
\]
As $\Omega$ is bounded, the function $t\mapsto \delta \XX(t,\bcdot)$ belongs to $L^{\infty}([0,T])$ and therefore $\delta \XX_t\mapsto g^\epsilon_x (\delta \XX_t)$ belongs to $L^{\infty}\bigl([0,T] , L^{2}(\Omega, \nu)\bigr)$. By \eqref{EqTgtLin2-1}, we have also that $\delta \XX_t \mapsto \partial_t \, g^\epsilon_x (\delta \XX_t)$ belongs to $L^{\infty}\bigl([0,T] , L^{2}(\Omega, \nu)\bigr)$ since $\sup_{\YY \in \Omega} \partial_X \Mop(\YY) < \infty$. We have therefore on $\Omega$ 
\[
\partial_t \, \mathcal{L}_{k} [g^\epsilon_x ]_i (\delta \XX_t)\; = \; \int_{\Omega} k(\delta \XX_t \,, \, z)\; \partial_t [g^\epsilon_x]_i(z) \, \nu(\mathrm{d}z) \; = \mathcal{L}_{k} \, (\partial_t \, [g^\epsilon_x ]_i \,) (\delta \XX_t),
\]
and thus, 
\[
 \partial_t \, [g^\epsilon_x]_i(\delta \XX_t)  \; = \; {\cal L}^{-1}_{k}  \, A_{_U} \, {\cal L}_{k} \, [g^\epsilon_x]_i (\delta \XX_t).
\]
By the commutation property of remark \ref{CommutationL_k-1/2KoopmanU_t}, together with \eqref{EqTgtLin2-1} and as ${\cal L}^{\alf}_{k} \, [g^\epsilon_x]_i \in \HS$, the following equalities hold for all $\delta \XX_t \in \Omega$
\[
 \partial_t \, [g^\epsilon_x]_i(\delta \XX_t)  \; = \; {\cal L}^{-\alf}_{k}  \, A_{_U} \, {\cal L}^{\alf}_{k} \, [g^\epsilon_x]_i (\delta \XX_t),
\]
\begin{equation}\label{TgtLineaire4-1}
\partial_X \Mop (\XX_t) \, g^\epsilon_x \, (\delta \XX_t) = {\cal L}^{-\alf}_{k} \,  A_{_U} \, {\cal L}^{\alf}_{k} \, g^\epsilon_x(\delta \XX_t).
\end{equation}

Let's specify now the kernel expression of the tangent linear dynamics in $\Omega_s$ for all $s\geq 0$. The function $\mathcal{L}_{k}^{\alf} \, [g^\epsilon_x] \restriction_{\Omega_s}$ belongs to $\HS_s$. We have as well $\mathcal{L}_{k}^{\alf}(\partial_t g^\epsilon_x) = \mathcal{L}_{k_s}^{\alf}(\partial_t g^\epsilon_x)$ on $\Omega_s$. For all $\YY_s$ and $\XX'_s \in \Omega_s$ we have that $A_{_U} \, k(\,\bcdot \, , \, \XX'_s)(\YY_s) = A_{U,s} \, k_s(\,\bcdot \, , \, \XX'_s)(\YY_s)$. In particular, we get $A_{_U} \, \mathcal{L}_{k}^{\alf} \, [g^\epsilon_x] (\YY_s) = A_{U,s} \,\mathcal{L}_{k_s}^{\alf} \, [g^\epsilon_x](\YY_s)$. We obtain hence the kernel expression for all $\YY_s \in \Omega_s$
\begin{equation}\label{TgtLineaire3-1}
\partial_X \Mop (\XX_s) \, g^\epsilon_x \, (\YY_s) = {\cal L}^{-\alf}_{k_s} A_{U\!,\,s} \;{\cal L}^{\alf}_{k_s} \, g^\epsilon_x \, (\YY_s).
\end{equation}
 Note that the domains of the kernel expressions  of the tangent linear \eqref{TgtLineaire4-1} and \eqref{TgtLineaire3-1} are  different. The right-hand side of \eqref{TgtLineaire3-1}  provides a convenient kernel expression of the tangent linear operator, enabling us to evaluate  the tangent linear dynamics from an ensemble of feature maps. The adjoint of the tangent linear dynamics is straightforwardly given by:
\begin{equation}\label{AdjTgtLineaire3-1}
\partial_X \Mop^* (X_s)  g^\epsilon_x  \, (\YY_s) =  - {\cal L}^{-\alf}_{k_s} A_{U\!,\,s} \;{\cal L}^{\alf}_{k_s}  g^\epsilon_x  \, (\YY_s).
\end{equation}

\begin{remark}[Projection observables]\label{projection observable}
The point observable functions $g_x^\epsilon$ used above can be extended to other functions defined from basis $(\Psi_j)_{j\geq 0}$ of $L^2(\Omega_x, \mathbb{R}^d)$, with  
\begin{align*}
&g^\psi_x\XX_t := \sum_j^\infty \langle \XX_t\, , \Psi_j \rangle_{L^2(\Omega_x)} \Psi_j(x),\\
&g^\psi_x: \Omega \mapsto \sum_j^\infty \langle \bcdot \, , \Psi_j \rangle_{L^2(\Omega_x)} \Psi_j(x) \in \comp.
\end{align*}
\end{remark}

\begin{remark}[Dependence on $\XX_t$]
It should be noted that in \eqref{TgtLineaire3-1} and \eqref{AdjTgtLineaire3-1} the right-hand side does not depend on $\XX_s$ whereas the tangent linear operator or its adjoint on the left-hand side does. Function $\XX_s$ indicates around which function of $\Omega_s$ the nonlinear system is linearized. The equivalent of the infinitesimal generator of the Koopman operator, representing the dynamics' linear tangent operator on the RKHS family, depends necessarily also on this function. This dependence is here implicit and induced by the  considered sampled functions used to define the RKHS $\HS_t$. If the set of members $(\XX_t^{(i)})_{i\geq 1}$ are centered around a particular function $\XX_s$, the operators \AD{$A_{U\!,\,t}$}  can be interpreted as a representation of the tangent linear operator around function $\XX_s$. In ensemble methods, $\XX_s$, is in general taken as the ensemble mean, and $(\XX^{(i)}_s)_{i\geq 1}$ is an ensemble of of time-dependent perturbations  around this mean. 
\end{remark}
\begin{remark}[More regularity on $\delta \XX_t$]\label{MoreRegularity}
For all $1\leq i \leq d$,  if we suppose that the function $[ g^\epsilon_x ]_i$ belongs to $\HS \subset L^{2}_{\comp}(\Omega , \nu)$, the proof can be simplified and \eqref{TgtLineaire3-1} is replaced by 
\begin{equation}
\partial_X \Mop (\XX_s) \, g^\epsilon_x  \, (\delta \XX_s) =  A_{U\!,\,s} \, g^\epsilon_x \, (\delta \XX_s).
\end{equation}
\end{remark}
It can be pointed out that the expression above corresponds to the approximation of the tangent linear dynamics  used in ensemble method if we work in a finite dimensional space such that $\Omega \subset  \real^n$ and assume  that $k_t(\bcdot, X_t)$ is  defined  as $(p-1)^{-\alf}\langle ( \XX_t -  \overline{\XX}_t ),\bcdot\, - \overline{X}_t\rangle_{\real^n}$,  with $\overline{\XX}_t$ the empirical ensemble mean. With that definition, we have $A_{U\!,\,t} k_t(\bcdot , X^{(i)}) =  \sum_j {\partial_{X_j}} k_t(\bcdot ,X_t^{(i)})  \Mop(\bcdot)_j = \langle (\delta_{X^{(i)}_t} -  \frac{1}{p}\sum_{\ell} \delta_{X^{(\ell)}_t})\Mop(\bcdot),\bcdot\, - \overline{X}_t\rangle_{\real^n}$ that reads $(p-1)^{-\alf}\langle \Mop(X_t^{(i)}) -\frac{1}{p}\sum_j\Mop(X_t^{(j)}), \bcdot\, - \overline{X}_t\rangle_{\real^n} $,  for $i=1,\ldots,p$ and for which, when associated to the $\real^n$ Euclidean inner product on a resolution grid of size $n$, the left-hand side of this latter expression corresponds to the so-called $(p \times n)$ anomaly matrix built from $p$ ensemble members of the dynamical system. The tangent linear approximation provided by ensemble methods can be thus immediately interpreted as a particular instance of feature maps together with a given choice of specific inner product to define the reproducing kernel. Keeping a finite dimensional approximation but working without assuming that the functions  $[ g^\epsilon_x ]$ belong to $\HS$, and thus with now expression  \eqref{TgtLineaire4-1} for the ensemble tangent linear expression, corresponds to the case in which a localization procedure identified to the square root operator ${\cal L}^{\alf}_{k_t}$ has been considered. These two choices embed the problem within a particular RKHS family of functions.  The relation between the tangent linear dynamics and the anomaly matrix is in our case  exact and does not correspond to a finite difference approximation as classically presented in ensemble methods. \AD{The RKHS} family can be seen as a way of linearizing locally a nonlinear system in a convenient sequence of spaces  of smooth functions.   

\subsection{Finite time Lyapunov exponents}
\label{sec:Lyapunov}
The kernel of the Koopman operator  provides also  a direct access to the finite time Lyapunov exponents. Recalling from \eqref{EqTgtLin2-1} that for any punctual observable $g_x$, as defined previously, we have 
\begin{equation}\label{Lyapunov1}
g_x (\delta \XX_t) \; = \int_{0}^{t} \partial_X \Mop (\XX_s) g_x (\delta \XX_s)\dif s + o(\|\delta X_0 \|).
\end{equation}
With the expression of the tangent linear operator in terms of the \AD{evolution operator} \eqref{TgtLineaire4-1} on  $\Omega_s$, we have then
\begin{equation}
\label{delta-x}
 g_x(\delta \XX_t)  =     \int_0^t\!\!\!\; {\cal L}^{-\alf}_{k_s} A_{U,s} \;{\cal L}^{\alf}_{k_s}  g_x  (\delta \XX_s) \; \mathrm{d} s.
\end{equation}
For all $s\in [0,t]$, we consider a perturbation $g_x (\delta \XX_s) =  {\cal L}^{-\alf}_{k_s} \psi_\ell^{s}(\delta \XX_s)$ along a Koopman generator eigenfunction associated to the eigenvalue of maximal modulus $|\lambda_\ell|$. By {T}heorem \ref{W-spectral-representation}, we have 
\[
    \partial_t g_x  (\delta \XX_s)  = \lambda_{\ell} \, {\cal L}^{-\alf}_{k_s} \psi^s_{\ell} \\
    = \lambda_{\ell} \, g_x (\delta \XX_s).
\]
Therefore we get
\begin{equation}
\label{norm-delta-x}
 | g_x  (\delta \XX_t)| =  e^{|\lambda_{\ell}| t} \; | g_x(\delta \XX_0)|,
\end{equation}
and the finite time Lyapunov exponent is consequently defined as 
\begin{equation}
\sigma =  |\lambda_{\ell}|.
\label{eq:modal_Lyap}
\end{equation}
For regular perturbations $g_x  \in {\cal H}_s$ with unitary perturbation $g_x (\delta \XX_s)=  \psi_\ell^{s}(\delta \XX_s)$, the derivation is even simpler as we obtain from remark \ref{MoreRegularity}
\[
\partial _t g_x (\delta \XX_s) \; =  A_{U\!,\,t}\;  g_x (\delta \XX_s),
\]
which yields directly to expression \eqref{norm-delta-x} and to the same expression for the Lyapunov exponent.\\
\noindent
The modulus of the larger Koopman eigenvalue in \AD{the RKHS family} provides thus an estimate of the Lyapunov exponent.
It can be outlined that the computation of Lyapunov exponents for large scale systems is computationally very demanding as it requires the construction of the linear tangent dynamics operator and the solution of an eigenvalue problem of very big dimension.
The construction of the exact numerical tangent linear operator is in general a tedious task when expressed in $L^2(\Omega_x)$, as in equation~\eqref{EqTgtLin1}.
The ensemble based method provided by our formalism is on the contrary very simple by expressing the tangent linear operator in $L^2_\mathbb{C}(\Omega_t,\nu)$, as in equation~\eqref{EqTgtLin2-1}.
It can be noticed that, by this change of norm in the definition, the computed values are not the same.

Three distinct values can then be defined for practical computations. 
First, the Lyapunov spectrum expressed in $L^2_\mathbb{C}(\Omega_t,\nu)$ can be determined by computing the singular values of ${\cal L}^{-\alf}_{k_t} A_{U,t} \;{\cal L}^{\alf}_{k_t}$.
The time integral is dropped since the evaluations are constant along trajectories.
It can be viewed as an advantage of working in $L^2_\mathbb{C}(\Omega_t,\nu)$ instead of $L^2(\Omega_x)$ as performed classically.
The time independence is due to the fact that the Koopman operator is intrinsic to the dynamical system. However, from a numerical point of view, as the computation is performed in practice through an ensemble with a limited number of members, the learned spectrum is representative only of the local dynamics at the time ($t=t_0$) at which the kernel has been evaluated.
As an alternative, modal Lyapunov exponents can be defined by the square root of the first singular values of ${\cal L}^{-\alf}_{k_t}\psi_\ell |\lambda_\ell|^2 \psi_\ell^*K_t^{-1}{\cal L}^{\alf}_{k_t}={\cal L}^{-\alf}_{k_t}\psi_\ell |\lambda_\ell|^2 \psi_\ell^*{\cal L}^{-\alf}_{k_t}$
 with $K_t(i,j)= k_t(\XX_t^{(i)}\,,\, \XX_t^{(j)})$.
We call these singular values the \emph{Koopman modal Lyapunov exponents} (KMLE).
Finally, as a third option, equation~\eqref{eq:modal_Lyap} can simply be considered to evaluate modal exponents.
We can notice that the two modal Lyapunov exponents definitions are very similar; the former being expressed in  $L^2_\mathbb{C}(\Omega_t,\nu)$ and the latter in $\mathcal{H}_t$. An example of the estimation of the two types of modal exponents will be presented for a quasi-geostrophic dynamics in the numerical section.
 
\subsection{Practical considerations}\label{Practical-Considerations}
Let us stress again  that in practice, we only have access to the mappings $k_t$ and $A_{U\!,\,t}$ with $t\geq 0$. The mapping $k$ and $A_{_U}$ are completely inaccessible for high-dimensional systems as they require the complete knowledge of the phase space or at least of a long enough orbit with a density assumption in the whole phase space. This last assumption is associated to strong requirements of the dynamical system and is not necessarily valid for a given time series of a particular observable. 
Instead of working with an infinite (dense) trajectory, Theorem\;\ref{W-spectral-representation} enables us to estimate the eigenvalues and eigenfunctions of the Koopman operator locally in the RKHS family, which can locally conveniently be accessed from an ensemble of finite time trajectories. As it will be described in the following,  operators $A_{U\!,\,t}$, can be discretized as an ensemble  matrix -- itself related, as we saw it previously, to the tangent linear dynamics operator. This matrix is then diagonalized  to get access to Koopman eigenvalues and their  associated eigenfunctions.  In theory, the diagonalization of $A_{U\!,\,t}$ needs to be performed only once, at a given time, to access the Koopman eigenpairs  $(\psi_{\ell}^{t})_{\ell}$ and $(\lambda_{\ell})_{\ell}$. However, the exponential relation between distinct instants allows us also to consider averaging strategies to eventually robustify the estimation in practice. This capability will be exploited in the numerical section  for the data assimilation of time series. 
\subsubsection*{Diagonalization in practice}
For all $t\geq 0$, let $m_t$ be the kernel expression of the operator $A_{U\!,\,t}$ given by $m_{t}(\XX_{t} \, , \, \YY_{t}) := A_{U\!,\,t}^{*}\bigl[ \, k_{t}( \, \bcdot \, , \, \YY_t  ) \, \bigr](\XX_{t})$ for all $\XX_{t}$ and $\YY_{t} \in \Omega_{t}$. By {P}roposition \ref{AntisymAUt}, we have $m_{t}(\XX_{t} \, , \, \YY_{t}) = - A_{U\!,\,t}\bigl[ \, k_{t}( \, \bcdot\, , \,  \YY_t ) \, \bigr](\XX_{t})$. 

Let us denote by $\{ X_t^{(i)} \, : \, 1\leq i \leq n \}$ an ensemble of members generated by the dynamical system and by $\{ k_{t}(\, \bcdot \, , \, \XX^{(i)}_{t}) \; : \; 1\leq i \leq n \}$ the $n$ associated feature maps. For all time $t\geq0$, these $n$ feature maps enable us to build a kernel  expression of operator $A_{{U}, t}$ as the $n \times  n$ matrix  ${\mathbb M}_t = \left(\, {m_t}(\XX^{(i)}_t \, , \, \XX^{(j)}_t) \, \right)_{1\leq i , j \leq n}$ with:
\[
{\left(\mathbb{M}_t\right)_{ij}} \; := \; - A_{{U}\!,\,t} \left[ \, k_{t}( \bcdot \, , \, \XX^{(j)}_t ) \, \right] (\XX_{t}^{(i)}) = 
\left[-  \partial_{\Mop(\bcdot)}\, k_t (\, \bcdot \, , \, \XX^{(j)}_{t})\right]\,(\XX^{(i)}_{t}).
\]
As shown in the following, this matrix enables us to access to the Koopman generator eigenvalues and to the evaluation of the eigenfunctions at the ensemble members.

By definition of $A_{U\!,\,t}$, we have $A_{U\!,\,t}[\, k_{t}( \bcdot \, , \, \XX^{(j)}_t ) \, ]  = A_{_U} [ \, k( \bcdot \, ,\, \XX^{(j)}_t) \, ] $. 
Otherwise, as $ -\mathcal{L}_{k}^{\alf} \circ A_{_{U}} [ \, k( \, \bcdot \, , \, \XX^{(j)}_t  ) \, ]$ belongs to $\HS$, by the proof (Step 1) (\ref{Step-1-2}) of Theorem \ref{W-spectral-representation} and denoting $(\psi_{\ell})_{\ell}$ an orthonormal basis of $\HS$ sets from the eigenfunctions of $A_{_U}$, we get
\begin{align*}
- \mathcal{L}_{k}^{\alf} \circ A_{_U} \bigl[ \, k( \bcdot \, , \, \XX^{(j)}_t  ) \, \bigr]
&=  - \sum_{\ell=0}^{\AM{\infty}} \langle \mathcal{L}_{k}^{\alf}  \circ A_{_U} \, \bigl[ \, k( \bcdot \, , \, \XX^{(j)}_t) \, \bigr] \, , \, \psi_{\ell} \, \rangle_{_{\HS}} \; \psi_{\ell}, \\
&= - \sum_{\ell=0}^{\AM{\infty}} \langle A_{_U} \, k( \, \bcdot \, , \, \XX^{(j)}_t) \; , \; \mathcal{L}_{k}^{-\alf} \, \psi_{\ell} \, \rangle_{L^{2}_{\comp}(\Omega , \nu)} \; \;  \psi_{\ell}.
\end{align*}
Noting $\beta_{\ell}$ the eigenvalues of $\mathcal{L}_{k}$, we have also (step 1 \ref{Step-1-2} theorem \ref{W-spectral-representation})  $\mathcal{L}_{k}\AD{j(\psi_{\ell})} = \beta_{\ell} \, \AD{j(\psi_{\ell})}$ and 
\[
- \mathcal{L}_{k}^{\alf} \circ A_{_U} \bigl[ \, k( \bcdot \, , \, \XX^{(j)}_t  ) \, \bigr]  = - \sum_{\ell=0}^{\AM{\infty}} \, \beta_{\ell}^{-\alf} \, \langle A_{_U} \, k( \, \bcdot \, , \, \XX^{(j)}_t) \; , \; \AD{j(\psi_{\ell})} \, \rangle_{L^{2}_{\comp}(\Omega , \nu)} \; \;  \psi_{\ell} .
\]
By the skew symmetry of the generators (proposition \ref{KoopAntisym}) we have that
\[
- \mathcal{L}_{k}^{\alf} \circ A_{_{U}} \bigl[ \, k( \bcdot \, , \,  \XX^{(j)}_t ) \, \bigr] = \\\sum_{\ell=0}^{\AM{\infty}} \, \lambda_{\ell} \; \beta_{\ell}^{-\alf} \;  \langle  k( \bcdot \, , \, \XX^{(j)}_t) \, , \, \AD{j(\psi_{\ell})} \, \rangle_{L^{2}_{\comp}(\Omega , \nu)} \; \;  \psi_{\ell} =\sum_{\ell=0}^{\AM{\infty}} \, \lambda_{\ell} \; \beta_{\ell}^{-\alf} \; \overline{\mathcal{L}_{k}\AD{j(\psi_{\ell})}(\XX_{t}^{(j)})} \;  \psi_{\ell},
\]
which leads to
\[
- \mathcal{L}_{k}^{\alf} \circ A_{_{U}} \bigl[ \, k( \bcdot \, , \,  \XX^{(j)}_t ) \, \bigr] = \sum_{\ell=0}^{\AM{\infty}} \, \lambda_{\ell} \; \beta_{\ell}^{\alf} \overline{\psi_{\ell}(\XX_{t}^{(j)})} \; \psi_{\ell},
\]
and, applying $\mathcal{L}_{k}^{-\alf}$ on both sides, we get
\[
- A_{_{U}} \bigl[ \, k( \bcdot \, , \,  \XX^{(j)}_t ) \, \bigr] = \sum_{\ell=0}^{\AM{\infty}} \, \lambda_{\ell} \; \beta_{\ell}^{\alf} \; \overline{\psi_{\ell}(\XX_{t}^{(j)})} \; \; \mathcal{L}_{k}^{-\alf}\psi_{\ell} = \sum_{\ell=0}^{\AM{\infty}} \, \lambda_{\ell} \; \overline{\psi_{\ell}(\XX_{t}^{(j)})} \;  \AD{j(\psi_{\ell})}.
\]
By the restriction expression (remark \ref{ExpressionRestrictionFonctionRKHS}, \ref{EXT-REST}), \AM{and keeping only $n$ eigen functions to represent the values at the $n$ members}, we finally obtain the following equality for all $t\geq 0$
\[
{m_t}(\XX^{(i)}_t \, , \, \XX^{(j)}_t) = \sum_{\ell=0}^{\AD{n}} \, \lambda_{\ell}  \;  \psi^{t}_{\ell}(\XX_t^{(i)}) \; \overline{\psi^{t}_{\ell}(\XX_{t}^{(j)})},
\]
which shows that the diagonalization of ${\mathbb M}_t$ provides a set of the Koopman generator eigenpairs in \AD{the RKHS family}. 

In practice, the skew-symmetric matrix $\mathbb{M}_t := \bigl( {m_t}(\XX^{(i)}_t \, ,\, \XX^{(j)}_t)\,\bigr)_{_{1\leq i , j \leq n}}$ is assembled from the definition of $A_{U\!,\,t}$ and a given choice of the kernel. As explained in the previous section this matrix corresponds to a kernel expression of $A_{U\!,\,t}$ with 
\[
{\mathbb M}_t(\XX^{(i)}_t \, , \, \XX^{(j)}_t) \; := \; - A_{U\!,\,t} \bigl[ \, k_{t}( \bcdot \, , \, \XX^{(j)}_t ) \, \bigr] (\XX_{t}^{(i)}) = 
[- {\partial_{\Mop(\bcdot)}}\, k_t (\, \bcdot \, , \, \XX^{(j)}_{t})]\,(\XX^{(i)}_{t}).
\]
This matrix can be interpreted in the RKHS setting as resulting from the matrix multiplication:
\[
{\mathbb M}_t(\XX^{(i)}_t \, , \, \XX^{(j)}_t) \; := \;  \sum_\ell{\mathbb F}_t(j,\ell) K_t(\ell,i{)},
\]
with $ {\mathbb F}_t(i,\ell) = - A_{U\!,\,t} \bigl[ \, k_{t}(  \AM{X}_t^{(\ell)}\, , \, \XX^{(j)}_t )\bigr] $. 
This indeed corresponds to a discretization of the kernel expression of operator $A_{U\!,\,t}$ through the empirical Dirac measure. 
Numerically, instead of working with matrix ${\mathbb M}_t$ (i.e. the evaluation of $A_{U\!,\,t} k_t (\bcdot, \XX)$ at several discrete points), we will work directly with matrix ${\mathbb F}_t= {\mathbb M}_t K_t^{-1}$.  This has  the advantage of  directly working with an implicit discretization of operator $A_{U\!,\,t}$, and to relax somewhat its dependency on the kernel choice. The  skew-symetric matrix ${\mathbb M}_t K_{{t}}^{-1}$ is then diagonalized through a direct numerical procedure (using LAPACK library and working numerically on the anti-symetric part of ${\mathbb M}_t K_{{t}}^{-1}$) and can be written $\mathbb{M}_t K_{{t}}^{-1} =  V_{t}\Lambda V^{T}_{t}$, with $V_{t}$ a unitary matrix  and $\Lambda$ a diagonal matrix. 
The matrix $V_{t}$ gathers eigenvectors $V_j^{t}$ of ${\mathbb F}$, which {is} a discretization of $A_{U\!,\,t} k_t (\bcdot, \XX)$ {giving} access to the values of the Koopman eigenfunctions $\psi^{t}_j={\cal L}_k^{\alf}V_j^{t}$ at the $n$ ensemble members points $X^{(i)}_t$. The matrix $\Lambda$ is composed of Koopman eigenvalues with conjugate pairs of pure imaginary eigenvalues.

As previously mentioned this diagonalization can be performed at a single time or at several instants accompanied with an averaging procedure. Theorem  \ref{W-spectral-representation} and \eqref{RelationVectpropreTemps}, give access to the eigenvectors evaluation along a trajectory for all time instants. 
\newline 

We provide below as examples, expressions of the evaluation of $\mathbb{M}_t(\XXi,\XXj)=\left(-\partial_{\Mop(\bcdot)}\,k_t(\bcdot,\XXj)\right)\left(\XXi\right)$   for the empirical covariance kernel ($k_E$) and the Gaussian kernel ($k_G$).
The empirical covariance kernel is defined through the kernel isometry property  as
{$${k_t}(\XXi,\XXj)={k_E(\XXiz,\XXjz)}=\left\langle\XXiz,\XXjz\right\rangle_{\Omega_x},$$
with $(\bcdot,\bcdot)_{\Omega_x}$ the inner product of $L^{2}(\Omega_x , \real^{d})$ and where $\Omega_x$ denotes the physical domain  of the considered dynamics.
We obtain 
\begin{equation}
\label{Der-Emp-K}
\mathbb{M}_t(\XXi,\XXj)=\left\langle \frac{\partial}{\partial t} \XXjz,  \XXiz\right\rangle_{\Omega_x}.
\end{equation}
In this expression, we see that the time derivative of the ensemble members at the initial time is required.
Similarly, the Gaussian kernel is defined as
$$k_t(\XXi,\XXj)={k_G(\XXiz,\XXjz)}=\exp{\left(-\frac{1}{{\ell_G}^2}\left\|\XXiz-\XXjz\right\|_{\Omega_x}^2\right)},$$
This leads to
\begin{equation}
\label{Der-Gaussian-K}
\mathbb{M}_t(\XXi,\XXj)= -\frac{2}{{\ell_G}^2}\left\langle \frac{\partial }{\partial t}\XXjz, \left(\XXiz-\XXjz\right)\right\rangle_{\Omega_x} \exp{\left(-\frac{1}{{\ell_G}^2}\left\|\XXiz-\XXjz\right\|_{\Omega_x}^2\right)}.
\end{equation}
Thanks to the isometry property, this matrix needs to be evaluated only at a single time.

\section{Application to ensemble methods}
\label{Sec-Numerical-Results}
\subsection{Quasi-geostrophic model}
Based on a set of $p$ ensemble members, the ability of the proposed method to predict a trajectory associated with a new initial condition is demonstrated on a quasi-geostrophic barotropic model of a double-gyre \citep{greatbatch2000} in a \enquote{rectangular-box ocean}.
This idealistic model of North-Atlantic large-scale circulation   considers the transport of the ageostrophic state variable components by a velocity in geostrophic balance.
The \AD{adimensioned} barotropic vorticity equation written for the potential vorticity $q$ and the stream function $\psi$ in $\Omega_x=\left([0:L]\times[-L:L]\right)$, with $L=1$, is
\begin{equation}
   \left\{
   \begin{aligned}
      &\frac{\partial q}{\partial t}+J(\psi, q)= f - \left(\delta/L\right)^5\nabla^4\omega\\
      & \Delta\psi=\omega \quad ; \quad q=R_o \omega + y_c,
   \end{aligned}
   \right.
   \label{eq:QG}
\end{equation}
where $\omega$ is the vorticity and $\mbs{x}_c=(x_c,y_c)$ are the latitude and the longitude  increasing eastward and northward, respectively, while $J(f,g)=\partial_x f \partial_y g - \partial_x g \partial_y f$ stands for the Jacobian (with $x$ and $y$ denoting horizontal and vertical coordinates of $\real^2$. \AD{In the model \eqref{eq:QG}, the state of the system can be uniquely determined from the potential vorticity $q$. The set of all possible states of the system is then $\Omega\subset L^{2}(\Omega_x,\mathbb{R})$ (i.e. $d=1$).}
The dimensionless Rossby number $R_o=0.0036$ measures the ratio between inertial forces and the Coriolis force, and a steady wind-forcing is defined as $f=\sin(\pi y_c)$.
The parameter $\delta/L=0.032$ controls the hyper-viscosity model, with $\delta$ the Munk boundary layer width.
Scales smaller than $\delta$ will not be resolved.
Numerical details and links with dimensional variables associated with a mid-latitude ocean basins, such as North Atlantic, \AD{as well as the way the initial conditions of the ensemble have been generated} are given in \ref{A-Numerical-details}. \AD{Let us note that this system is ergodic when associated to an uncorrelated random forcing \cite{Brannan-Duan-98,Duan-2001,Yang-Pu-17}. Here the forcing is stationary, and the system is not ergodic. As a consequence, such a system cannot be used for all the techniques relying on an ergodic assumption to extract a spectral representation of the Koopman operator.}

\subsection{Reconstructions using RKHS}
In this study we focus only on two specific kernels: the empirical covariance kernel ($k_E$), and the Gaussian Kernel ($k_G$).
They are defined respectively by 

\[k_{E}(\XXi,\,\XXj)=\left\langle {{q}}^{(i)}(\bcdot,t_0)\,,\,{{q}}^{(j)}(\bcdot,t_0)\right\rangle_{\Omega_x},
\]
 and 
\[k_{G}(\XXi,\,\XXj)=\exp{\frac{1}{{\ell_G}^2}\left\| {q}^{(i)}(\bcdot,t_0)-{q}^{(j)}(\bcdot,t_0)\right\|^2_{\Omega_x}},
\]
  where $\left\langle\bcdot,\bcdot\right\rangle_{\Omega_x}$ stands for the standard $L^2$ inner product in $\Omega_x$ {and $q^{(i)}(x,t)$ is the solution of equation~\eqref{eq:QG} of the $i^{\text{th}}$ ensemble member at time $t$ at the space location $x$.}
The former has the advantage to be simple, and to be directly related to standard ensemble methods employed in data assimilation.
However, it is strongly rank deficient and introduces spurious correlations in the physical domain $\Omega_x$ between two far apart locations.
Localization techniques \citep{hamill2001} have been introduced to tackle this problem and may be interpreted in our setting as embeddings in specific RKHS.
This interpretation may open new computationally efficient generalisations related for instance to localization techniques in the model space as explained after remark \ref{MoreRegularity}. The parallel done between operators ${\cal L}^{\alf}_{k_t}$ and localization procedure \AD{ should in particular enable to define time-varying and flow dependant localization techniques that are notably difficult to built in practice}. Gaussian kernels are widely used in machine learning for its high regularity and non-compact support in $L^2(\Omega,\nu)$ \citep{minh2010}.
It can be also interpreted as a way to perform localization.
In this study, compared to other tested -- but not presented -- kernels (polynomial and exponential for instance), it led to better performances.  

Once the kernel matrix $K_{ij}=k_t(\XXi\,,\,\XXj)$ for $(i,j)\in [1\dots p]\times[1\dots p]$ (equal for all $t$ -- kernels isometry,  Theorem~\ref{KoopVarieteIsometry}) is evaluated, a very simple reconstruction of the whole trajectory can be performed.
Denoting $\XX'_0$ a new ensemble member initial condition,  the linear combination vector $\bm{\beta}$ associated to the reconstructed trajectory $\hat{\XX}'_t=\sum_{i=1}^p\beta_i\XXi$ can be estimated by kernel regression from the evaluation $k(\XX'_0,\XX_0^{j})$, for $j=[1\dots p]$, and the reproducing property as
\begin{equation}
    \bm{\beta}=K^{-1}k(\XX'_0\,,\,\XX_0^{[1\dots p]}).
    \label{eq:kernel-inversion}
\end{equation}
Strikingly, Koopman isometry implies that these coefficients are constant along time, and brings hence a direct reconstruction of  the whole trajectory.

\subsection{Evaluation of the Koopman eigenfunctions}
Performing the $\mathcal{H}_t$ inner product of equation~\eqref{eq:Au} with the ensemble members' feature maps leads to  the matrix
 \begin{equation}
      \mathbb{M}(\XXi,\,\XXj)=\left(-{\partial_{\Mop(\bcdot)}}\, k_t (\bcdot, \XXj)\right)\left(\XXi\right).
      \label{eq:M}
 \end{equation}
As explained previously eigen-elements of $\mathbb{M}K^{-1}$, noted $(\lambda_\ell,V_\ell(\XXi))$, give access to the exact evaluation of Koopman eigenfunctions on the RKHS family at the ensemble members $\psi_\ell^t(\XXi)= {\cal L}_{k}^{\alf}V_\ell(\XXi)$, and to the associated eigenvalues $\lambda_\ell$.
 We note that $\partial_{\Mop(\bcdot)} k_t (\bcdot\,,\, \XXj)$ is obtained by differentiating the kernel expression (see equations \ref{Der-Emp-K}, \ref{Der-Gaussian-K} for the two selected kernels).
Skew-symmetry of $\mathbb{M}K^{-1}$, inherited from the skew-symmetry of $A_{U\!,\,t}$, is only approached numerically.
Evaluations are thus enhanced by considering only the skew-symmetric part of the matrix $\mathbb{M}K^{-1}$, which is dominant in the numerical evaluation of~\eqref{eq:M}.

\subsection{Lyapunov times}
The Lyapunov times evaluated using the two kernels $k_E$ and $k_G$ are compared.
Three methods to compute the Lyapunov exponents $\sigma$ are detailed in section~\ref{sec:Lyapunov}.
The associated Lyapunov times are defined by $T=\sigma^{-1}\ln C$, with, $C=10^3$.

The Lyapunov times associated with exponents expressed in $L^2_\mathbb{C}(\Omega_t,\nu)$ are shown in figure~\ref{fig:Lyapunov_global}.
The spectrum associated with all singular values, truncated such that $\sigma>10^{-5}$, is displayed.
Due to this truncation, very slow modes are not well captured.
We can see that the Gaussian kernel $k_G$ estimates in average larger Lyapunov times.

\begin{figure}
    \centering
    \subfigure[Lyapunov times spectrum expressed in $L^2_\mathbb{C}(\Omega_t,\nu)$ for the Gaussian kernel $k_G$ and the empirical covariance $k_E$.]{\includegraphics[width=0.45\textwidth]{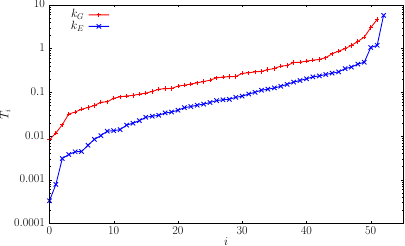}\label{fig:Lyapunov_global}}
    \hfill
    \subfigure[Modal Lyapunov times associated with individual Koopman modes. Comparison with expression in $L^2_\mathbb{C}(\Omega_t,\nu)$ and in $\mathcal{H}_t$.]{\includegraphics[width=0.45\textwidth]{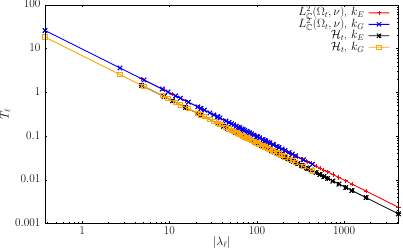}\label{fig:Lyapunov_modal}}
    \caption{Lyapunov times computed using the Gaussian kernel $k_G$ and the empirical covariance kernel $k_E$.}
    \label{fig:Lyapunov}
\end{figure}

In figure~\ref{fig:Lyapunov_modal}, the Koopman modal Lyapunov times are shown as a function of the absolute value of the Koopman eigenvalue.
The Lyapunov times expressed with the RKHS norm equation~\eqref{eq:modal_Lyap} is directly $T_\ell=1/|\lambda_\ell|\ln C$, which explains the perfect decaying slope in logarithmic scale.
We observe that the computation of the modal exponents in $L^2_\mathbb{C}(\Omega_t,\nu)$ leads simply to a vertical shift of the curves due to the scaling factor induced by the change of norm.
We can notice that with the Gaussian kernel $k_G$, slower Koopman eigenfunctions are evaluated compared to $k_E$, which leads to longer Lyapunov times.
This explains longer predictions, which will be presented in section~\ref{sec:reconstruction_filtering}. As it will be seen, the modal exponents and their associated Lyapunov times will reveal very useful, as they will allow us to filter out sequentially the contributions of the Koopman eigenfunctions beyond their predictability time.  

\subsection{Reconstructions}
\label{sec:reconstruction_filtering}
The ability of trajectories reconstruction is illustrated in figure~\ref{fig:reconstruction_Gaussian}, where a full trajectory is estimated from the embedding of a new initial condition in the RKHS family.
We compare the anomaly (figure~\ref{fig:reconstr_anomaly}) of a given new ensemble member at an advanced time (3.7 days), with reconstructions using the Gaussian kernel {$k_G$} (figure~\ref{fig:reconstr_Gaussian_1}).
Figure~\ref{fig:best_reconstruction} shows the best possible reconstruction as a linear combination of ensemble members, with regards to the $L^2(\Omega_x)$-norm of the projection error.
\begin{figure}
    \centering
    \subfigure[Test ensemble member]{
        \includegraphics[height=0.3\textwidth]{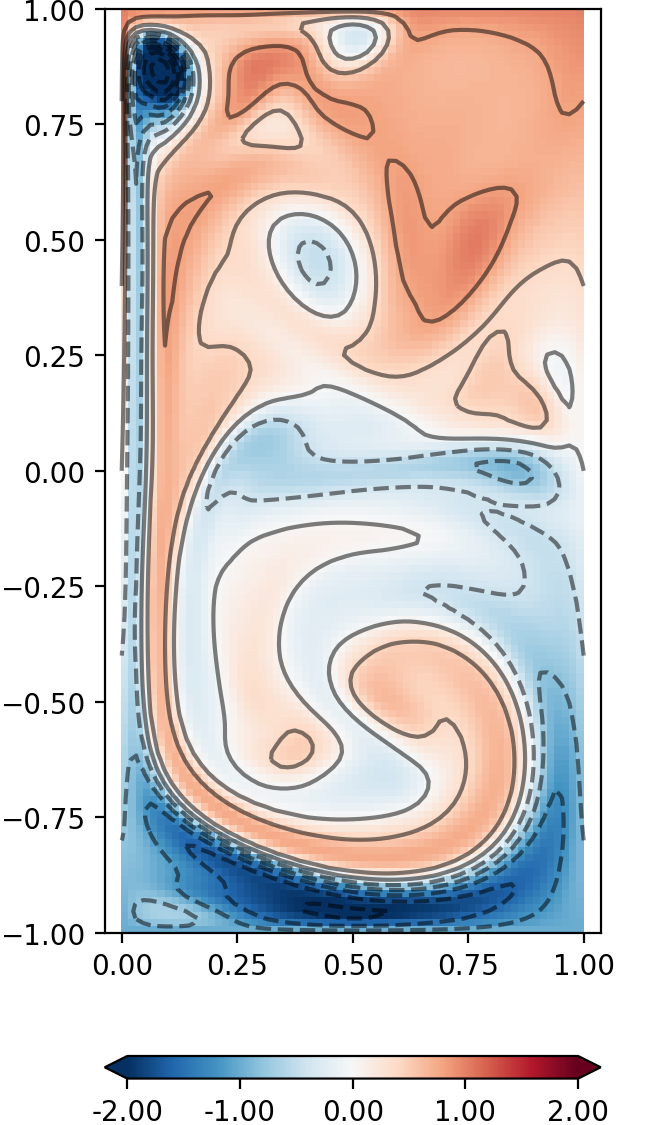}
        \label{fig:reconstr_anomaly}
    }
    \subfigure[$L^2$norm projection]{
        \includegraphics[height=0.3\textwidth]{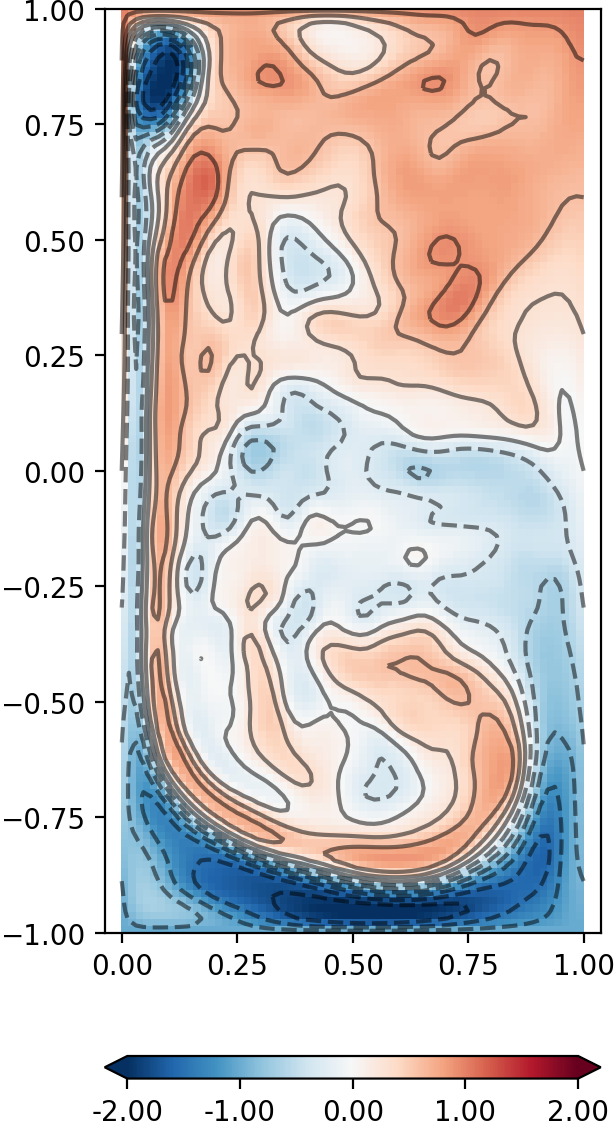}
        \label{fig:best_reconstruction}
    }
    \subfigure[Gaussian kernel]{
        \includegraphics[height=0.3\textwidth]{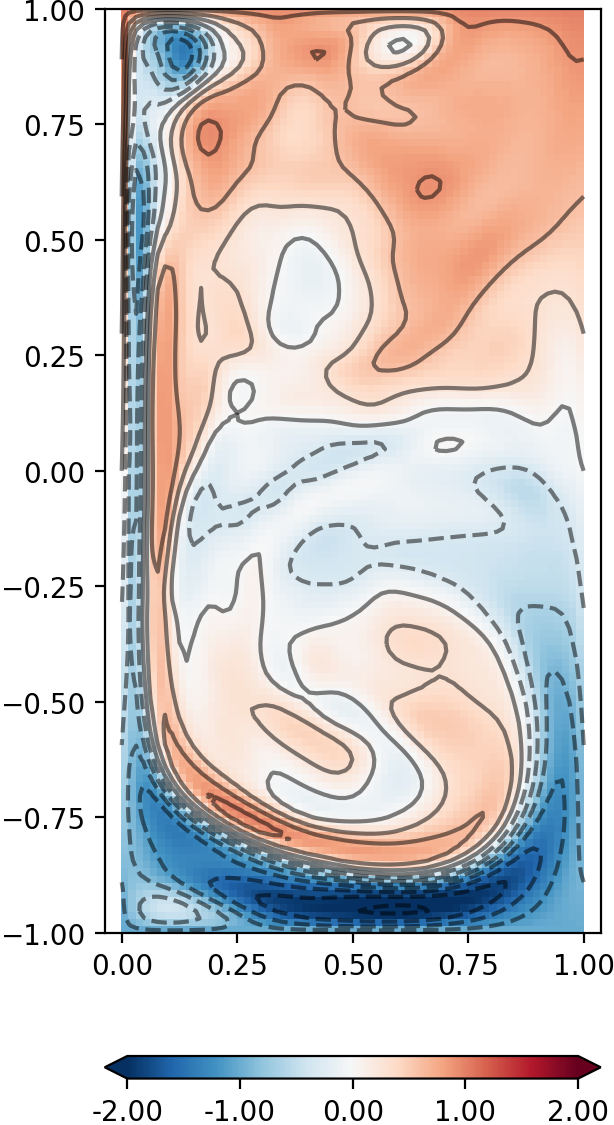}
        \label{fig:reconstr_Gaussian_1}
    }
  \caption{Comparison between the reference potential vorticity of a test ensemble member at $t=0.07$
             (corresponding to $3.7$ days of prediction)
              (a) and the predicted reconstruction using the Gaussian kernel (${\ell_G}=1$) (c). The best possible reconstruction using a linear combination of ensemble members with respect to the $L^2$ norm is shown in (b).}
    \label{fig:reconstruction_Gaussian}
\end{figure}

Reconstruction errors are quantified by averaging the square of the $L^2(\Omega_x)$-norm over the whole test ensemble composed of $100$ new members (i.e. not used in the training phase).
The standard deviations of the error are also displayed on this figure.
The time evolution of this error is presented in figure~\ref{fig:errors-1}. It is compared 
to the projection error defined as the minimal linear combinations of ensemble members \AM{with a constant-in-time coefficient} w.r.t. $L^2(\Omega_x)$-norm. The reconstruction error performed by the ensemble average is also shown. This error is obviously high for small times, but estimations are difficult to be more precise than the ensemble average estimation once the predictability time of the system has passed. 
The reconstructions performed with $k_E$, grow and overtake the ensemble average reconstruction around $t=0.2$ (i.e. $18.4$ days).
On this basis, the use of $k_G$ (${\ell_G}=1.0$) leads to better results.
Indeed, despite a slight reconstruction penalty at the initial time due to a stronger smoothing (see figure~\ref{fig:reconstruction_Gaussian}), the reconstruction errors become quickly significantly lower, and the error never reaches the ensemble anomaly error.
This small loss of performance in $L^2(\Omega_x)$-norm at the initial time allows to have a gain in robustness in terms of long-time prediction.  \AD{Let us note that the lengthscale parameter $\ell_G$ is sensitive. In this work, we fixed it a priori. However, techniques that aim at learning such a parameter have been proposed in \cite{Hamzi-Phys-D-I-21,Owhadi-JCP-19}.}
\begin{figure}
    \centering
    \includegraphics[width=0.5\textwidth]{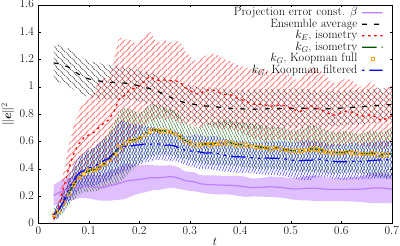}
    \caption{Mean estimation error using an empirical covariance kernel (short red dashes), a Gaussian kernel (green dash-dots), the reconstructed Gaussian kernel using all Koopman eigenfunctions (orange squares) and the Lyapunov filtered kernel (blue dot-dashes). ``Ensemble average'' (long black dashes) is the estimation error of the ensemble average. The projection error (purple solid line) corresponds to  the optimal \GT{constant-in-time} linear combination of ensemble members w.r.t. the  $L^2(\Omega_x)$-norm. The vertical width of the stripes regions over the mean error is equal to the standard deviation of the error. One dimensionless time $t$ corresponds to $92$ days. Errors are averaged over the $100$ members of the test ensemble.}
    \label{fig:errors-1}
\end{figure}

The reconstructions based on \eqref{eq:kernel-inversion} are valid for all times.
Koopman isometry warrants the use of the same linear combination for the whole trajectory. Nevertheless, in practice the ensemble are  of limited (small) size, and the eigenpairs of the infinitesimal generator close to zero (representing almost stationary functions) are difficult to estimate accurately. Modal predictability time associated to Lyapunov exponents \eqref{norm-delta-x} can be used to robustify the forecast of new initial conditions. The kernel can be expanded in the Koopman eigenfunctions basis as
\begin{equation}
        k(\XXi,\,\XXj)=\sum_{\AD{\ell}=1}^\infty \psi_\ell^t(\XXi)\overline{\psi_\ell^t}(\XXj)
                    =\sum_{\AD{\ell}=1}^\infty \psi_\ell^0(\XX_0^{(i)})\overline{\psi_\ell^0}(\XX_0^{(j)}).
    \label{eq:kernel-expansion}
\end{equation}
With the numerical procedure proposed here, the $p$ evaluated Koopman eigenfunctions allow to perform an exact reconstruction of the kernel.
An appealing potential of equation~\eqref{eq:kernel-expansion} for faithful predictions consists in filtering the kernel by switching-off components that have passed their predictability time.
We set this time from the finite-time Lyapunov exponents \eqref{norm-delta-x}  as: $T_{\ell}=|\lambda_{\ell}|^{-1}\ln C$, with, $C=10^3$.
The filtered kernel $\widetilde{k}(\XXi,\,\XXj)$ and its associated $p\times p$ kernel matrix $\widetilde{K}$ allows us to define  reconstructions, with the modified linear coefficients $\widetilde{\bm\beta}=\widetilde{K}^{+}\widetilde{K}\bm \beta $.
The superscript $\bullet^{+}$ stands for the Moore--Penrose pseudo-inverse.
It can be remarked that the operator $\widetilde{K}^+\widetilde{K}$ is a projector performing filtering directly on the vector of coefficients $\bm{\beta}$.

In figure~\ref{fig:errors-1}, it is assessed that the kernel is exactly reconstructed by the Koopman eigenfunctions 
(superimposition of the green dash-dots curve and the orange squares).
When the kernel is Lyapunov filtered, the reconstruction almost fits at short time the previous predictions (superimposition of the blue dot-dashes curve, the green dash-dots curve and the orange squares).
For longer time it significantly  reduces the error and clearly separates the spread of the ensemble from the ensemble average estimation. 
For very long times (not shown in the figure), filtering acts on all modes, which reduces to the ensemble average estimation.
This strategy takes advantage of the modal predictability time. Slow eigenfunctions are used for long time predictions, and fast ones only for short time forecasts. As a final remark, we note that very fast modes (with large Koopman eigenvalues modulus) does not seem to contribute significantly to the reconstruction, since at short times the filtered and full reconstructions match. 

\subsection{Data assimilation}
A possible application of the proposed procedure is ensemble-based data assimilation.
In the context of satellite data, ensemble optimal interpolation (EnOI) is widely used to reconstruct spatial fields \citep{traon1998,ducet2000}. It is for instance employed in operational gridded altimetry reconstructions \citep{aviso}.
EnOI consists in expressing at a given time the solution as a linear combination of ensemble members, which  minimises an objective cost function. As shown in \cite{ubelmann2016,leguillou2021}, incorporating information from other times together with an adequate dynamical model significantly improves the estimations.
This is obviously at the price of forecasting the whole ensemble by the dynamics.

Koopman isometry provided by \AD{the RKHS family} allows us to incorporate time-series of observations without any new simulation of the dynamics.
The spatial reconstruction is embarrassingly simply given as the constant-in-time best linear combination of ensemble members trajectories.   

Based on the same ensemble as in the previous sections, we propose two demonstration tests considering \textit{i)} a time-series of extremely sparse (two satellite swaths like) observations \textit{ii)} an even sparser case with the same swaths but at a single time only. To that end, we consider a synthetic observation mimicking two along tracks measurements of an orbiting satellite, idealized by two vertical lines. In the first case these lines are assumed to be known every $0.9$ days $\Delta t=0.01$, while in the second case only day $5.5$ is available  $(t=0.06)$.
The observations are perturbed by a centered Gaussian white noise with the value of the selected standard deviation equal to 10\% of the signal root-mean-square. These 2 benchmarks are carried out for each of the 100 members of the test ensemble. As previously, the  members of the test ensemble have not been used in the learning phase. 

In this numerical experiment, we search for the constant in time coefficients $\bm{\beta}$ (thanks to isometry), which minimises
\begin{equation}
    \mathcal{J}=\frac{1}{2T}\int_0^T\|\mathbb{H}(\hat{X}_t)-\mathcal{Y}\|_{R^{-1}}^2\,\mathrm{d}t +\frac{\alpha^2}{2}\|\bm{\beta}{-\bar{\bm{\beta}}}\|^2
\text{ s.t. }\mathbb{H}(\hat{X}_t)=\sum_{i=1}^p\beta_i\mathbb{H}(\XXi),
\end{equation}
where $\mathbb{H}$ is the observation operator associated with the satellite measurements $\mathcal{Y}$, $\hat{X}_t$ is the estimated state{, and $\bar{\beta}=\frac{1}{p}(1,\dots,1)^T$}.
The observation operator is assumed to have its range in the RKHS family and hence provides observables belonging to the RKHS family, which can be interpolated using the \AD{RKHS family} kernels.
It can be noticed that the expression of the constraint is due to kernel interpolation, and not to some linear property of $\mathbb{H}$.
This treatment is similar to the one performed for instance in discrete empirical interpolation (DEIM) \citep{chaturantabut2010}.
In the case of a single observation, the time integral is dropped, and the cost function corresponds to a standard EnOI cost function. The matrix $R$ is the observation covariance matrix, defined by the time average of the empirical covariance, localized by performing a Schur product with matrix $\exp{(-r_{ij}^2/\ell_{loc}^2)}$, where $r_{ij}$ is the distance between measurements $i$ and $j$, and $\ell_{loc}=0.05$.
The regularization parameter has been set to $\alpha=10$.
For this problem we selected the simple empirical kernel, defined from a non localized potential vorticity matrix as feature maps of the ensemble. Working with the isometry property justifies the estimation of constant in time coefficients for the linear combination. This simplistic method could be easily extended by introducing more elaborated kernels, Koopman eigenfunctions representation to project any system's observable and forecast them, as well as the Lyapunov filtering schme presented in section~\eqref{sec:reconstruction_filtering}. 

Figure~\ref{fig:sat} shows the potential vorticity trajectory (panel~a), the single-time observation (panel b), the trajectory reconstruction associated to this single observation (panel c) and the one obtained from the swaths time-series (panel d).
Three consecutive times spaced of $0.02$ (1.8 days) are displayed.
\begin{figure}[htbp]
    \centering
    \includegraphics[width=0.7\textwidth]{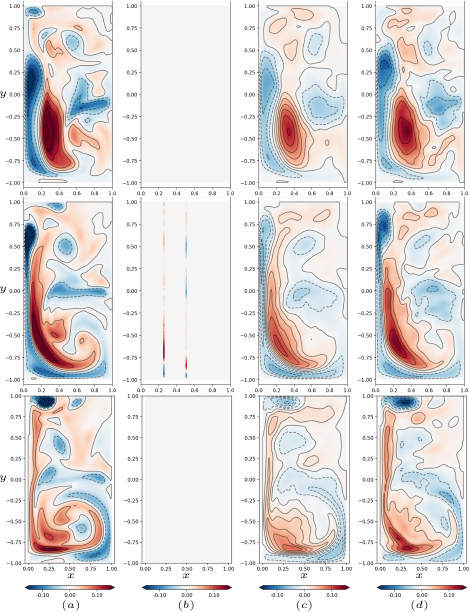}
    \caption{Estimation of potential vorticity at times $t=0.04$ (top), $t=0.06$ (middle) and $t=0.08$ (bottom).
    Panel (a): True anomaly.
    Panel (b): Observation associated with the single-time swath observation.
    Panel (c): Estimation knowing the two satellite swaths only at day 5.5 $(t=0.06)$.
    Panel (d): Estimation knowing the two swaths every 0.9 days $(\Delta t=0.01)$.
    }
    \label{fig:sat}
\end{figure}
We can see that despite extremely sparse measurements, a single observation already provides a fairly good estimation, together with forward and backward in time forecasts based only on the estimated linear combination of the ensemble members (i.e. without any simulation). As expected, a time-series of swaths improves further the performances.

In figure~\ref{fig:errors-2}, the average of the  errors obtained over all $100$ members of the test ensemble are plotted.
\begin{figure}[htbp]
    \centering
    \includegraphics[width=0.5\textwidth]{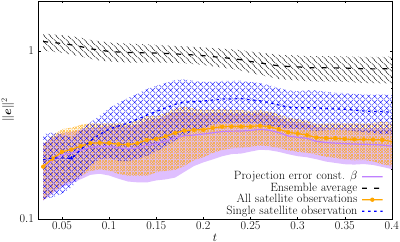}
    \caption{Estimation errors averaged over the whole $100$ members of the test ensemble. The best reconstruction error with constant coefficients (\GT{purple}), estimation by the ensemble average (black), estimation with the time-series of satellite swaths observations (yellow) and only at day $5.5$ $(t=0.06)$ (blue).}
    \label{fig:errors-2}
\end{figure}
As a reference, the best possible reconstruction with respect to the $L^2$ norm with constant in time linear coefficients $\bm{\beta}$ is shown, as well as the estimation by the ensemble average.
The error performed  with the time-series of satellite observations is almost optimal in terms of average and standard deviation.
We note that for the second test-case, near the observation time, the error is lower than  the one obtained for the whole time-series. Compared to the time series of observations it may be interpreted as an over-fitting at the observation time since only a single observation is used to infer the ensemble members' coefficients. As expected, this error grows, but reasonably, as the estimation is propagated forward and backward in time. \AD{As a final remark, let us emphasize that considering time series of the sparse observables (in space and time) such as the satellite swaths considered in this study, which constitute typical altimetric observations of the ocean surface, would require an extremely long series to reconstruct the Koopman eigenpairs (see \cite{Zhen-et-al-2022}) at the scale of a large basin such as the North Atlantic basin. For geophysical flows this can be highly problematic in the nowadays state of the fast climate change we are observing. The proposed technique, based on an ensemble of realizations instead of time series, constitutes a promising  alternative to face these difficulties. We assessed here the technique with an idealistic dynamics model and more realistic models need to be tested. This will be the subject of future works.}

\section{Concluding remarks}
We proposed a theoretical setting to estimate the Koopman eigen-pairs of compact dynamical systems in embedding the system in a RKHS family. Beyond the Koopman operator spectral representation, we showed that \AD{the RKHS family} provides access to finite time  Lyapunov exponents, as well as to the tangent linear dynamics. These features enabled us to propose very simple strategies for data assimilation of very sparse observations in time and space. Obviously more efficient data assimilation techniques could be used instead of the simple regression strategy explored here. As  \AD{the RKHS family} boils down to a linear setting it should be perfectly adapted to ensemble Kalman filters. Extension of classical ensemble data assimilation filters on \AD{the RKHS family} are currently in development and will be the subject of future publications.  The technique lies on the sequence of time-evolving RKHS to represent the ensemble dynamics and on the superposition principle attached to \AD{the RKHS family}. This study provides a theorem-based setting to estimate the Koopman eigenvalues and eigenfunctions, as well as exact ensemble-based representation of the tangent linear dynamics and of  the flow Lyapunov exponents. 
This work relies on the assumption of an invertible deterministic dynamical system defined on a compact manifold.
In future studies, we will focus  on the generalisation to stochastic fluid flow dynamics with transport noise models \citep{Memin14,Bauer-et-al-JPO-20, Chapron-18, Debussche-2023} and the introduction of more sophisticated ensemble data assimilation techniques \cite{Dufee-STUOD-22,Dufee-QJRMS-22}.

\section*{Acknowledgments}
The authors acknowledge the support of the ERC project 856408-STUOD and the French National program LEFE (Les Enveloppes Fluides et l'Environnement).

\appendix
\section{Reproducing kernel Hilbert spaces}
\label{RKHS}

In the following appendices we provide complementary elements for the detailed proofs of the results listed in the paper. We first recall some classical theorems on RKHS and present their adaptations to our functional setting. In particular we present useful results on differentiability of RKHS functions and reproducing kernels.
In the following $L^2(\Omega_x)$ stands for the Hilbert space of Lebesgue measurable $\mathbb{R}^d$-valued of square-integrable functions on the set $\Omega_x\subset \mathbb{R}^{p}$ $p=2,3$ endowed at a fixed time, $t$, with the inner product
\[
\langle f \, , g \rangle_{L^{2}(\Omega_x)} = \int_{\Omega_x} f(y) \, \,g(y) \;\mathrm{d}y.
\]
\noindent

Let $E\subset L^2(\Omega_x)$ denote a  compact subset equipped with the norm of  $L^2(\Omega_x)$, by the Aronszajn Theorem, \cite{Aronszajn-50}, the reproducing kernel Hilbert spaces (RKHS) of $\mathbb{C}$-valued functions defined on  $E$ are uniquely characterized through a reproducing kernel $k : E\times E \mapsto \comp$. This kernel is a Hermitian, positive definite  function. In the following we will {assume} furthermore it is {$C^{{(}1,1{)}}(E\times E)$}, which means here it is twice differentiable with respect to its two variables. The {(Gateaux)} derivative with respect to the first  and second variable  in the direction $v\in E$ are denoted $\partial_v^{\tiny{(1)}} k$, and  $\partial_v^{\tiny{(2)}} k$, respectively.  The reproducing kernel defines an evaluation function $k(\bcdot,x) \in \HS$ with a reproducing property (Def. \ref{def1}): 
\[\forall x \in E, \forall f \in \HS, \; \bigl\langle f ,  k(\, \bcdot\, , \, x) \bigr\rangle_{\HS} = f(x).\]

The RKHS $\HS$ corresponds to a space \AM{of functions} with a strong uniform convergence property. The elements of $\HS$ are in addition continuous functions and we have for all $f\in \HS$ that
\begin{equation}\label{NormeRKHSnormeUniforme}
    \forall \, x \in E \qquad \left| f(x) \right| \, \leq \, \sup_{y\in E} k(y,y)^{1/2} \; \| f \|_{_\HS}. 
\end{equation}

\subsection{Functional description of $\HS$}
\label{sec1-RKHS}
We recall in the following the Mercer's theorem which gives a functional description of the RKHS $\HS$ and of its associated reproducing kernel $k$.\\
\noindent
 Let $L^{2}_{\mathbb{C}}(E , \nu)$ be the space of measurable complex valued functions defined on the compact set E for which the square modulus is integrable with respect to measure $\nu$, $L^{2}_{\mathbb{C}}(E , \nu)$  is a Hilbert space equipped with the inner product $\langle \bcdot \, , \, \bcdot \rangle_{{L^{2}_{\mathbb{C}}(E , \nu)}}$ given for $f$ and $g \in L^{2}_{\mathbb{C}}(E , \nu)$ by 
\[
\langle f \, , \, g \rangle_{{L^{2}_{\mathbb{C}}(E , \nu)}} := \int_{E} f(y) \, \overline{\,g(y)\,} \; \nu(\mathrm{d}y),
\]
so that it is linear in its first argument and antilinear in its second. Let us first define ${\cal L}_{k}: L_{\mathbb{C}}^{2}(E,\nu)\mapsto L_{\mathbb{C}}^{2}(E,\nu)$ the integral operator with kernel ${ k\in {C^{{(}1,1{)}}}(E\times E)}$ as: 
\begin{equation}
\bigl( {\cal L}_{k} f\bigr)(x)= \int_E k(x,y) f(y) \,\nu(\mathrm{d}y).
\end{equation}
We give below a useful remark relating the integral operator $\mathcal{L}_k$ and the inner product $\langle \bcdot \, , \, \bcdot \rangle_{L^{2}_{\comp}(E , \nu)}$.
\begin{remark}[Links between $\mathcal{L}_k$ and $\langle \, \bcdot \, , \, \bcdot \, \rangle_{L^{2}_{\comp}(E , \nu)}$]\label{ReecritureOpIntegral} As $k$ is Hermitian, it can be noticed that for all $x\in E$,  $\bigl( {\cal L}_{k} f\bigr)(x)= \int_E f(y) \, \overline{k(y , x)} \,\nu(\mathrm{d}y)$, and we have thus $\bigl( {\cal L}_{k} f\bigr)(x) = \langle f \, , \, k(\,\bcdot \, , \,  x) \rangle_{L^{2}_{\comp}(E , \nu)}$.
\end{remark}
This kernel integral operator can be shown to be well-defined and compact (by Arzel\`{a} Ascoli theorem). It is immediate to see that as $k$ is positive definite  and Hermitian, ${\cal L}_{k}$ is positive and self-adjoint on $L_{\mathbb{C}}^{2}(E,\nu)$. By the spectral theorem of compact self-adjoint operators, there exists an orthonormal basis $(\varphi_i)_i$ of $L_{\mathbb{C}}^{2}(E,\nu)$ consisting of eigenvectors of ${\cal L}_{k}$ with  corresponding decreasing  real eigenvalues $(\mu_i)_i$ : $\mu_i \geq \mu_{i+1}\geq 0$ for all integers $i$. These results are proved in \cite{Cucker-Smale-01}. We consider in the following that  $\mu_i \neq 0$ for all $i$ and $\mathcal{L}_k$ is injective. This denseness assumption could be relaxed (see remark \ref{ValPropreNulles}) but simplifies {significantly} the presentation. The eigenpairs $(\mu_i , \varphi_i)_{i\in\mathbb{N}}$ enables to formulate the following theorems (with proofs available for instance in  \cite{Cucker-Smale-01} and \cite{Konig86}). Elaborating on remark \ref{ReecritureOpIntegral} and the kernel eigenfunctions (often referred to as Mercer's eigenfunctions due to theorem \ref{MercerDecomposition})  we get immediately relations between the eigenfunction evaluation and the $L^2(E,\nu)$ inner product.
\begin{remark}[Links between $(\varphi_i)_{i\geq 0}$ and $\langle \bcdot \, , \, \bcdot \rangle_{L^{2}_\comp(E , \nu)}$]\label{CalculOpIntegral1}
For all $x\in E$, by {R}emark \ref{ReecritureOpIntegral}, we note that 
\begin{align*}
&\langle \varphi_i\, , \, k(\, \bcdot \, , \, x )  \rangle_{L^{2}_{\comp}(E , \nu)} \; = \; \mathcal{L}_k(\varphi_i)(x) \; =  \; \mu_i \, \varphi_i(x) \qquad \text{and}\\&\langle k(\, \bcdot \, , \, x ) \, , \, \varphi_i \rangle_{L^{2}_{\comp}(E , \nu)} \; = \; \overline{\langle \varphi_i\, , \, k(\, \bcdot \, , \, x )  \rangle}_{L^{2}_{\comp}(E , \nu)} \; = \; \mu_i \, \overline{\varphi_i(x)}. 
\end{align*}
\end{remark}

\begin{remark} As $k\in {C^{{(}1,1{)}}}(E\times E)$, the range of $\mathcal{L}_k$ is included in $C^{1}(E)$ and the eigenfunctions of $\mathcal{L}_k$ are $C^{1}(E)$.
\end{remark}

\begin{theorem}[Mercer decomposition]\label{MercerDecomposition}
For all $x$ and $y\in E$, $\displaystyle k(x,y) = \sum_{i=0}^{\infty} \mu_i \, \varphi_{i}(x) \; \overline{\varphi_{i}(y)}$ where the convergence is absolute (for each $x , y \in E \times E$) and uniform on $E \times E$. 
\end{theorem}
In particular, the series $\sum \mu_{i}$ is convergent as 
\begin{equation}
\label{SommeMui}
\displaystyle
\int_E k(y,y)\nu(\dif y) = \int_E \sum_{i=0}^{\infty} \mu_{i} \varphi_{i}(y) \; \overline{\varphi_{i}(y)} \nu(\dif y)= \sum_{i=0}^{\infty} \mu_{i}  < \infty .
\end{equation}

\begin{theorem}[Mercer representation of RKHS]\label{CaracterisationRKHS}
The RKHS $\left(\HS, \langle \, \bcdot \, , \, \bcdot \, \rangle_{_\HS} , \|\bcdot\|_{\HS} \right) $ is the Hilbert space characterized by 
\[
\displaystyle \HS = \bigl\{ f\in L_{\comp}^{2}(E,\nu), \; f=\sum_{i=0}^{\infty} a_i \varphi_i \; : \; \sum_{i} \dfrac{\, |a_i|^{2}}{\mu_{i}} < \infty\bigr\},  
\]
\[
\bigl\langle f,g \bigr\rangle_{\HS} = \sum_{i=0}^\infty \frac{1}{\mu_i} \,\langle f \, , \, \varphi_i\rangle_{L^2_{\comp}(E , \nu)} \,  \overline{\langle g \, , \, \varphi_i\rangle}_{L^2_{\comp}(E , \nu)} \quad
\] \[\text{and}\quad \|f\|^2_{_\HS} = \sum_{i=0}^{\infty} \frac{|\langle f \, , \, \varphi_i\rangle_{L^2_{\comp}(E , \nu)}|^2}{\mu_i} .
\]
\end{theorem} 
The inner product $\langle \, \bcdot \, , \, \bcdot \, \rangle_{_\HS}$ is sesqui-linear (linear in its first argument and antilinear in its second). The proof of Mercer's theorem warrants that the linear space spanned by $\{ k(\,\bcdot\, , x) : x\in E\}$ is dense in $L^{2}_{\mathbb{C}}(E, \nu)$.
If $f= \sum_{i=0}^{\infty} a_i \varphi_i  \in \HS$ then the series converges absolutely and uniformly to $f$.\\

\noindent
The following remark shows that the feature map's conjugate as well as its real and imaginary parts belong to $\HS$.
\begin{remark}[Consequence of the Mercer representation]\label{ConseqMercer}
For all $y\in E$, the functions $\overline{\, k(\, \bcdot \, , \, y)}$, \, $Re[k(\, \bcdot \, , \, y)]$ and $Im[k(\, \bcdot \, , \, y)]$ belong to $\HS$. 
By {T}heorems \ref{CaracterisationRKHS} and \ref{MercerDecomposition}, we have indeed that 
\begin{equation}
\sum_{i=0}^{\infty} \dfrac{1}{\mu_i} \, | \langle \overline{\, k(\, \bcdot \, , \, y) \,} \, , \, \varphi_i \rangle_{L^2_{\comp}(E , \nu)}|^2 \; = \;\\ \sum_{i=0}^{\infty} \dfrac{1}{\mu_i}  \, | \, \mathcal{L}_{k} \varphi_i(y) |^2 = \sum_{i=0}^{\infty} \mu_i  \, | \varphi_i(y) |^2 <\infty,
\end{equation}
and $\overline{\, k(\, \bcdot \, , \, y)} \in \HS$. By operation, $Re[k(\, \bcdot \, , \, y)]$ and $Im[k(\, \bcdot \, , \, y)]$ belong thus to $\HS$.
\end{remark}

\begin{remark}[Comparison of the norms $\HS$ and $L_{\comp}^{2}(E , \nu)$]\label{ComparisonNormL2HS} For all $f\in \HS$, by {T}heorem \ref{CaracterisationRKHS} we have for all $x\in E$, the inequality
\[
\displaystyle  f(x)\leq|f(x)| =\left|\sum_{i=0}^{\infty} a_i \varphi_{i}(x)\right| \leq \left(\sum_{i=0}^{\infty} \dfrac{\, |a_i |^{2} \, }{\mu_i} \right)^{\alf} \; \left( \sum_{i=0}^{\infty} \mu_i |\varphi_{i}(x)|^{2} \right)^{\alf} ,
\]
with theorem \ref{MercerDecomposition}  we obtain
\begin{equation}\label{eq-ComparisonNormL2HS}
 \displaystyle  \| f \|_{{L_\comp^{2}(E, \nu)}} \; \leq \; \left( \int_{E} | k(x,x)|^{2} \, \nu(\mathrm{d}x) \right)^{\alf} \, \| f \|_{_\HS} .
\end{equation}
\end{remark}
\begin{remark}\label{ContLalfH}
For all $f\in L^{2}_{\mathbb{C}}(E,\nu)$ , $\mathcal{L}_k(f)$ belongs to $\HS$, in fact we have 
\[
\sum_{i=0}^{\infty} \frac{1}{\mu_i} \left| \langle \mathcal{L}_k(f) \, ,\, \varphi_i \rangle_{L^{2}_{\mathbb{C}}(E,\nu)} \right|^{2} = \sum_{i=0}^{\infty} \mu_i \left| \langle f \, ,\, \varphi_i \rangle_{L^{2}_{\mathbb{C}}(E,\nu)} \right|^{2} \leq \left(\sum_{i=0}^{\infty} \mu_i\, \right) \, \|f\|^{2}_{L^{2}_{\mathbb{C}}(E,\nu)} < \infty. 
\]
Moreover, by  {R}emark $\ref{ComparisonNormL2HS}$ the mapping $\mathcal{L}_{k} : \HS \to \HS$ is continuous : for all $f\in \HS$ we have  
\[
\| \mathcal{L}_{k}(f) \|^{2}_{\HS}  \leq \left(\sum_{i=0}^{\infty} \mu_i\, \right) \; \|f\|^{2}_{L^{2}_{\mathbb{C}}(E,\nu)} \leq C \, \| f \|^{2}_{\HS} .
\]
\end{remark}
 With a $C^{\AD{(1,1)}}(E\times E)$ regularity condition on the kernel, we will prove that the RKHS $\HS$ is a subspace of smoother functions than $C^{0}(E , \comp)$ (theorem \ref{RKHSC1}). In the following appendix we discuss differentiability property of the feature maps.

\subsection{Differentiability of feature maps and RKHS functions $\HS$}
\label{sec2-RKHS}
In this appendix we extend results given in \cite{Simon-Gabriel-Scholkopf-18,Steinward-Christmann, Zhou-2008} for $E\subset \real^d$ to the case where $E$ is a \AM{compact subset of a function vector space}. In the following, \AM{aligned with the assumptions enunciated in the paper's preamble,} $E$ designates a compact set included in $L^2(\Omega_x, \mathbb{R}^{d})$, the space of measurable functions that are square integrable on $\Omega_x$, \AD{and like \cite{Zhou-2008}, to enable differentiability of the RKHS function, we will assume that $E$ is the closure of its nonempty interior $\overset{\circ}{E}$}. \AD{We assume also that  $\nu$ is a finite measure on $E$, (i.e. $\nu(E)<\infty$). }
Denoting by $\partial^{(2)}_{u}\, k(\, \bcdot \, ,y)$ the kernel functional  derivative in the direction $u$ with respect to the second variable $y$ -- while the upper-script ${\partial_u^{(1)}}$ indicates a directional derivative in the first variable --, we have first the following lemma on the feature maps derivative. \\

\begin{lemma}[Feature maps differentiability]\label{NoyauDerivRKHS}
The kernel derivatives belong to $\HS$. More precisely, let $k\in {C^{(1,1)}}(E\times E)$, then for all $y\in E$ and ${u}\in E$, we have $\partial^{(2)}_{{u}}\, k(\, \bcdot \, ,y)\in \HS$.
\end{lemma}
\begin{proof}
Let $y \in \AD{\overset{\circ}{E}}$ and ${u}\in E$. Let $(h_n)_{n\geq 0}$ be a sequence of $\mathbb{C}\backslash\lbrace{0}\rbrace$ which converges to $0$ and such that for all $n$, $y+ h_n {u}\in \AD{\overset{\circ}{E}}$ (since $\AD{\overset{\circ}{E}}$ is an open set). For all integers $n$, let 
\[
d_n(\, \bcdot \, , \, y ) := k( \, \bcdot \, , \, y+h_n {u}) - k( \, \bcdot \, , \, y )\text{ and } \;
t_n(\, \bcdot \, , \, y \,) := \dfrac{\, d_n( \, \bcdot \, , \, y ) \, }{ h_n}.
\]
We first notice that  $(t_n(
\bcdot\, , \, y ))_{n\geq0}$ converges pointwise to $ {\partial^{(2)}_{{u}}}\, k (\, \bcdot\, , \, y)$ since $k$ is ${C^{(1,1)}}(E \times E )$. Then, we now show that $(t_n(\, \bcdot \, , \, y ))_n$ is a Cauchy sequence of $(\HS , \| \bcdot \|_{_\HS})$.  For all integers $p$ and $q$, we have 
\begin{equation}\label{IntermCauchy}
\| t_p( \, \bcdot \, , \, y) - t_q( \, \bcdot \, , \, y  ) \|^{2}_{_\HS} = \| t_p( \, \bcdot\, , \, y  ) \|^{2}_{_\HS} +  \, \|  t_q( \, \bcdot \, , \, y  ) \|^{2}_{_\HS} - 2 Re[\langle \, t_p( \,  \bcdot \, , \, y ) \; , \; t_q( \, \bcdot \,  , \, y ) \, \rangle_{_\HS}] .
\end{equation}
By the reproducing property of $\HS$ and as $k$ is ${C^{(1,1)}}(E \times E)$, we have for all $x\in E$ and for some $p$ and $q$ large enough, 
\begin{align*}
\langle \, t_p( \, \bcdot \, , \, y ) \; , \; t_q( \, \bcdot \, , \, x) \, \rangle_{_\HS} 
&= \dfrac{1}{h_p\, h_q}\langle \, k(\,\bcdot\, , \, y+h_{p} u) - k(\,\bcdot\, , \, y)\,  \; , \; k( \, \bcdot \, , \, x+h_{q} u) -k(\,\bcdot\, , x) \, \rangle_{_\HS}\\
&=\dfrac{1}{h_p\, h_q} \, \left(\,
k(x+h_{q}u\, , \, y+h_{p} u) - k(x\,,\,y+h_{p} u) -  k(x+h_{q}u\, , \, y) +k(x,y)\,  \right)\\
&=\dfrac{1}{h_p\, h_q} \, \left(\,
\partial^{(1)}_{u} k(x\, , \, y+h_{p} u)\, \times \, h_q - \partial^{(1)}_u k(x\,,\,y) \times h_q  + o(h_q)\,  \right)\\
&=\partial^{(2)}_{u}\, \partial^{(1)}_{u} k(x\, , \, y) + o_{p,q}(1)
\end{align*}
We have also for some $p$ and $q$ large enough (with $y=x$) and $\varepsilon >0$,
\[
 \left| \bigl\langle \, t_p( \, \bcdot \, , \, y ) \; , \; t_q( \, \bcdot \, , \, y) \, \bigr\rangle_{_\HS} - \partial^{(2)}_{{u}} \partial^{(1)}_{{u}} k(y \, , \, y) \right| \leq \varepsilon,
\]
and by \eqref{IntermCauchy} we obtain hence $\| t_p( \, \bcdot \, , \, y) - t_q( \, \bcdot \, , \, y  ) \|^{2}_{_\HS} \leq 4 \varepsilon$, which means $\bigl(t_n( \, \bcdot \, , \, y)\bigr)_{\!n}$ is a Cauchy sequence in $\HS$. The sequence $\bigl(t_n( \, \bcdot \, , \, y)\bigr)_{\!n}$ converges to a function $t(\, \bcdot \, , \, y )$ in $\HS$ and by \eqref{NormeRKHSnormeUniforme} the convergence is pointwise. Consequently, we obtain also $ \partial^{(2)}_{{u}}\, k (\, \bcdot \, , \, y) = t(\, \bcdot\, , \, y ) \in \HS$ \AD{ for $y \in \overset{\circ}{E}$}.

\AD{Let $y \in E \backslash \overset{\circ}{E}$. There exists a sequence $(y^{i})$ of $\overset{\circ}{E}$ converging to $y$. For all integers $i$ and $j$, we have, }
\[
\| \partial^{(2)}_{{u}}\, k (\, \bcdot \, , \, y^{i}) - \partial^{(2)}_{{u}}\, k (\, \bcdot \, , \, y^j) \|_{\HS} = \partial^{(1)}_{{u}}\partial^{(2)}_{{u}}\, k (\, y^{i} \, , \, y^{i}) + \partial^{(1)}_{{u}} \partial^{(2)}_{{u}}\, k (\, y^{j} \, , \, y^{j}) - 2 Re( \partial^{(1)}_{{u}} \partial^{(2)}_{{u}}\, k (\, y^{i} \, , \, y^{j}) .
\]
\AD{As $(y^{i})$ is a Cauchy sequence of $E$ and $\partial^{(1)}_{{u}} \partial^{(2)}_{{u}}\, k$ is continuous on $E\times E$, we infer that $(\partial^{(2)}_{{u}}\, k (\, \bcdot \, , \, y^{i}))_i$ is a Cauchy sequence of $\HS$. In particular, $(\partial^{(2)}_{{u}}\, k (\, \bcdot \, , \, y^{i}))_i$ converge to a function $h$. For all $x\in E$, we have }
\[
h(x)= \langle h \, , \, k(\,\bcdot\, , x)\rangle_\HS = \lim_{i\to \infty} \langle \partial^{(2)}_{{u}}\, k (\, \bcdot \, , \, y^{i}) \, , \, k(\,\bcdot\, , x)\rangle_\HS = \partial^{(2)}_{{u}}\, k (\, x \, , \, y^{i}).
\]
\AD{We obtain thus that $h=\partial^{(2)}_{{u}}\, k (\bcdot \, , \, y) \in \HS$ for all $y\in E \backslash \overset{\circ}{E}$, and consequently, $\partial^{(2)}_{{u}}\, k (\bcdot \, , \, y) \in \HS$ for all $y\in E$.} 
\cqfd
\end{proof}

As  the sequence $\bigl(t_n(\cdot\,, x )\bigr)_n$ converges to ${\partial^{(2)}_{{u}}}\, k(\,\cdot\,, x)$ in $\HS$, we obtain also
\begin{equation}\label{Deriv}
\left\langle f \, , \, {\partial^{(2)}_{{u}}}\, k(\,\cdot\, , \, x) \right\rangle_{\HS} = \lim_{p\to \infty} \left\langle f \, , \, \dfrac{k(\,\cdot\, , \, x +h_p) - k(\,\cdot\, , \, x)}{h_p}\, \right\rangle_{\HS} = \lim_{p\to\infty} \dfrac{f(x+h_pu) - f(x)}{h_p} = \partial_u f(x) .
\end{equation}

We provide in theorem \ref{RKHSC1} an expression of the directional derivative along direction $u \in E$ of an RKHS function $f$. To that end we first introduce useful remarks together with an intermediate lemma.

As $k$ is ${C^{(1,1)}}(E\times E)$,  the Mercer eigenfunctions $\varphi_{i}$ are $C^{1}(E,\comp)$. Similarly to remark \ref{CalculOpIntegral1}, we get a computational rule relating the kernel operator derivatives and the inner product $\langle \bcdot \, , \, \bcdot \rangle_{L^{2}_{\comp}(E , \nu)}$. For all $y\in E$ and ${u}\in E$, we note that
\[
{\langle \partial^{(2)}_{{u}}\, k(\cdot\,,\,y) , \varphi_i \rangle_{{L_{\comp}^{2}(E, \nu)}} = \int_E \partial^{(2)}_{{u}}\, k(x\,,\,y) \, \overline{\varphi_i(x)}\, \, \nu(\mathrm{d}x) = \partial_u \int_E  \overline{k(y\,,\,x)} \,\, \overline{\varphi_i(x)}\, \, \nu(\mathrm{d}x) = \partial_u \overline{{\cal L}_{k} (\varphi_i) (y)} = \mu_i \, \partial_u \overline{\varphi_i(y)}}
{.}
\]
We obtain for all $y\in E$ and ${u}\in E$ that
\begin{equation}\label{AA}
\langle \partial^{(2)}_{{u}}\, k( \,\bcdot\, , \, y ) \, , \, \varphi_{i} \rangle_{{L_{\comp}^{2}(E, \nu)}} = \mu_i \; \partial_{{u}}  \overline{\varphi_i(y)} \qquad \text{and}\qquad  \langle \varphi_i  \, , \, \partial^{(2)}_{{u}}\, k( \,\bcdot\, , \, y ) \rangle_{L_{\comp}^{2}(E, \nu)} = \mu_i \; \partial_{{u}} \varphi_i(y).
\end{equation}
These remarks allow us to prove the next lemma. 

\begin{lemma}\label{Lemma1}
The series $\displaystyle\sum \mu_i \, |\partial_{{u}} \varphi_i(y)|^{2}$ converges uniformly on $E$.
\end{lemma}
\begin{proof}
The functions $\mu_i \, \varphi_i$ are $C^{1}(E, \comp)$.  By definition of $\|\bcdot\|_{_\HS}$ and the previous remark, we have
\[
\|{\partial^{(2)}_{{u}}}\, k( \, \bcdot \, , y) \|^2_{_\HS} 
= \sum_{i=0}^{\infty} \frac{|\langle {\partial^{(2)}_{{u}}}\, k( \, \bcdot \, , y) \, , \, \varphi_i \rangle_{_{L_{\comp}^{2}(E, \nu)}} |^2}{\mu_i} 
= \sum_{i=0}^{\infty} \mu_i \, |\partial_{{u}} \varphi_i(y)|^{2}. 
\]
The series $\sum \mu_i \, |\partial_{{u}} \varphi_i(y)|^{2}$ converges pointwise to the continuous function $\|{\partial^{(2)}_{{u}}}\, k( \, \bcdot \, , y) \|^2_{_\HS}$. By  Dini's theorem (the sequence of partial sums is monotonically increasing as $\mu_i$ is positive), the convergence is also uniform.\cqfd
\end{proof}

\begin{theorem}[Mercer's decomposition of derivatives]\label{RKHSC1}
Let $u\in E$. For all $\displaystyle f=\sum_{i=0}^{\infty} a_i \varphi_i \in \HS$ (with $a_i = \langle f \, , \, \varphi_i \rangle_{L^2_{\mathbb{C}}(E, \nu)}$), the  directional derivative of the function $f$ in the direction $u$ exists and we have
\[
\partial_{{u}} f = \sum_{i=0}^{\infty} a_i \, \partial_{{u}}(\varphi_i),
\]
where the convergence is uniform on $E$.
\end{theorem}
\begin{proof}[Theorem \ref{RKHSC1}]
Let $\displaystyle f=\sum_{i=0}^{\infty} a_i \varphi_i \in \HS$. For all $i$, the function $a_i\, \varphi_i$ belongs to $C^{1}(E, \comp)$ and the series $\sum_i a_i \varphi_i$ converges pointwise to $f$ on $E$. As for the convergence of the derivative series, we have, by Cauchy-Schwarz inequality, that for all  $n \in \mathbb{N}$ and $y \in E$,
\[ 
\left| \displaystyle\sum_{i=n}^{\infty} a_i \, \partial_{{u}} (\varphi_i)(y) \right| \leq  \left( \, \displaystyle\sum_{i=n}^{\infty} \dfrac{\, |a_i|^{2} \, }{\mu_i} \right)^{1/2} \, \left( \, \displaystyle \sum_{i=n}^{\infty}  \mu_i  \left|\, \partial_{{u}} (\varphi_i)(y) \right|^{2} \right)^{1/2}\, .
\]
By {L}emma \ref{Lemma1}, the term $\sum_{i=n}^{\infty}  \mu_i  \left|\, \partial_{{u}} (\varphi_i)(y) \right|^{2}$ converges uniformly on $E$ to $0$. Since $f\in \HS$, we obtain also that the series $\sum a_i \, \partial_{{u}} (\varphi_i)$ converges uniformly on $E$. Then, through term by term differentiation the result is proved.
\cqfd
\end{proof}
We remark immediately that as functions of the RKHS the kernel feature maps $k(\bcdot,x)$ can be differentiated to obtain a term by term  Mercer's type decomposition of the kernel derivatives.
\begin{proposition}[Mercer Kernel derivatives decomposition]\label{DerivNoyau}
For all $x , y\in E$ and for all ${u}, w\in  E$, we have
\begin{equation}
\partial^{(1)}_{{u}} \, \partial^{(2)}_{w} k(x,y) \; = \sum_{i=0}^{\infty} \mu_i \, \partial_{{u}} \varphi_i(x) \, \partial_{w} \overline{\varphi_i(y)} 
\end{equation}
where the convergence is uniform on $E \times E$.
\end{proposition}
\begin{proof}
The functions $(x,y) \mapsto \mu_i \, \varphi_i(x) \, \overline{\varphi_i}(y) $ are ${C^{(1,1)}}(E\times E)$. For all ${u}$ and $w\in E$, we apply theorem \ref{RKHSC1}: the series $\sum \mu_i \, \partial_{{u}} \varphi_i(x) \overline{\varphi_i(y)}$ converges uniformly to ${\partial^{(1)}_{u}} k(x,y)$ and $\sum \mu_i \, \varphi_i(x) \, \partial_{w} \overline{\varphi_i(y)}$ converges uniformly to ${\partial^{(2)}_{w}}  k(x,y)$ on $E \times E$. For all $w\in E$, theorem \ref{RKHSC1} can be applied to the function $x \mapsto \partial^{(2)}_{w} k(x\,,\, y) \in \HS$ and for all ${u}\in E$ and all $x,y \in E$, 
\[
\partial^{(1)}_{{u}} \, \partial^{(2)}_{w} \, k( x\, , \, y ) = \displaystyle \sum_{i=0}^{\infty} a_{i, w}(y) \; \partial_{{u}} \varphi_i(x),
\]
with $a_{i,w}(y):= \langle \partial^{(2)}_{w} k( \,\bcdot\, , \, y ) \, , \, \varphi_{i} \rangle_{{L_{\comp}^{2}(E, \nu)}}$. By $(\ref{AA})$, we obtain that $a_{i,w}(y) = \mu_i \, \partial_{w}\, \overline{\varphi_{i}(y)}$ and 
\[
\partial^{(1)}_{{u}} \, \partial^{(2)}_{w} \; k(x,y) = \sum_{i=0}^{\infty} \mu_i \, \partial_{{u}}\varphi_i(x) \;  \partial_{w} \overline{\varphi_i(y)} .
\]
Finally, by Cauchy-Schwarz, we have that for all $N\in\mathbb{N}$
\[
\left| \sum_{i=0}^{N} \mu_i \, \partial_{{u}} \varphi_i(x) \, \partial_{w} \overline{\varphi_i(y)} \right|^{2} \leq \left( \sum_{i=0}^{N} \mu_i \, |\partial_{{u}} \varphi_i(x)|^{2} \right) \left( \sum_{i=0}^{N} \mu_i \, |\partial_{w} \varphi_i(y) \, |^{2} \right),
\]
and by {L}emma \ref{Lemma1} the convergence of $\sum \mu_i \, \partial_{{u}} \varphi_i(x) \, \partial_{w} \overline{\varphi_i(y)}$ to $\partial^{(1)}_{{u}} \partial^{(2)}_{w} k(x,y)$  is uniform on $E \times E$.\cqfd
\end{proof}

\noindent

These results can  be easily extended to higher order derivatives provided a kernel with enough regularity. They correspond, in  slightly different forms, to results given in \cite{Simon-Gabriel-Scholkopf-18,Steinward-Christmann} for $\real^d$. Note that the existence of such higher order derivatives pertains to differentiability conditions of the Mercer eigenfunctions, $(\varphi_i)_{i \in \mathbb{N}}$, associated with kernel $k$.

By {T}heorem \ref{RKHSC1}, the functions of $\HS$ admit a directional derivative  along any vector of $E$. The $L^{2}_{\comp}(E, \nu)$ norm of the directional derivatives of functions of $\HS$  can then be controlled by the RKHS norm.
\begin{theorem}[Continuity in $L_{\comp}^2(E, \nu)$ of RKHS  derivatives]\label{ContKoopmTHEOREM}
There exists a constant $C>0$ for which: \[ \| \partial_{{u}} f \|_{{L_{\comp}^2(E, \nu)}}\leq C \, \left\| f \right\|_{_{\HS}} \text{ for all } f\in \HS \; \text{and} \; u\in E .
\] 
\end{theorem}
\begin{proof}
Let $u\in E$. We prove the existence of $C>0$ such that $\| \partial_{{u}} f \|_{{L_{\comp}^2(E, \nu)}} \leq C \, \left\| f \right\|_{_{\HS}}$ for all $\displaystyle f= \sum_{i=0}^{\infty} a_i \varphi_{i} \in \HS$.\\
By theorem \ref{RKHSC1} \AD{and Fubini's theorem}, we have 
\[
\partial_{{u}} f = \displaystyle \sum_{i=0}^{\infty} a_i \, \partial_{{u}} \varphi_i = \displaystyle \sum_{i=0}^{\infty} a_i  \, \sum_{\ell=0}^{\infty} \langle \partial_{{u}}(\varphi_i) \, , \, \varphi_{_\ell} \rangle_{_{L^{2}_{\comp}(E, \nu)}} \varphi_{_\ell} =\displaystyle \sum_{\ell=0}^{\infty} \left(\sum_{i=0}^{\infty} a_i \langle \partial_{{u}}(\varphi_i) \, , \, \varphi_{_\ell} \rangle_{_{L^{2}_{\comp}(E, \nu)}} \right)\varphi_{_\ell} \, .
\]
With Cauchy-Schwarz inequality we infer that
\begin{multline*}
\displaystyle \sum_{\ell=0}^{\infty}  \; \left| \sum_{i=0}^{\infty} a_i \langle \partial_{{u}}(\varphi_i) \, , \, \varphi_{\ell} \rangle_{_{L^{2}_{\comp}(E , \nu)}} \right|^{2} 
\leq \displaystyle \sum_{\ell=0}^{\infty} \; \left(\sum_{i=0}^{\infty} \dfrac{\, |a_i|^{2}\,}{\mu_{i}} \right) \, \left( \, \displaystyle \sum_{i=0}^{\infty} \mu_{i} \; \left| \, \langle \partial_{{u}}(\varphi_i) \, , \, \varphi_{\ell} \rangle_{_{L^{2}_{\comp}(E, \nu)}}\right|^{2} \right)
\leq \\ \| f \|^{2}_{_{\HS}} \left( \,\displaystyle \sum_{\ell=0}^{\infty} \displaystyle \sum_{i=0}^{\infty} \, \mu_i \; \left|  \, \langle \partial_{{u}}(\varphi_{i}) \; , \; \varphi_{\ell} \rangle_{_{L^{2}_{\comp}(E , \nu)}}\right|^{2} \right) \, .
\end{multline*}
The last sum is finite as we have, using Fubini's theorem, 
\[
\displaystyle \sum_{\ell=0}^{\infty} \displaystyle \sum_{i=0}^{\infty} \, \mu_i \; \left|  \, \langle \partial_{{u}}(\varphi_{i}) \; , \; \varphi_{_\ell} \rangle_{_{L^{2}_{\comp}(E, \nu)}}\right|^{2}
= \,\displaystyle \sum_{i=0}^{\infty} \mu_{i} \, \displaystyle \sum_{\ell=0}^{\infty} \left|  \, \langle  \, \partial_{{u}}(\varphi_{i})\; , \; \varphi_{_\ell} \rangle_{_{L^{2}_{\comp}(E, \nu)}}\right|^{2} 
\,\leq \, \displaystyle \sum_{i=0}^{\infty} \, \mu_{i} \; \displaystyle\int_{E} | \partial_{{u}}(\varphi_{i})(x) |^{2}\, \nu( \mathrm{d}x). 
\]
Upon applying the monotone convergence theorem (since $\mu_{i}>0$) and proposition \ref{DerivNoyau} we obtain 
\[
\displaystyle \sum_{i=0}^{\infty} \, \mu_{i} \; \displaystyle\int_{E} | \partial_{{u}}(\varphi_{i})(x) |^{2}\, \nu( \mathrm{d}x) =
\displaystyle \int_{E} \;  \displaystyle \sum_{i=0}^{\infty} \mu_{i} \; \left| \partial_{{u}}(\varphi_{i})(x)\right|^{2}  \, \nu(\mathrm{d}x)  \, = \, \int_{E} \partial^{(2)}_{{u}} \, \partial^{(1)}_{{u}} \; k( x , x ) \, \nu(\mathrm{d}x),
\]
and $\displaystyle \left(  \int_{E} \partial^{(2)}_{{u}} \partial^{(1)}_{{u}} k( x , x ) \, \nu(\mathrm{d}x) \right)^{\alf}\leq {\left( \nu(E)\; \sup_{x\in E} | k(x,x)| \; \;  \sup_{u\in E }{|u|^2}\,\right)^{\alf}< \infty}$,  {where continuity of the kernel  has been used on both arguments \AD{and  noting that $\nu$ is a finite measure on $E$}. }
We denote by $C$ the last constant. \cqfd
\end{proof}

\subsection{Isometry between $L^{2}_{\comp}(E , \nu)$ and $\HS$}
\label{sec3-RKHS}
Let us now focus on  the links between $L^{2}_{\comp}(E , \nu)$ and $\HS$. Let us recall that integral kernel operator ${\cal L}_{k}$ defined eq.\eqref{def-L_k} is positive definite, compact and self-adjoint, with a range that is assumed to be dense in $L_{\comp}^2(E,\nu)$. As a consequence, it admits a summable non-increasing sequence of strictly positive eigenvalues $(\mu_i)_{i\in\mathbb{N}}$ with corresponding eigenfunctions, $(\varphi_i)_{i\in\mathbb{N}}$ forming an orthonormal basis of $L^{2}_{\comp}(E , \nu)$. Since $\mu_i>0$ for all integers $i$, the integral operator ${\cal L}_{k}$ admits a unique symmetric square root operator denoted by ${\cal L}_{k}^{\alf}$ and defined for all integers $i$ as
${\cal L}_{k}^{\alf} ( \varphi_i) = \mu^{\alf}_{i} \, \varphi_i$. The operator ${\cal L}_{k}^{\alf} : L_{\comp}^{2}(E , \nu) \mapsto \HS$ is a bijective isometry: for all $f, \, g \in L_{\comp}^2(E, \nu)$
\begin{equation}\label{IsometrieLk1/2}
\bigl\langle {\cal L}^{\alf}_{k} f \; , \; {\cal L}^{\alf}_{k} g\bigr\rangle_{_\HS} = \bigl\langle f ,  g\bigr\rangle_{{L_{\comp}^2(\Omega , \nu )}}\,.
\end{equation}
This operator must be understood as ${\cal L}^{\alf}_{k} \circ {\cal L}^{\alf}_{k} = {\cal L}_{k}$. The inverse of ${\cal L}_{k}^{\alf}$, noted ${\cal L}_{k}^{-\alf} : \HS \mapsto L^{2}_{\comp}(E , \nu)$, and  defined for all integers $i\geq 0$ by ${\cal L}_{k}^{-\alf} ( \varphi_i) = \mu^{-\alf}_{i} \, \varphi_i$ is an isometry as well: for all $f, \, g \in \HS$
\begin{equation}\label{IsometrieLk-1/2}
 \qquad \bigl\langle {\cal L}^{-\alf}_{k} f \; , \; {\cal L}^{-\alf}_{k} g\bigr\rangle_{{L^{2}_{\comp}(E , \nu )}} = \bigl\langle f ,  g\bigr\rangle_{_{\HS}}.
\end{equation}
It can be checked easily that ${\cal L}^{\alf}_{k}$ and ${\cal L}^{-\alf}_{k}$ are Hermitian for the inner product of  $L^{2}_{\comp}(E,\nu$): for all $\; f , \, g \in L^{2}_{\comp}(E, \nu)$
\begin{equation}\label{AutoadjointLk1/2}
 \left\langle \AD{j\left( {\cal L}^{\alf}_{k}(f)\right)} \, , g \right\rangle_{{L^{2}_{\comp}(E , \nu)}} =\left\langle f \, , \,  \AD{j\left({\cal L}^{\alf}_{k}(g)\right)}\right\rangle_{{L^{2}_{\comp}(E , \nu)}}, 
\end{equation}
\AD{with $j : \HS \to L^{2}_{\comp}(E, \nu)$ the continuous injection (by comparison of the norms) } and for all $f, \, g \in  \HS$
\begin{equation}\label{AutoadjointLk-1/2}
 \left\langle {\cal L}^{-\alf}_{k}(f) \, , g \right\rangle_{{L^{2}_{\comp}(E , \nu)}} = \left\langle f \, , \, {\cal L}^{-\alf}_{k}(g)\right\rangle_{_{L^{2}_{\comp}(E , \nu)}}.
\end{equation}
By construction, the family of Mercer eigenfunctions $(\varphi_i)_{i\in \mathbb{N}}$ is an orthonormal basis of $L_{\comp}^{2}(E , \nu)$ and ${\cal L}^{\alf}_{k}$ enables to get an orthonormal basis of $\HS$. \AD{ We prove that in the next proposition.  }
\begin{proposition}[Orthonormal basis of $\HS$]\label{BaseHilbertHS}
The sequence of rescaled kernel eigenfunctions $\left({\cal L}^{\alf}_{k} (\varphi_{i})\right)_{i\geq 0}$ is an orthonormal basis of $\HS$.
\end{proposition} 
\begin{proof}
By \eqref{IsometrieLk1/2}, the family $\left({\cal L}^{\alf}_{k} (\varphi_{i})\right)_i$ is orthonormal in $\HS$. Furthermore, we have for all $f\in \HS$
\[
\displaystyle \sum_{i=0}^{\infty} \left| \, \left\langle f \; , \; {\cal L}^{\alf}_{k} (\varphi_{i}) \, \right\rangle_{_{\HS}} \right|^{2}
= \sum_{i=0}^{\infty}  \left| \, \left\langle {\cal L}^{-\alf}_{k}f \; , \; \varphi_{i} \, \right\rangle_{{L^{2}_{\comp}(E , \nu)}} \right|^{2} 
= \, \left\| {\cal L}^{-\alf}_{k}(f) \right\|^{2}_{{L^{2}_{\comp}(E , \nu)}} = \left\| f \right\|^{2}_{_{\HS}} \;  , 
\]
by \eqref{AutoadjointLk-1/2} and \eqref{IsometrieLk-1/2}. We obtain the result through Bessel characterization. \cqfd
\end{proof}
The next theorem and proposition provide us useful results for the following
\begin{theorem}[Compact injection]\label{InjectionContinueCompacte}
The injection $j : \HS \mapsto L^{2}_{\comp}(E , \nu)$ is continuous and compact.
\end{theorem}
\begin{proof}
The comparison of the norms of $\HS$ and $L^{2}_{\comp}(E , \nu)$ and \eqref{SommeMui} justifies the continuity of $j$ (see remark \ref{ComparisonNormL2HS}). By proposition \ref{BaseHilbertHS} and \eqref{SommeMui}, we have 
\[
\displaystyle \sum_{i=0}^{\infty} \left\| \; j ({\cal L}^{\alf}_{k} \; \varphi_i) \, \right\|^{2}_{{L^{2}_{\comp}(E, \nu)}} = \sum_{i=0}^{\infty} \mu_{i}  < \infty;
\]
$j$ is also an Hilbert-Schmidt operator and therefore a compact operator. \cqfd
\end{proof}

\begin{proposition}[Compact injection range denseness]\label{DensiteResultat}The {range of $j$,} $j(\HS)$, is dense in $L^{2}_{\comp}(E, \nu)$, more precisely the linear set spanned by $\{ k(\, \bcdot \, , \, \XX) \, : \, \XX \in E \}$ is dense in $L^{2}_{\comp}(E, \nu)$.
\end{proposition}
\begin{proof}
Let $f\in L^{2}_{\comp}(E, \nu)$ such that $\langle f \, ,\,  g \rangle_{_{L^{2}_{\comp}(E , \nu)}}=0$ for all $g\in \HS$. In particular, for all $\XX \in E$ we have $\langle f \, ,\,  k(\,\bcdot\, , \XX) \,  \rangle_{{L^{2}_{\comp}(E , \nu)}}=0$, which means that $(\mathcal{L}_{k}f)(\XX)=0$.  Since $\mathcal{L}_k$ is injective, we obtain that $f=0$. \cqfd
\end{proof}

From now on, we note $\HS^{\AD{\mathrm{sp}}} : = {\rm Span} \{ \, k(\,\bcdot\, , \, \XX) \; : \; \XX \in E \, \}$. {P}roposition \eqref{DensiteResultat} justifies that $\HS^{\AD{\mathrm{sp}}}$ is dense in $(L^{2}_{\comp}(E, \nu) \, , \, \| \bcdot \|_{L^{2}_{\comp}(E, \nu)})$. Furthermore, by the Moore-Aronszajn theorem \cite{Aronszajn-50}, $\HS^{\AD{\mathrm{sp}}}$ is also dense in $(\HS , \| \bcdot \|_{\HS})$. 

\begin{remark}[Relaxation of ${\cal L}_{k}$ injectivity]\label{ValPropreNulles}
The injectivity property of $\mathcal{L}_k$ (equivalent to its positive definite character) is crucial to show the denseness of $j(\mathcal{H})$.  If some eigenvalues of ${\cal L}_{k}$ are zero, the embedding space $L^{2}_{\comp}(E, \nu)$ we are working with should be replaced by the linear subspace,${\cal V}=Ker(\mathcal{L}_k)^\perp\subset L^{2}_{\comp}(E, \nu)$, which is the actual closure of $j(\mathcal{H})$, and satisfies ${\cal V}= {\rm Span}\{ \varphi_i \; : \; i\geq 0 \; \text{with} \; \mu_i\neq 0 \, \}$.
\end{remark}

\begin{remark}[Higher-order RKHS differentiability]\label{GeneralisationRKHSCp}
Let $p\geq 1$ be an integer. Assuming $k\in {C^{(p,p)}}(E\times E)$, the result of {L}emma \ref{NoyauDerivRKHS} and theorem \ref{RKHSC1} can be extended to higher order differentiation. As a matter fact,  for all {$u\in E$},  all $x \in E$ and all integer $0\leq m\leq p$, we can prove that  the $m$th order derivative $\partial^{(2) ,\, m}_{{u}} k(\, \bcdot \, , \, x ) \in \HS $ and $\HS \subset C^{p}(E , \comp)$. This corresponds,  in slightly different forms,  to results provided   in \cite{Simon-Gabriel-Scholkopf-18,Steinward-Christmann} for $\real^d$.
\end{remark}

\begin{remark}[Links between basis of $L^{2}_{\comp}(E , \nu)$ and $\HS$]\label{LienBaseHilbertL2HS}
Proposition \ref{BaseHilbertHS} proves that
if $(e_i)_{i}$ is a  orthonormal basis of $L^{2}_{\comp}(E , \nu)$ then  $(\mathcal{L}_{k}^{\alf}e_i)_{i}$ is a  orthonormal basis of $\HS$. 
The converse statement is also true and can be shown with the same arguments: if $(f_i)_i$ is a  orthonormal basis of $\HS$ then $(\mathcal{L}_{k}^{-\alf}f_i)_{i}$ is a Hilbert basis of $L^{2}_{\comp}(E , \nu)$. 
\end{remark}
\section{Properties of the Extension and Restriction operators $E_t$ and $R_t$}
\label{EXT-REST}
In this appendix we elaborate on the properties of the restriction and extension operators (namely that they are  isometries; form an adjoint pairs  in $\HS_t$; provide a representation of $R_t$ in $\HS$ and state the continuity of $R_t$ in $L^2(\Omega_t,\nu)$).
\begin{proposition}[Extension/Restriction isometry]\label{RestrProlongementIsometrie}
For all $t\geq 0$, the mapping $E_t : \HS_t \mapsto \HS$ is an isometry : 
\begin{equation}\label{E_tIsometrie}
\forall \, f , g \in { \HS_t}  \qquad \langle E_t(f) \, , \, E_t(g) \rangle_{_{\HS}} \, = \, \langle f \, , \, g \rangle_{_{\HS_t}}
\end{equation}
and $R_t : \HS \mapsto \HS_t$ is an isometry as well:
\begin{equation}\label{R_tIsometrie}
\forall \, f , g \in {\HS}  \qquad \langle R_t(f) \, , \, R_t(g) \rangle_{_{\HS_t}} \, = \, \langle f \, , \, g \rangle_{_{\HS}} \, .
\end{equation}
\end{proposition}
\begin{proof}[For the extension]
Let $f$ and $g \in {\rm Span}\{ k_t(\,\bcdot\, , \XX_t) \, : \, \XX_t\in \Omega_t \}$ : $f = \displaystyle \sum_{i=1}^{N_1}\; \alpha_i \, k_t(\,\bcdot\, , \XX^{(i)}_{t})$ and $g= \displaystyle \sum_{j=1}^{N_2}\beta_{j} \, k_t(\,\bcdot\, , \YY^{(j)}_{t})$ with $\XX^{(i)}_{t}$ and $\YY^{(j)}_{t} \in \Omega_t$. By definition, we have, for $\bm{X}_t = (X_t,t)\in \Omega$ and $\bm{Y}_t = (Y_t,t)\in \Omega$
\[
\displaystyle\langle \,  E_t(f) \, , \, E_t(g) \, \rangle_{_{\HS}}
= \sum_{i=1}^{N_1} \sum_{j=1}^{N_2} \alpha_i \, \overline{\beta_{j}} \; \langle \,  k(\,\bcdot\, , \bm{\XX}^{(i)}_{t}) \; , \; k(\,\bcdot\, , \bm{\YY}^{(j)}_{t}) \,  \rangle_{_{\HS}} = \sum_{i=1}^{N_1} \displaystyle \sum_{j=1}^{N_2} \alpha_i \, \overline{\beta_{j}} \,  k(\,\bm{\YY}^{(j)}_{t}\, , \, \bm{\XX}^{(i)}_{t}).
\]
By construction of $k$, we have $k(\,\bm{\YY}_{t}^{(j)} \, , \, \bm{\XX}^{(i)}_{t})= k_{t}(\,\YY^{(j)}_{t}\, , \, \XX^{(i)}_{t} )$ for all $i$ and $j$, hence
\[
\displaystyle \langle \,  E_t(f) \, , \, E_t(g) \, \rangle_{_{\HS}} =  \sum_{i=1}^{N_1}  \sum_{j=1}^{N_2} \alpha_i \, \overline{\beta_{j}} \, k_{t}(\,\YY^{(j)}_{t}\, , \, \XX^{(i)}_{t} ) = \sum_{i=1}^{N_1} \sum_{j=1}^{N_2} \alpha_i \, \overline{\beta_{j}} \; \langle \,  k_{t}(\,\bcdot\, , \XX^{(i)}_{t}) \; , \; k_{t}(\,\bcdot\, , \YY^{(j)}_{t}) \,  \rangle_{_{\HS_{t}}} =  \langle \,  f \, , \, g \,  \rangle_{_{\HS_t}}\;.
\]
By density of ${\rm Span} \{ k_t(\,\bcdot\, , \XX_t) \, : \, \XX_t\in \Omega_t \}$ in $\left( \HS_t \, , \, \| \bcdot \|_{_{\HS_t}} \right)$, this result extends to  all $f$ and $g\in \HS_t$ and \eqref{E_tIsometrie} is proved. \cqfd
\end{proof}
With similar arguments for $R_t$,  \eqref{R_tIsometrie} is also proved. By \eqref{ExtensionDefRKHS} and \eqref{RestrictDefRKHS}, the mappings $E_t$ and $R_t$ have inverse action on the mapping $k$ and $k_t$. The next proposition justifies that the adjoint of $E_t$ is $R_t$. 

\begin{proposition}\label{RestrictionProlongelementAdjoint}
For all $t\geq 0$, the mappings $E_t$ and $R_t$ are such that for all $f \in \HS_t$ and $g \in \HS$ 
\begin{equation}\label{AdjointE_tR_t}
\langle \, E_t(f) \; , \; g \rangle_{_\HS} \; = \; \langle \, f \; , \; R_t(g) \, \rangle_{_{\HS_t}} .
\end{equation}
\end{proposition}
\begin{proof}
{Let us first show the property for feature maps: for $X_t=\Phi_t(X_0)\in\Omega_t$ and $Y=\Phi_s(Y_0)\in\Omega$, we have 
\[
\big\langle E_t[k_t(\bcdot,X_t)],k(\bcdot,Y)\big\rangle_\mathcal{H}=k(Y,X_t)=k_0(Y_0,X_0)\ell(s,t)
\]
and 
\[
\big\langle k_t(\bcdot,X_t),R_t[k(\bcdot,Y)]\big\rangle_{\mathcal{H}_t}=\big\langle k_t(\bcdot,X_t),k_t(\bcdot,\Phi_t(Y_0))\ell(t,s)\big\rangle_{\mathcal{H}_t}=k_t\big(\Phi_t(Y_0),X_t\big)\ell(t,s)=k_0(Y_0,X_0)\ell(s,t),
\]
where the first equality comes from the definition of the restriction in equation (\ref{RestrictDefRKHS}), and the last equality follows from Definition \ref{def-kt-direct} and the properties of the time kernel in Definition \ref{def_global_kernel}.}
By linearity we obtain that $\langle \, E_t(f) \; , \; g \rangle_{_\HS} \; = \; \langle \, f \; , \; R_t(g) \, \rangle_{_{\HS_t}}$ for all $f \in {\rm Span} \{ \, k_t(\, \bcdot \, , \, \XX_t) \; : \; \XX_t \in \Omega_t \, \}$ and $g\in {\rm Span} \{ \, k(\, \bcdot \, , \, \XX) \; : \; \XX \in \Omega \, \}$.

\noindent
Let $f\in \HS_t$ and $g\in\HS$. There exists  sequences $(f_n)_{n}$ of ${\rm Span} \{ \, k_t(\, \bcdot \, , \, \XX_t) \; : \; \XX_t \in \Omega_t \, \}$ and $(g_n)_n$ of ${\rm Span} \{ \, k(\, \bcdot \, , \, \XX) \; : \; \XX \in \Omega \, \}$ such that
 $\| f - f_n \|_{_{\HS_t}} \longrightarrow 0 $ and $ \| g - g_n \|_{_{\HS}} \longrightarrow 0 $. We have for all $n$ the equality
\[
\langle \, E_t(f_n) \; , \; g_n \rangle_{_\HS} \; = \; \langle \, f_n \; , \; R_t(g_n) \, \rangle_{_{\HS_t}} .
\]
Furthermore, by Cauchy-Schwarz and \eqref{E_tIsometrie} we have that
\begin{multline*}
| \langle \, E_t(f_n) \; , \; g_n \rangle_{_\HS} - \langle \, E_t(f) \; , \; g \rangle_{_\HS} | \;
\leq \; | \langle \, E_t(f_n) - E_t(f)\; , \; g \rangle_{_\HS} | \; + \; | \langle \, E_t(f) \; , \; g_n -g  \rangle_{_\HS} |\\
\leq \; \|f_n - f \|_{_{\HS_t}} \, \| g \|_{_{\HS}} \; + \; \|f \|_{_{\HS_t}} \, \|g_n - g \|_{_{\HS}} 
\end{multline*}
and $\langle  E_t(f_n) \, , \, g_n \rangle_{_\HS} \longrightarrow \langle  E_t(f) \, , \, g \rangle_{_\HS}$. With similar arguments, it can be readily proved  that $\langle  f_n \, , \, R_t(g_n) \rangle_{_{\HS_t}} \longrightarrow \langle  f \, , \, R_t(g)  \rangle_{_{\HS_t}}$ by using \eqref{R_tIsometrie}. We obtain hence \eqref{AdjointE_tR_t}.
\cqfd
\end{proof}

\begin{remark}[Expression of $R_t(f)$ for $f\in \HS$]\label{ExpressionRestrictionFonctionRKHS}
The function $R_t(f)$ belongs to $\HS_t$. By proposition \ref{RestrictionProlongelementAdjoint}, we have for all $\XX_t \in \Omega_t$
\[
 R_t(f)(\XX_t) \, = \, \langle R_t(f) \, , \, k_{t}(\,\bcdot\, , \, \XX_t) \, \rangle_{_{\HS_t}} \, = \, \langle f \, , \, k(\,\bcdot\, , \, \XX_t) \, \rangle_{_{\HS}} \, = \, f(\XX_t) \, .
\]
\end{remark}

\noindent
More generally, for all $t\geq 0$  the mapping $R_t$ can be extended on $L^{2}_{\comp}(\Omega, \nu)$ (keeping the same notation) as: 
\begin{equation}\label{RestrictDefL2}
 \begin{array}{ccccc}
R_{t}  : \; &  & L^{2}_{\comp}(\Omega , \nu) & \mapsto &  L^{2}_{\comp}(\Omega_t , \nu) \\
 & & f \; & \mapsto & \;  f \restriction_{_{\Omega_t}}  \\
\end{array} \, . 
\end{equation}
In particular, on $\HS \subset L^{2}_{\comp}(\Omega , \nu)$ this boils down to the definition given in \eqref{RestrictDefRKHS}. As $\Omega_t \subset \Omega$, we have directly the next proposition.
\begin{proposition}[Continuity of the restriction]\label{ContRestrictionL2}
The mapping $R_{t} : (L^{2}_{\comp}(\Omega , \nu) , \| \bcdot \|_{{L^{2}_{\comp}(\Omega, \nu)}}) \mapsto (L^{2}_{\comp}(\Omega_t , \nu) , \| \bcdot \|_{{L^{2}_{\comp}(\Omega_t, \nu)}}) $ is continuous.
\end{proposition}
\section{Proof of {L}emma 1 (Perron-Frobenius infinitesimal generator)}
\label{Ap-lemma1}
In this appendix we provide the proof on the Perron-Frobenius infinitesimal generator expression and of its domain{, $\mathcal{D}\left(A_{_P}\right)$}. We recall that $\HS^{\AD{\mathrm{sp}}}$ is the linear span of the set $\{k(\,\bcdot\, , x) \; : \; x \in \Omega \}$.
\begin{proof}
Let $f\in \mathcal{D}\left(A_{_P}\right)$ and $\psi\in \HS^{\AD{\mathrm{sp}}}$, we have for all $t\geq 0$, 
\begin{equation}\label{CalculInterm}
\left\langle \, \frac{\, P_t \circ f - f \, }{t}  \; , \; \psi \right\rangle_{L^{2}_{\comp}(\Omega , \nu)} \; = \; \left\langle f \; , \; \frac{ \, U_t \circ \psi - \psi \, }{t} \, \right\rangle_{L^{2}_{\comp}(\Omega , \nu)}
\end{equation}
For all $x\in \Omega$, we have
\[
\dfrac{\, U_t(\psi)(x) - \psi (x) }{t} \xrightarrow[t\mapsto 0^{+}]{} \dfrac{d}{dt} \left[ U_t(\psi)(x) \right] \restriction_{_{t=0}} \, = \, \partial_{\Mop(x)} \psi(x) \,.
\]
We pass to the limit $[t\to 0^{+}]$ in \eqref{CalculInterm}, and recalling  that  $\Mop$ is real valued, divergence free with $\Mop \restriction{_{\partial\Omega}}=0$, and  measure invariant we obtain
\[
\displaystyle \int_{\Omega} A_{_P} f(x) \; \overline{\psi(x)} \;\,  \nu(dx)  \;  = \; \displaystyle \int_{\Omega} - \partial_{\Mop(x)} f(x) \, \overline{\psi(x)} \;\,  \nu(dx) 
\]
thus, $x\mapsto \partial_{\Mop(x)} f(x)$ belongs to $L^{2}_{\comp}(\Omega , \nu )$ and 
\[
\mathcal{D}\left(A_{_P}\right) \subset \left\{ f\in L^{2}_{\comp}(\Omega , \nu) \; : \; x\mapsto \partial_{\Mop(x)} f(x) \, \in L^{2}_{\comp}(\Omega , \nu) \right\} \]
\[\text{and }\; A_{_P} f(x) \, = \, - \partial_{\Mop(x)} f(x)  \quad \text{for} \quad  f\in \mathcal{D}\left(A_{_P}\right) \; \text{and} \;  x\in \Omega . 
\]
Conversely, considering now $f\in L^{2}_{\comp}(\Omega , \nu)$ such that $x\mapsto \partial_{\Mop(x)} f(x) \, \in L^{2}_{\comp}(\Omega , \nu) $, we show that $f\in \mathcal{D}\left(A_{_P}\right)$. For all $\psi \in \HS^{\AD{\mathrm{sp}}}$, we have
\begin{multline*}
\displaystyle \langle P_t(f) - f \, , \, \psi \rangle_{{L^{2}_{\comp}(\Omega , \nu)}} = \langle f \, , \, U_t(\psi) - \psi \rangle_{{L^{2}_{\comp}(\Omega , \nu)}} = \\\langle f \, , \, \int_{0}^{t} \dfrac{\, \mathrm{d} \, U_s(\psi) \,}{\mathrm{d}s}\, \mathrm{d}s \rangle_{{L^{2}_{\comp}(\Omega , \nu)}} = \langle f \, , \, \int_{0}^{t} \partial_{\Mop(\phi_s(\bcdot))} \psi\left[  \Phi_s(\bcdot) \right] \, \mathrm{d}s \, \rangle_{{L^{2}_{\comp}(\Omega , \nu)}} .
\end{multline*}
For all $\psi \in \HS^{\AD{\mathrm{sp}}}$, we have $\langle -\partial_{_{\Mop(\bcdot)}} f(\bcdot) \; , \; \psi \rangle_{{L^{2}_{\comp}(\Omega , \nu)}} = \langle f \; , \;  -\partial_{_{\Mop(\bcdot)}} \psi(\bcdot) \rangle_{{L^{2}_{\comp}(\Omega , \nu)}}$ (since $\Mop$ is divergence free) and we obtain also the following equality
\[
\displaystyle \left\langle \, \frac{ \, P_t(f) - f \,}{t}  -\partial_{_{\Mop(\bcdot)}} f(\bcdot) \; , \; \psi \right\rangle_{{L^{2}_{\comp}(\Omega , \nu)}} 
= \left\langle f \; , \; \dfrac{1}{t} \int_{0}^{t}  \partial_{_{\Mop(\phi_s(\bcdot))}} \psi\left[  \Phi_s(\bcdot) \right]\, \mathrm{d}s \; - \; \partial_{\Mop(\bcdot)} \psi(\bcdot) \, \, \right\rangle_{{L^{2}_{\comp}(\Omega , \nu)}} .
\]
Since $\displaystyle t\mapsto \int_{0}^{t}  \partial_{_{\Mop(\phi_s(\bcdot))}} \psi\left[  \Phi_s(\bcdot) \right]\, \mathrm{d}s $ \; is differentiable in $0^{+}$, we obtain in particular
\[
\displaystyle \left\langle \, \frac{ \, P_t(f) - f \,}{t} - \partial_{\Mop(\bcdot)} f(\bcdot) \; , \; \psi \right\rangle_{{L^{2}_{\comp}(\Omega , \nu)}}  \xrightarrow[t\mapsto 0^{+} ]{} 0,
\]
for all function $\psi \in \HS^{\AD{\mathrm{sp}}}$. By density of $\HS^{\AD{\mathrm{sp}}}$ in $L^{2}_{\comp}(\Omega , \nu)$ we get 
\[
\left\| \, \dfrac{\, P_{t}f - f \,}{t} \; - \; \partial_{\Mop(\bcdot)} f(\bcdot) \; \right\|_{L^{2}_{\comp}(\Omega, \nu)} \xrightarrow[t\mapsto 0^{+} ]{} 0
\]
and hence $f\in \mathcal{D}(A_{_P})$. \cqfd
\end{proof}
\section{Proof of step1 and step2 of \AD{the RKHS family} spectral theorem}
\label{Step-1-2}

This appendix gathers step\;1 and step\;2 of the proof of the \AD{RKHS family} representation theorem. 
Step 1 Uses Lions-Lax-Milgram theorem to show that an approximation  $\widetilde{A}_{_U}$ of the restriction on $\mathcal{H}$ of the infinitesimal generator $A_U$ is bijective. The inverse operator $\widetilde{A}_{_U}^{-1}$ seen  as an operator of  $L^{2}_{\comp}(\Omega , \nu) \mapsto  L^{2}_{\comp}(\Omega , \nu)$  through the composition with the compact injection $j:\,  \HS \mapsto  L^{2}_{\comp}(\Omega , \nu)$   (Theorem \ref{InjectionContinueCompacte}, \ref{sec3-RKHS}) is shown to be compact and self-adjoint (Lemmas \ref{CompaciteOperateurIntermediaire} and \ref{AutoAdjointOperateurIntermediaire}).  The diagonalization of $A_U$ is then obtained from  
the diagonalization of $\widetilde{A}_{_U}^{-1}$  and $\mathcal{L}_k$ (Prop.\ref{AntiCommutation}).
Step\;2 finally constructs from the eigenfunctions of $A_U$ an  orthonormal basis of $\HS_t$ that diagonalizes the operator $A_{U\!,\,t}$ (Prop.\ref{Ht-basis}).
\subsection*{Step 1} 
\noindent
Let $\widetilde{A}_{_U}$ be the linear mapping defined by
\begin{equation}
\begin{array}{ccccc}
\widetilde{A}_{_U}  & : & ( \HS\, , \, \|\bcdot\|_{_\HS} )  & \mapsto & ( L^{2}_{\mathbb{C}}(\Omega , \nu)\, , \, \|\bcdot\|_{{L^{2}_{\mathbb{C}}(\Omega , \nu)}})\\&&&&\\
 & & f & \mapsto & \left[ \, i \,  A_{_U} \;  +  \; \lambda \, \mathcal{L}_{k}^{- \alf} \right] (f) \label{def-widetilde-A}
\end{array}
\end{equation}
with $\lambda$ a fixed real such that 
\[|\lambda|  > \| A_{_U}\|  \, \left( \int_{\Omega} | k(x,x)|^{2} \, \nu(\mathrm{d}x) \right)^{\alf}.
\] 
By Theorem \ref{ContKoopmTHEOREM} (\ref{sec2-RKHS}), the norm of $A_{_U}$  verifies 
\[
\displaystyle \| A_{_U}\| \, \leq  \, \left( \nu({\Omega})\; \sup_{x\in {\Omega}} \|k(x,x)\| \; \;  \sup_{x\in {\Omega} }{|\Mop(x)|^{{2}}}\,\right)^{\alf},
\] 
and  to lighten the notations in the following,  we note $M:= \| A_{_U}\|$. We start analyzing the properties of $\widetilde{A}_{_U}$.  First of all, the mapping $\widetilde{A}_{_U}$ is  continuous. We have for all $f\in \HS$ the inequality
\[
\| \widetilde{A}_{_U} f \|_{{L^{2}_{\mathbb{C}}(\Omega , \nu)}}  \; \leq \;  \| i A_{_U} f \|_{{L^{2}_{_\mathbb{C}}(\Omega , \nu)}} \; + \; |\lambda| \, \| \mathcal{L}_{k}^{- \alf}  f \|_{{L^{2}_{\comp}(\Omega , \nu)}} \; = \; \| A_{_U} f \|_{{L^{2}_{\comp}(\Omega , \nu)}} \; + \; |\lambda| \, \| \mathcal{L}_{k}^{- \alf}  f \|_{{L^{2}_{\comp}(\Omega , \nu)}}.
\]
The identity of \eqref{IsometrieLk-1/2} and {T}heorem \ref{GenerateurKoopBorne} allow us to write \[
\| \widetilde{A}_{_U} f \|_{{L^{2}_{\mathbb{C}}(\Omega , \nu)}} \; \leq \; ( M + |\lambda| ) \;  \| f \|_{_\HS},
\]
\AM{which shows $\widetilde{A}_{_U}$ is continuous on $\HS$. As shown below, the constant $\lambda$ is chosen to ensure the inverse of $\widetilde{A}_{_U}$ is well defined.}
\begin{lemma} [Bijectivity of $\widetilde{A}_{_U}$]\label{BijectiviteOperateurIntermediaire}
The operator $\widetilde{A}_{_U} :( \HS\; , \;  \|\bcdot\|_{_\HS} )  \mapsto ( L^{2}_{\mathbb{C}}(\Omega , \nu)\, , \, \|\bcdot\|_{{L^{2}_{\mathbb{C}}(\Omega , \nu)}})$  is bijective. 
\end{lemma}
\begin{proof}
Let $a$ be the bilinear form defined by
\[
\begin{array}{ccccc}
a & : & \HS \times \HS  & \mapsto & \mathbb{C} \\
 & &  ( f \, , \, g ) & \mapsto & \langle \widetilde{A}_{_U}(f) \, , \, \mathcal{L}_{k}^{-\alf} (g) \, \rangle_{{L^{2}_{\mathbb{C}}(\Omega , \nu)}}
\end{array} .
\]
The bilinear form $a$ is continuous on $\HS$. By \eqref{IsometrieLk-1/2},  for all $f\in \HS$ we have  
\[a(f,f) = \left\langle \, i \, A_{_U} (f) \, , \, \mathcal{L}_{k}^{-\alf}f \, \right\rangle_{{L^{2}_{_\mathbb{C}}(\Omega , \nu)}} \; + \; \lambda \, \| f  \|^{2}_{_\HS}.\] 
In particular, we get 
\[
| a(f,f) | \geq \left| \, |\lambda| \,  \|f \|^{2}_{_\HS} \; - \; |\langle \,A_{_U} (f) \, , \, \mathcal{L}_{k}^{-\alf} f \rangle_{{L^{2}_{\mathbb{C}}(\Omega , \nu)}} |  \, \right| \geq\\ |\lambda| \,  \|f \|^{2}_{_\HS} \; - \; |\langle \,A_{_U} (f) \, , \, \mathcal{L}_{k}^{-\alf}f \rangle_{{L^{2}_{\mathbb{C}}(\Omega , \nu)}} | \; .
\]
Using Cauchy-Schwarz inequality we have
\[
- \, | \langle   A_{_U} (f) \, , \, \mathcal{L}_{k}^{-\alf}f \rangle_{{L^{2}_{\comp}(\Omega , \nu)}}  | \; \geq  - \|  A_{_U}(f) \|_{{L^{2}_{\comp}(\Omega , \nu)}} \, \|\mathcal{L}_{k}^{-\alf}f \|_{{L^{2}_{\comp}(\Omega , \nu)}} .
\]
Through the continuity of $A_{_U} : \HS \mapsto L^{2}_{\comp}(\Omega , \nu)$, remark \ref{ComparisonNormL2HS} and \eqref{IsometrieLk-1/2} we infer that 
\[
-|\langle \, A_{_U} (f) \, , \, \mathcal{L}_{k}^{-\alf} f \rangle_{{L^{2}_{\mathbb{C}}(\Omega , \nu)}} |  \; \geq  \; - M \, \left(  \int_{\Omega} | k(x,x)|^{2} \, \nu(\mathrm{d} x) \right)^{\alf} \; \|f \|^{2}_{_\HS} .
\]
We obtain hence that $|a(f,f)| \geq  \AM{C}  \, \| f \|^{2}_{_\HS}$ for all $f\in \HS$ with the  constant
\[
\AM{C} := \left( |\lambda| - M \left( \int_{\Omega} | k(x,x)|^{2} \, \nu(\mathrm{d}x) \right)^{\alf} \right) > 0,
\]
which is positive by definition of $\lambda$. The bilinear form $a$ verifies thus a coercivity condition in $\HS$. 

\noindent To show that $\widetilde{A}_{_U} : \HS \mapsto L^{2}_{\mathbb{C}}(\Omega , \nu)$ is bijective, let us consider $f \in L^{2}_{\mathbb{C}}(\Omega , \nu)$.
\\The linear mapping $g \mapsto  \langle f \, , \, \mathcal{L}_{k}^{-\alf}(g) \, \rangle_{{L^{2}_{\mathbb{C}}(\Omega , \nu)}}$ is continuous on $\HS$. By the Lions–Lax–Milgram theorem, there exists thus a unique $x_f\in \HS$ such that 
\[
 a( x_f \, , \, g ) \; = \; \langle f \; , \; \mathcal{L}_{k}^{-\alf}(g) \, \rangle_{{L^{2}_{\mathbb{C}}(\Omega , \nu)}} \qquad  \text{for all} \; g \in \HS. 
\]
As the operator $\mathcal{L}_{k}^{-\alf} : \HS \mapsto L^{2}_{\comp}(\Omega , \nu)$ is bijective, we obtain also 
\begin{equation}\label{IntermBijectivite}
\langle \widetilde{A}_{_U}(x_f) \, , \, h \, \rangle_{{L^{2}_{\mathbb{C}}(\Omega , \nu)}} \; = \; \langle f \; , \; h \, \rangle_{{L^{2}_{\mathbb{C}}(\Omega , \nu)}} \qquad \text{for all} \;  h \in L^{2}_{\comp}(\Omega , \nu).
\end{equation}
In particular, we obtain hence $\widetilde{A}_{_U} (x_f) \, = \, f$. 
\cqfd
\end{proof}

\noindent This lemma justifies the existence of the operator $\widetilde{A}_{_U}^{-1} : ( L^{2}_{{\mathbb{C}}}(\Omega , \nu)\, , \, \|\bcdot\|_{{L^{2}_{{\mathbb{C}}}(\Omega , \nu)}} ) \mapsto \left( \HS \, , \,  \|\bcdot\|_{_\HS}\right)$. \AM{By the bounded inverse theorem $\widetilde{A}_{_U}^{-1}$ is also bounded.} In order to diagonalize $A_{_U}$, we will  diagonalize the inverse of $\widetilde{A}_{_U}$ in $L^{2}_{{\mathbb{C}}}(\Omega , \nu)$. To that end, we consider $ j \circ \widetilde{A}_{_U}^{-1} :  L^{2}_{\comp}(\Omega , \nu) \mapsto  L^{2}_{\comp}(\Omega , \nu)$ with $j$ the  injection of $\HS$ in $L^{2}_{\comp}(\Omega , \nu)$.  As the injection is compact (Theorem \ref{InjectionContinueCompacte}), \ref{sec3-RKHS}) , we obtain directly that the mapping
\[
 j \circ \widetilde{A}_{_U}^{-1} : ( L^{2}_{_{\mathbb{C}}}(\Omega , \nu)\, , \, \|\bcdot\|_{{L^{2}_{_{\mathbb{C}}}(\Omega , \nu)}} ) \mapsto ( L^{2}_{\comp}(\Omega , \nu)\, , \, \|\bcdot\|_{_{L^{2}_{\comp}(\Omega , \nu)}} )
\]
is continuous and compact. By abuse of notation denoting the composition $j \circ \widetilde{A}_{_U}^{-1} $ as $\widetilde{A}_{_U}^{-1} : L^{2}_{{\mathbb{C}}}(\Omega , \nu) \mapsto L^{2}_{{\mathbb{C}}}(\Omega , \nu)$ we have the next lemma . 

\begin{lemma}[Compactness of $\widetilde{A}_{_U}^{-1}$]\label{CompaciteOperateurIntermediaire}
The operator $\widetilde{A}_{_U}^{-1} : L^{2}_{{\mathbb{C}}}(\Omega , \nu) \mapsto L^{2}_{{\mathbb{C}}}(\Omega , \nu)$ is compact. 
\end{lemma}

\begin{lemma}[Self-adjointness $\widetilde{A}_{_U}^{-1}$]\label{AutoAdjointOperateurIntermediaire}
The operator $\widetilde{A}_{_U}^{-1}$ is self-adjoint in $L^{2}_{{\mathbb{C}}}(\Omega , \nu)$.
\end{lemma}
\begin{proof}
For all $f$ and $g\in L^{2}_{{\mathbb{C}}}(\Omega, \nu)$, we have 
\begin{equation}\label{AutoAdjointOperateurIntermediaireEq1}
\langle \widetilde{A}_{_U}^{-1}(f)\, , \, g \rangle_{{L^{2}_{{\mathbb{C}}}(\Omega , \nu)}} \, = \, \langle \AM{j\circ x_{f}} \, , \, g \rangle_{{L^{2}_{{\mathbb{C}}}(\Omega , \nu)}} \, = \, \langle \AM{j\circ  x_{f}} \, , \, \widetilde{A}_{_U} \, x_g \, \rangle_{{L^{2}_{{\mathbb{C}}}(\Omega , \nu)}} 
\end{equation}
where we used the notation of the proof of  {L}emma \ref{BijectiviteOperateurIntermediaire} (and the elements $x_f$ and $x_g$ belong to $\HS$ with $j$ \AM{the injection $\HS \mapsto L^{2}_{{\mathbb{C}}}(\Omega, \nu)$)}. By definition, we have 
\[
\langle \AM{j\circ x_{f}} \, , \, \widetilde{A}_{_U} \, x_g \, \rangle_{{L^{2}_{{\mathbb{C}}}(\Omega , \nu)}} \; =\; -i \, \langle \AM{j\circ x_{f}} \, , \, A_{_U} (x_g) \, \rangle_{{L^{2}_{{\mathbb{C}}}(\Omega , \nu)}} \, + \, \lambda \, \langle \AM{j\circ x_f} \, , \, \mathcal{L}_{k}^{-\alf}(x_g) \, \rangle_{{L^{2}_{{\mathbb{C}}}(\Omega , \nu)}}.
\]
Skew symmetry of the generator {$A_{_U}$} (Proposition \ref{KoopAntisym}), symmetry of $\mathcal{L}_{k}^{-\alf}$ \eqref{AutoadjointLk-1/2} and as $\lambda$ is real allow us to write
\[
\langle \AM{j\circ x_{f}} \, , \, \widetilde{A}_{_U} \, x_g \, \rangle_{{L^{2}_{{\mathbb{C}}}(\Omega , \nu)}} \; = \; i \, \langle  A_{_U} (x_{f}) \, , \,\AM{j\circ x_g} \, \rangle_{{L^{2}_{{\mathbb{C}}}(\Omega , \nu)}} \, +  \lambda \, \langle \mathcal{L}_{k}^{-\alf}(x_f) \, , \, \AM{j\circ x_g} \, \rangle_{{L^{2}_{{\mathbb{C}}}(\Omega , \nu)}} \; = \; \langle \widetilde{A}_{_U} \, x_{f} \, , \, \AM{j\circ x_g} \, \rangle_{{L^{2}_{{\mathbb{C}}}(\Omega , \nu)}} . 
\]
With similar arguments as \eqref{AutoAdjointOperateurIntermediaireEq1}, we conclude that $\langle \widetilde{A}_{_U}^{-1}(f)\, , \, g \rangle_{{L^{2}_{{\mathbb{C}}}(\Omega , \nu)}} \, = \, \langle f\, , \, \widetilde{A}_{_U}^{-1}(g) \rangle_{{L^{2}_{{\mathbb{C}}}(\Omega , \nu)}}$ for all $f$ and $g\in L^{2}_{{\mathbb{C}}}(\Omega, \nu)$.
\cqfd
\end{proof}
With lemmas \ref{CompaciteOperateurIntermediaire} and \ref{AutoAdjointOperateurIntermediaire}, we may apply the spectral theorem of compact self-adjoint operator and consequently there exists a  orthonormal basis $(e_n)_{n\geq 0}$ of $(L^{2}_{{\mathbb{C}}}(\Omega , \nu) \, , \, \|\bcdot\|_{{ L^{2}_{{\mathbb{C}}}(\Omega , \nu)}})$ and a sequence of real $(\alpha_n)_n \in (\sigma_p( \widetilde{A}_{_U}^{-1}) - \{0\})^{\mathbb{N}}$ such that
\begin{equation}\label{DiagoAUwidetilde}
\lim_{n\mapsto \infty} \alpha_n =0 \qquad \text{and} \qquad  \widetilde{A}_{_U}^{-1} (e_n) = \alpha_n \,  e_n \quad  \text{for all} \;  n \in \mathbb{N} . 
\end{equation}
Moreover, for all integers $n$,  the real eigenvalue $\alpha_n$ has finite multiplicity. It can be noticed that $\alpha_n \neq 0$ for all $n$ since $\widetilde{A}_{_U}^{-1}$ is invertible and the eigenfunctions $e_n$ belong to the RKHS $\HS$. Finally, we have for all $n\in \mathbb{N}$  
\begin{equation*}
\widetilde{A}_{_U} (e_n) \; = \; \left[ \, i \, A_{_U} \, + \, \lambda \, \mathcal{L}_{k}^{- \alf}\right] (e_n) \; = \;  \alpha_{n}^{-1}  \, e_n . 
\end{equation*}
We now are in position to fully specify the eigenvalues of $A_{_U}$. If all the eigenvalues of $\widetilde{A}_{_U}^{-1}$ are simple, it is easier to relate them to the eigenvalues of $A_{_U}$. This is  explained in {R}emark \ref{EigenvalueSimple}. In the other case, we need first to establish a commutativity property between $\widetilde{A}_{_U}^{-1}$ and $\mathcal{L}_{k}$.

\begin{lemma}[Commutation between $A_{_U}$ and $\mathcal{L}_{k}$]\label{GenerateurKoopmanCommutation}
The operators $A_{_U}$ and $\mathcal{L}_{k}$ are commutative on $\HS$. For all $f \in \HS$ we have the identity: $A_{_U} \circ \mathcal{L}_{k} (f) \, = \, \mathcal{L}_{k} \circ A_{_U} (f)$ .
\end{lemma}
Before demonstrating this result, we first {detail the commutation between $\mathcal{L}_k$ and $U_t$} on $\HS^{\AD{\mathrm{sp}}}$. Considering $f \in \HS^{\AD{\mathrm{sp}}}$ the element $\mathcal{L}_{k} (f)$ belongs to $\HS$. We apply {P}roposition \ref{RelationKernel} and we have for all $\XX\in \Omega$ and all $t\geq 0$
\[
\mathcal{L}_{k} [U_t(f)](\XX) \, = \, \int_{\Omega} k(\XX  ,  y) \; U_t(f)(y) \, \nu(\mathrm{d}y) \, = \, \int_{\Omega} k(\Phi_t(\XX), \Phi_{t}(y)) \; f \circ \Phi_t(y) \, \nu(\mathrm{d}y) \, .
\]
Since the dynamical system $(\Phi_t)_{t\geq0}$ is measure preserving and $\mathcal{L}_{k}(f)$ belongs to $\HS$, we obtain 
\[
\mathcal{L}_{k} [U_t(f)](\XX) \, = \, \, \int_{\Omega} k(\Phi_t(\XX) , y) \; f(y) \, \nu(\mathrm{d}y) \, =  [\mathcal{L}_{k} f](\Phi_t(\XX)){=U_t[\mathcal{L}_k f](X)}.
\]
 
\begin{proof}[Lemma \ref{GenerateurKoopmanCommutation}]
First, we prove the result for $f \in \HS^{\AD{\mathrm{sp}}}$.\\
By the previous equality we have for all $g\in L^{2}_{\comp}(\Omega , \nu)$ and $t\geq 0${
\[
\left\langle \dfrac{\; U_t\left[\mathcal{L}_{k} f \right] - \mathcal{L}_{k} f \; }{t} \; , \; g \, \right\rangle_{\!\!\!{L^{2}_{\comp}(\Omega , \nu)}}\!\!\!\!\!\!\!\!\!\!\!\! = \; \left\langle \dfrac{\; \mathcal{L}_{k}\left[U_t f \right] - \mathcal{L}_{k} f \; }{t} \; , \; g \, \right\rangle_{{L^{2}_{\comp}(\Omega , \nu)}}. 
\]
The operator $\mathcal{L}_k$ being self-adjoint on $L^{2}_{\comp}(\Omega , \nu)$ we infer that 
\[
\left\langle \dfrac{\; U_t\left[\mathcal{L}_{k} f \right] \, - \, \mathcal{L}_{k} f \; }{t} \; , \; g \, \right\rangle_{\!\!\!{L^{2}_{\comp}(\Omega , \nu)}}\!\!\!\!\!\!\!\!\!\!\!\! = \; \left\langle \dfrac{\; U_t f \, -  \, f \; }{t} \; , \; \mathcal{L}_{k}(g) \, \right\rangle_{{L^{2}_{\comp}(\Omega , \nu)}}. 
\]
Taking the limit as $[t \to 0^{+} ]$ in the left-hand side, as  $\mathcal{L}_{k}(f)$ belongs to $\HS \subset \mathcal{D}\left(A_{_{U}}\right)$,  we have 
\[
\lim_{t\to0^{+}} \left\langle \dfrac{\; U_t\left[\mathcal{L}_{k} f \right] \, - \, \mathcal{L}_{k} f \; }{t} \; , \; g \, \right\rangle_{\!\!\!{L^{2}_{\comp}(\Omega , \nu)}}\!\!\!\!\!\!\!\!\!\!\!\! = \left\langle A_u \, \mathcal{L}_k (f) \; , \; g \right\rangle_{{L^{2}_{\comp}(\Omega , \nu)}}.
\]
Doing the same in the right-hand side,  as $f$ belongs to $\HS \subset \mathcal{D}\left(A_{_{U}}\right) = \mathcal{D}\left(A_{_{P}}\right)$, we also have }
\[
\lim_{t\to0^{+}} \; \left\langle \dfrac{\; U_t f \, -  \, f \; }{t} \; , \; \mathcal{L}_{k}(g) \, \right\rangle_{{L^{2}_{\comp}(\Omega , \nu)}} = \left\langle A_{_U}f \, , \, \mathcal{L}_k(g) \right\rangle_{{L^{2}_{\comp}(\Omega , \nu)}} = \left\langle \mathcal{L}_k \circ A_{_U}f \, , \, g \right\rangle_{{L^{2}_{\comp}(\Omega , \nu)}}.
\]
We deduce that  $\langle A_{_U}\circ \mathcal{L}_{k} (f) \, , \, g \rangle{_{L^{2}_{\comp}(\Omega, \nu)}} \; = \; \langle  \mathcal{L}_{k} \circ A_{_U} (f) \, , \, g \rangle_{L^{2}_{\comp}(\Omega, \nu)}$
for all $g\in L^{2}_{\comp}(\Omega , \nu)$. The lemma is proved on $\HS^{\AD{\mathrm{sp}}}$.

Then, we suppose that $f\in \HS$. There exists a sequence $(f_n)_n$ of $\HS^{\AD{\mathrm{sp}}}$ which converges to $f$ in $\HS$.\\  The continuity of $A_{_U}: \HS \to L^{2}_{\mathbb{C}}(E,\nu)$ and {R}emark \ref{ContLalfH} justifies for all integers $n$
\[
\|A_{_U}\circ \mathcal{L}_k(f_n) - A_{_U}\circ \mathcal{L}_k(f) \|_{L^{2}_{\mathbb{C}}(E,\nu)}
\leq C \, \| \,\mathcal{L}_k(f_n-f)\,\|_{\HS} 
\leq C \, \| f_n-f\|_{\HS} .\]
As $L_k : L^{2}_{\mathbb{C}}(E,\nu) \to L^{2}_{\mathbb{C}}(E,\nu)$ is continuous, we infer 
\[
\|\mathcal{L}_k \circ A_{_U} (f_n) -  \mathcal{L}_k \circ A_{_U}(f) \|_{L^{2}_{\mathbb{C}}(E,\nu)}
\leq C\, \| \, A_{_U}(f_n-f)\,\|_{L^{2}_{\mathbb{C}}(E,\nu)} \leq C \, \| f_n-f\|_{\HS} .\]
For all integers $n$, we have $A_{_U} \circ \mathcal{L}_k (f_n) = \mathcal{L}_k \circ A_{_U} (f_n)$. We pass to the limit to obtain  $A_{_U} \circ \mathcal{L}_k (f) = \mathcal{L}_k \circ A_{_U} (f)$.

\cqfd
\end{proof}
\begin{lemma}[Commutation between $\widetilde{A}^{-1}_{_U}$ and $\mathcal{L}_k$.]\label{CommutationIntermediaire}
The operators $\widetilde{A}^{-1}_{_U}$ and $\mathcal{L}_{k}$ are commutative on $L^{2}_{\mathbb{C}}(\Omega , \nu)$. 
\end{lemma}
\begin{proof}
We apply the lemma \ref{GenerateurKoopmanCommutation} and we infer for all $f\in \HS$
\[
\mathcal{L}_{k} \circ \widetilde{A}_{_U} (f) = i \, \mathcal{L}_{k} \circ A_{_U} (f) \; + \; \lambda \, \mathcal{L}_{k} \circ \mathcal{L}_{k}^{-\alf}(f) \\= i \, A_{_U} \circ \mathcal{L}_{k}(f) \; + \; \lambda \, \mathcal{L}_{k}^{{-}\alf} \circ \mathcal{L}_{k}(f) \; = \; \widetilde{A}_{_U} \circ \mathcal{L}_{k}(f)
\; .
\]
As $\widetilde{A}_{_U} : \HS \mapsto L^{2}_{\mathbb{C}}(\Omega , \nu)$ is bijective, we obtain that $\mathcal{L}_k \circ \widetilde{A}^{-1}_{_U}(g) \; = \;  \widetilde{A}^{-1}_{_U} \circ \mathcal{L}_{k}(g)$ for all $g\in L^{2}_{\comp}(\Omega , \nu)$.
\end{proof}

\begin{proposition}[Diagonalization of $\widetilde{A}_{_U}^{-1}$  and $\mathcal{L}_k$]\label{AntiCommutation}
The operators $\widetilde{A}_{_U}^{-1}$  and $\mathcal{L}_k$ are diagonalizable in the same orthonormal basis of $L^{2}_{\mathbb{C}}(\Omega , \nu)$. 
\end{proposition}
\begin{proof}
For $n\in \mathbb{N}$, we note $E_{\alpha_{n}} \subset L^{2}_{\comp}(\Omega , \nu)$ the eigenspace of  $\widetilde{A}_{_U}^{-1}$ associated with eigenvalue $\alpha_{n}$. By {L}emma \ref{CommutationIntermediaire}, we have for all $f\in E_{\alpha_n}$ that
\[
[ \widetilde{A}_{_U}^{-1} - \alpha_n I \, ] \circ \mathcal{L}_k(f) = \mathcal{L}_k \circ [ \widetilde{A}_{_U}^{-1} - \alpha_n I \, ](f) = 0, 
\]
and $E_{\alpha_{n}}$ is thus an invariant subspace of $\mathcal{L}_k$. The operator $\mathcal{L}_k : E_{\alpha_{n}} \mapsto E_{\alpha_{n}}$ can be diagonalized. There exists an orthonormal basis $(f_{n,p})_p$ of $E_{\alpha_{n}}$ and reals $(\beta_{n,p})_{p}$ (with $1\leq p \leq dim E_{\alpha_{n}}$) such that 
\[
\mathcal{L}_k(f_{n, p}) = \beta_{n,p} \; f_{n , p} \qquad \text{and} \qquad  \widetilde{A}_{_U}^{-1}(f_{n,p}) = \alpha_{n} \; f_{n ,p} .
\]
The family $(f_{n,p})_{p}$ diagonalizes both the operators $\mathcal{L}_k$ and $\widetilde{A}_{_U}^{-1}$ on $E_{\alpha_{n}}$. It can be noticed that $\beta_{n,p}>0$ since $\mathcal{L}_k$ is positive definite. As $\alpha_n \neq 0$ for all $n$, the functions $f_{n,p}$ belong to $\HS$. The concatenation of the vectors $(f_{n , p})_{n , p}$ for all $n \in \mathbb{N}$ and $1 \leq p \leq dim E_{\alpha_n}$ forms a  orthonormal basis of $L^{2}_{\comp}(\Omega , \nu)$. As a matter of fact, we have by construction that, ${\rm Span} \{ \, f_{n,p} \, : \, n\in\mathbb{N} \quad  \text{and} \quad 1\leq p \leq dim E_{\alpha_{n}} \} = {\rm Span} \{ \, e_{n} \, : \, n \in \mathbb{N} \, \}$ and the family $(f_{n,p})_{n,p}$ is also complete in $L^{2}_{\comp}(\Omega , \nu)$. For all $n$ the family $(f_{n,p})_p$ is orthonormal in $L^{2}_{\comp}(\Omega, \nu)$ (by construction). The eigenspaces of $\widetilde{A}_{_U}^{-1}$ are in addition orthogonal: the vectors $f_{n,p}$ and $f_{m,q}$ are orthogonal in $L^{2}_{\mathbb{C}}(\Omega, \nu)$ for all $n\neq m$, $1\leq p \leq dim E_{\alpha_n}$ and $1\leq q \leq dim E_{\alpha_m}$. The family $(f_{n,p})_{n\in\mathbb{N} , 1 \leq p \leq dim E_{\alpha_n}}$ is thus orthornormal in $L^{2}_{\mathbb{C}}(\Omega, \nu)$ and constitutes a  orthonormal basis of $L^{2}_{\mathbb{C}}(\Omega , \nu)$. 
\cqfd
\end{proof}
The result of step $1$ can now be obtained. Keeping the same notations for the eigenfunctions of $\widetilde{A}_{_U}$,  for all $n\in \mathbb{N}$ and $1 \leq p \leq dim E_{\alpha_n}$, we have $\widetilde{A}_{_U}(f_{n,p})  = \alpha_{n}^{-1} \, f_{n,p}$ \; and \;  $\mathcal{L}_k^{-\alf}(f_{n,p}) = \beta_{n,p}^{-1/2} \, f_{n,p}$. With the expression of $A_{_U}$ in \eqref{def-widetilde-A} we obtain 
\begin{equation}\label{DiagonalisationRKHS2}
A_{_U} (f_{n,p}) \; = \; i \, ( \, \lambda \, \beta_{n,p}^{-1/2} \; - \; \alpha_{n}^{-1} \, ) \, f_{n,p} .
\end{equation}
Denoting for the sake of clarity,  $(V_{\ell})_{\ell}$ the  orthonormal basis $( f_{n,p})_{n,p}$ of $L^{2}_{\mathbb{C}(\Omega, \nu)}$ and $(\lambda_{\ell})_{\ell}$ the associated eigenvalues of $A_{_U}$ (given \eqref{DiagonalisationRKHS2} with $\lambda_{\ell} \in i \,\mathbb{R}$) and  
upon applying  $\mathcal{L}_k^{\alf}$ on  the  orthonormal basis $(V_{\ell})_{\ell\in\mathbb{N}}$ of $L^{2}_{\comp}(\Omega , \nu )$, we get also a  orthonormal basis of $\HS$  ({R}emark \ref{LienBaseHilbertL2HS}) which diagonalizes $A_{_U}$:
\begin{equation}\label{DiagonalisationRKHS}
A_{_U} (\mathcal{L}_k^{\alf} V_{\ell} ) \;  = \; \lambda_{\ell} \; \mathcal{L}_k^{\alf} V_{\ell} \, .
\end{equation}
This completes the proof of step $1$.
\newline

\noindent
It can be noticed that for all $g\in \HS$ we have
\[
A_{_U} (g)=  \displaystyle \sum_{\ell=0}^{\infty}  \, \lambda_{\ell} \; \langle g \; , \; V_{\ell} \rangle_{{L^{2}_{\comp}(\Omega , \nu)}} \, V_{\ell}
\quad \text{and}\quad
A_{_U} (g) = \displaystyle \sum_{\ell=0}^{\infty} \, \lambda_{\ell}  \; \langle g \; , \;  \mathcal{L}_k^{\alf}(V_{\ell}) \rangle_{_{\HS}} \,  \mathcal{L}_k^{\alf}(V_{\ell}). 
\]
\begin{remark}[Particular case:  $\widetilde{A}_{_U}$ has simple eigenvalues]\label{EigenvalueSimple} A simpler case arises when the eigenvalues of $\widetilde{A}_U$ are simple. This corresponds to the case of ergodic dynamical systems \citep{EFHN-15}. 
For all $n$, by \eqref{DiagoAUwidetilde} and {L}emma \ref{GenerateurKoopmanCommutation}, the eigenspace $E_{\alpha_n}$ of $\widetilde{A}_{_U}$ is an invariant space of $\mathcal{L}_{k}^{-\alf}$. As $dim E_{\alpha_n} =1$ and $E_{\alpha_n}$ is an invariant subspace of $\mathcal{L}_{k}^{-\alf}$, there exists reals $\beta_n$ such that $\mathcal{L}_{k}^{-\alf} e_n = \beta_n \, e_n$ and we obtain  
\[
A_{_U} (e_n) \; = \; i \, (\lambda \, \beta_{n} \; - \; \alpha_{n}^{-1} ) \, e_{n},
\]
with the  orthonormal basis $(e_n)_n$ of $L^{2}_{\mathbb{C}}(\Omega , \nu)$ obtained from the diagonalization of $\widetilde{A}^{-1}_{_U}$. We have as well for the associated  orthonormal basis of $\HS$
\[
A_{_U} (\mathcal{L}_k^{\alf} e_n ) \; = \; i \, (\lambda \, \beta_{n} \; - \; \alpha_{n}^{-1} ) \, (\mathcal{L}_k^{\alf}e_n).
\]
\end{remark}

\begin{remark}[Exponential form of the Koopman semi-group]\label{Exponential form}
The range of $\widetilde{A}^{-1}_{_U}$ is included in $\HS$. It can be also noticed that $V_{\ell}$ belongs to $\HS$ and the following inclusions hold 
\[
{\rm Span}\{\, V_{\ell} \, : \, \ell \in\mathbb{N} \, \}\subset \HS \subset L^{2}_{\comp}(\Omega , \nu).
\]
The set ${\rm Span} \{\, V_{\ell} \, : \, \ell\in\mathbb{N} \, \}$ is dense in $(L^{2}_{\mathbb{C}}(\Omega , \nu) , \|\bcdot\|_{{L^{2}_{\mathbb{C}}(\Omega , \nu)}})$ and is  an invariant space of the Koopman infinitesimal generator $A_{_U}$. For all $f$ belonging to this set and all $t\geq 0$, as $A_{_U}$  is bounded, the Koopman operator is uniformly continuous and can be written in terms of an exponential series as
\[
U_{t} (f) \; = \; \exp\left(t\, A_{_U}\right) (f),
\]
with convergence in the operator norm.
Its adjoint, the Perron-Frobenius operator, can be written in the same way as 
\[
P_{t} (f) \; = \; \exp\left(-t\, A_{_U}\right) (f).
\]
Denoting for the sake of simplicity,  $\psi_{\ell} = \mathcal{L}_k^{\alf}(V_{\ell})$ for all integers $\ell$. The family  $(\psi_\ell)_{\ell}$ is a  orthonormal basis of $\HS$. We have $U_t( \psi_{\ell} ) =  \exp(t \, \,\lambda_{\ell}) \, \psi_{\ell} $  and  we obtain hence
\begin{equation}\label{EvolutionVectPropre}
\psi_{\ell}(\XX_t) =  \exp(t \, \lambda_{\ell}) \, \psi_{\ell} (\XX_0) .
\end{equation}
\end{remark}

\subsection*{Step 2}
 From this  orthonormal basis $(\psi_{\ell})_{\ell}$ of $\HS$, we now build a  orthonormal basis of $\HS_t$ that diagonalizes the operator $A_{U\!,\,t}$ for all time $t\geq 0$. 
\begin{proposition} [Hilbert basis of $\HS_t$]\label{Ht-basis}
For all $t\geq 0$, the family $( R_t(\psi_{\ell}))_{\ell}$ is a  orthonormal basis of $(\HS_t \, , \, \|\bcdot\|_{_{\HS_t}} )$. 
\end{proposition}

\begin{proof}
For all $\ell$ and $\ell'$, upon applying {P}roposition \ref{RestrProlongementIsometrie} we have $\langle R_t(\psi_{\ell}) \, , \, R_t(\psi_{\ell'}) \rangle_{_{\HS_{t}}} \, = \, \langle \psi_{\ell} \, , \, \psi_{\ell'} \rangle_{_\HS} = \delta_{\ell \, , \, \ell'}$,
consequently the family $(R_t(\psi_{\ell}))_{\ell}$ is orthonormal in $\HS_{t}$. Besides, the family ${\rm Span} \{ \, R_t(\psi_{\ell}) \; : \; \ell \in \mathbb{N}  \}$ is dense in $\HS_t$. Let $h \in \HS_t$ {be} such that $\langle h \, , \, R_t(\psi_{\ell}) \rangle_{\HS_t}=0$ for all integers $\ell$. By {P}roposition \ref{RestrictionProlongelementAdjoint}, we have $\langle h \, , \, R_t(\psi_{\ell}) \rangle_{_{\HS_{t}}} \, = \, \langle E_t(h) \, , \, \psi_{\ell} \rangle_{_\HS} = 0$. As $(\psi_{\ell})_{\ell}$ is an  orthonormal basis of $\HS$, we conclude that $E_t(h)=0$ and thus $h=0$ on $\HS_t$. 
\cqfd
\end{proof}

\noindent The  orthonormal basis $(R_t(\psi_{\ell}))_{\ell}$ of $\HS_t$ diagonalizes the operator $A_{U\!,\,t}$. As a matter of fact, for all $\ell$, 
\[
A_{U\!,\,t} \left[ R_t \psi_{\ell} \right] = R_t \circ A_{_U} \circ E_t \circ R_t (\psi_{\ell} ) \; = \;  \lambda_{\ell} \, \left[ R_t \psi_{\ell} \right] \; . 
\]
For conciseness purpose, we note in the following $\psi_{\ell}^{t} := R_t(\psi_{\ell})$. Through {R}emark \ref{ExpressionRestrictionFonctionRKHS} and \eqref{EvolutionVectPropre}, it can be noticed that  for all $\XX_0 \in \Omega_0$ 
\[
\psi_{\ell}^{t}(\XX_t) = \exp(t\, \lambda_{\ell}) \, \psi_{\ell}^{0}(\XX_0),
\]
which ends the proof of step 2, and completes the proof of the  spectral representation of Koopman operator in \AD{the RKHS family}. 

\begin{remark}
The set ${\rm Span} \{ \psi_{\ell}^{t} \, : \, \ell \in \mathbb{N} \}$ is an invariant space of $A_{U\!,\,t}$. As a matter of fact, considering remark \ref{Exponential form} and proposition~\ref{RestrProlongementIsometrie},  for all $f \in {\rm Span} \{ \psi_{\ell}^{t} \; : \; \ell \in \mathbb{N} \}$ we have $\exp(t A_{U\!,\, t})(f) \, = \, R_t \, \circ \exp(t \,A_{U}) \, \circ E_t (f)$.
\end{remark}
\section{Numerical details of the QG implementation}\label{A-Numerical-details}
%
The QG model is discretized on a staggered grid of size \AD{$N_x \times N_y$, where $N_x=64 $ and $N_y=128$}, with order-2 finite differences.
Advection term is computed using the Arakawa discretization \citep{arakawa1966} and centered differences are considered for the other terms. 
The Poisson equation is solved using spectral method by projection on a sine basis.
Time integration is performed with a fourth-order  Runge-Kutta scheme.
An ensemble of $p=100$ members is considered for the training, and another test ensemble of $p_2=100$ members is generated for predictions assessment.

\noindent
Starting from an established state of a particular trajectory, a set of initial conditions is generated by perturbing this state by the potential vorticity of $N_{POD}=50$ proper orthogonal decomposition (POD) modes \citep{lumley1967} weighted by a Gaussian random perturbation $\mathcal{N}(0,10\lambda_i/\lambda_1)$, where $\lambda_i$ is the $i^{\text{th}}$ POD eigenvalue. Then, a spin-up time integration of $t=0.03$ is performed to evacuate eventual nonphysical flow features. 
The training and testing sets are generated independently using the same rule.

\noindent
Following the dimensioning given in \cite{san2011}, $R_o=\frac{U_r}{\beta L_r^2}=0.0036$ with $U_r$ and $L_r$ the reference velocity and length scale, respectively.
We give as well the gradient of the Coriolis parameter $\beta=1.75\cdot 10^{-11}$, typical of mid-latitides at the center of the basin.
Setting $L_r=2000\,\mathrm{km}$, leads to $U_r=0.252\,\mathrm{m.s^{-1}}$, and then to a reference time $T_r=\frac{L_r}{U_r}=92\,\text{days}$.

\begin{thebibliography}{91}
\expandafter\ifx\csname natexlab\endcsname\relax\def\natexlab#1{#1}\fi
\providecommand{\url}[1]{\texttt{#1}}
\providecommand{\href}[2]{#2}
\providecommand{\path}[1]{#1}
\providecommand{\DOIprefix}{doi:}
\providecommand{\ArXivprefix}{arXiv:}
\providecommand{\URLprefix}{URL: }
\providecommand{\Pubmedprefix}{pmid:}
\providecommand{\doi}[1]{\href{http://dx.doi.org/#1}{\path{#1}}}
\providecommand{\Pubmed}[1]{\href{pmid:#1}{\path{#1}}}
\providecommand{\bibinfo}[2]{#2}
\ifx\xfnm\relax \def\xfnm[#1]{\unskip,\space#1}\fi
\bibitem[{Bocquet et~al.(2020)Bocquet, Brajard, Carrassi, and
  Bertino}]{Bocquet-et-al-20}
\bibinfo{author}{M.~Bocquet}, \bibinfo{author}{J.~Brajard},
  \bibinfo{author}{A.~Carrassi}, \bibinfo{author}{L.~Bertino},
\newblock \bibinfo{title}{Bayesian inference of chaotic dynamics by merging
  data assimilation, machine learning and expectation-maximization},
\newblock \bibinfo{journal}{F. of Data Science} \bibinfo{volume}{2}
  (\bibinfo{year}{2020}) \bibinfo{pages}{55--80}.
\bibitem[{Brunton et~al.(2016)Brunton, Proctor, and Kutz}]{Brunton-Pnas-16}
\bibinfo{author}{S.~Brunton}, \bibinfo{author}{J.~Proctor},
  \bibinfo{author}{N.~Kutz},
\newblock \bibinfo{title}{Discovering governing equations from data by sparse
  identification of nonlinear dynamical systems},
\newblock \bibinfo{journal}{Proceedings of the National Academy of Sciences}
  \bibinfo{volume}{113} (\bibinfo{year}{2016}) \bibinfo{pages}{3932--3937}.
\bibitem[{Fablet et~al.(2021)Fablet, Chapron, Drumetz, M{\'e}min, Pannekoucke,
  and Rousseau}]{Fablet-JAMES-21}
\bibinfo{author}{R.~Fablet}, \bibinfo{author}{B.~Chapron},
  \bibinfo{author}{L.~Drumetz}, \bibinfo{author}{E.~M{\'e}min},
  \bibinfo{author}{O.~Pannekoucke}, \bibinfo{author}{F.~Rousseau},
\newblock \bibinfo{title}{Learning variational data assimilation models and
  solvers},
\newblock \bibinfo{journal}{Journal of Advances in Modeling Earth Systems}
  \bibinfo{volume}{13} (\bibinfo{year}{2021}) \bibinfo{pages}{e2021MS002572}.
  \bibinfo{note}{E2021MS002572 2021MS002572}.
\bibitem[{Gottwald and Reich(2021)}]{Gottwald-Reich-21}
\bibinfo{author}{G.~A. Gottwald}, \bibinfo{author}{S.~Reich},
\newblock \bibinfo{title}{Supervised learning from noisy observations:
  Combining machine-learning techniques with data assimilation},
\newblock \bibinfo{journal}{Physica D: Nonlinear Phenomena}
  \bibinfo{volume}{423} (\bibinfo{year}{2021}) \bibinfo{pages}{132911}.
\bibitem[{Hamzi and Owhadi(2021)}]{Hamzi-Phys-D-I-21}
\bibinfo{author}{B.~Hamzi}, \bibinfo{author}{H.~Owhadi},
\newblock \bibinfo{title}{Learning dynamical systems from data: A simple
  cross-validation perspective, part i: Parametric kernel flows},
\newblock \bibinfo{journal}{Physica D: Nonlinear Phenomena}
  \bibinfo{volume}{421} (\bibinfo{year}{2021}) \bibinfo{pages}{132817}.
\bibitem[{Owhadi and Yoo(2019)}]{Owhadi-JCP-19}
\bibinfo{author}{H.~Owhadi}, \bibinfo{author}{G.~R. Yoo},
\newblock \bibinfo{title}{Kernel flows: From learning kernels from data into
  the abyss},
\newblock \bibinfo{journal}{Journal of Computational Physics}
  \bibinfo{volume}{389} (\bibinfo{year}{2019}) \bibinfo{pages}{22--47}.
\bibitem[{Pathak et~al.(2017)Pathak, Lu, Hunt, Girvan, and Ott}]{Pathak-17}
\bibinfo{author}{J.~Pathak}, \bibinfo{author}{Z.~Lu},
  \bibinfo{author}{B.~Hunt}, \bibinfo{author}{M.~Girvan},
  \bibinfo{author}{E.~Ott},
\newblock \bibinfo{title}{Using machine learning to replicate chaotic
  attractors and calculate {L}yapunov exponents from data},
\newblock \bibinfo{journal}{Chaos: An Interdisciplinary Journal of Nonlinear
  Science} \bibinfo{volume}{27} (\bibinfo{year}{2017}) \bibinfo{pages}{121102}.
\bibitem[{Pathak et~al.(2018)Pathak, Hunt, Girvan, Lu, and Ott}]{Pathak-PRL-18}
\bibinfo{author}{J.~Pathak}, \bibinfo{author}{B.~Hunt},
  \bibinfo{author}{M.~Girvan}, \bibinfo{author}{Z.~Lu},
  \bibinfo{author}{E.~Ott},
\newblock \bibinfo{title}{Model-free prediction of large spatiotemporally
  chaotic systems from data: A reservoir computing approach},
\newblock \bibinfo{journal}{Phys. Rev. Lett.} \bibinfo{volume}{120}
  (\bibinfo{year}{2018}) \bibinfo{pages}{024102}.
\bibitem[{Raissi et~al.(2019)Raissi, Perdikaris, and Karniadakis}]{Raissi-19}
\bibinfo{author}{M.~Raissi}, \bibinfo{author}{P.~Perdikaris},
  \bibinfo{author}{G.~Karniadakis},
\newblock \bibinfo{title}{Physics-informed neural networks: A deep learning
  framework for solving forward and inverse problems involving nonlinear
  partial differential equations},
\newblock \bibinfo{journal}{Journal of Computational Physics}
  \bibinfo{volume}{378} (\bibinfo{year}{2019}) \bibinfo{pages}{686--707}.
\bibitem[{Zhao and Giannakis(2016)}]{Zhao-Giannakis-16}
\bibinfo{author}{Z.~Zhao}, \bibinfo{author}{D.~Giannakis},
\newblock \bibinfo{title}{Analog forecasting with dynamics-adapted kernels},
\newblock \bibinfo{journal}{Nonlinearity} \bibinfo{volume}{29}
  (\bibinfo{year}{2016}) \bibinfo{pages}{2888--2939}.
\bibitem[{Koopman(1931)}]{Koopman-31}
\bibinfo{author}{B.~O. Koopman},
\newblock \bibinfo{title}{Hamiltonian systems and transformation in {H}ilbert
  space},
\newblock \bibinfo{journal}{Proceedings of the National Academy of Sciences}
  \bibinfo{volume}{17} (\bibinfo{year}{1931}) \bibinfo{pages}{315--318}.
\bibitem[{Eisner et~al.(2015)Eisner, Farkas, Haase, and Nagel}]{EFHN-15}
\bibinfo{author}{T.~Eisner}, \bibinfo{author}{B.~Farkas},
  \bibinfo{author}{M.~Haase}, \bibinfo{author}{R.~Nagel},
  \bibinfo{title}{Operator Theoretic Aspects of Ergodic Theory}, Graduate Texts
  in Mathematics, \bibinfo{publisher}{Springer}, \bibinfo{year}{2015}.
\bibitem[{Dellnitz and Junge(1999)}]{Dellnitz-99}
\bibinfo{author}{M.~Dellnitz}, \bibinfo{author}{O.~Junge},
\newblock \bibinfo{title}{On the approximation of complicated dynamical
  behavior},
\newblock \bibinfo{journal}{SIAM Journal on Numerical Analysis}
  \bibinfo{volume}{36} (\bibinfo{year}{1999}) \bibinfo{pages}{491--515}.
\bibitem[{Mezic(2005)}]{Mezic-05}
\bibinfo{author}{I.~Mezic},
\newblock \bibinfo{title}{Spectral properties of dynamical systems, model
  reduction and decomposition},
\newblock \bibinfo{journal}{Nonlinear Dynamics} \bibinfo{volume}{41}
  (\bibinfo{year}{2005}) \bibinfo{pages}{309--325}.
\bibitem[{Budi\v{s}i\'{c} et~al.(2012)Budi\v{s}i\'{c}, Mohr, and
  Mezi\'{c}}]{budisic2012}
\bibinfo{author}{M.~Budi\v{s}i\'{c}}, \bibinfo{author}{R.~Mohr},
  \bibinfo{author}{I.~Mezi\'{c}},
\newblock \bibinfo{title}{Applied {K}oopmanism},
\newblock \bibinfo{journal}{Chaos: An Interdisciplinary Journal of Nonlinear
  Science} \bibinfo{volume}{22} (\bibinfo{year}{2012}) \bibinfo{pages}{047510}.
\bibitem[{Das and Giannakis(2019)}]{Das-Giannakis-19}
\bibinfo{author}{S.~Das}, \bibinfo{author}{D.~Giannakis},
\newblock \bibinfo{title}{Delay-coordinate maps and the spectra of {K}oopman
  operators},
\newblock \bibinfo{journal}{J. Stat. Phys.} \bibinfo{volume}{175}
  (\bibinfo{year}{2019}) \bibinfo{pages}{1107--1145}.
\bibitem[{Das and D.Giannakis(2020)}]{Das-Giannakis-20}
\bibinfo{author}{S.~Das}, \bibinfo{author}{D.Giannakis},
\newblock \bibinfo{title}{{K}oopman spectra in reproducing kernel {H}ilbert
  spaces},
\newblock \bibinfo{journal}{Appl.Comput.Harmon.Anal.} \bibinfo{volume}{49}
  (\bibinfo{year}{2020}) \bibinfo{pages}{573--607}.
\bibitem[{Rowley et~al.(2009)Rowley, Mezic, Bagheri, Schlatter, and
  Henningson}]{Rowley09}
\bibinfo{author}{C.~Rowley}, \bibinfo{author}{I.~Mezic},
  \bibinfo{author}{S.~Bagheri}, \bibinfo{author}{P.~Schlatter},
  \bibinfo{author}{D.~Henningson},
\newblock \bibinfo{title}{Spectral analysis of nonlinear flows},
\newblock \bibinfo{journal}{J. Fluid Mech.} \bibinfo{volume}{641}
  (\bibinfo{year}{2009}) \bibinfo{pages}{115--127}.
\bibitem[{Schmid(2010)}]{Schmid10}
\bibinfo{author}{P.~Schmid},
\newblock \bibinfo{title}{Dynamic mode decomposition of numerical and
  experimental data},
\newblock \bibinfo{journal}{J. Fluid Mech.} \bibinfo{volume}{656}
  (\bibinfo{year}{2010}) \bibinfo{pages}{5--28}.
\bibitem[{Tu et~al.(2014)Tu, Rowley, Luchtenburg, Brunton, and Kutz}]{tu2014}
\bibinfo{author}{J.~H. Tu}, \bibinfo{author}{C.~W. Rowley},
  \bibinfo{author}{D.~M. Luchtenburg}, \bibinfo{author}{S.~L. Brunton},
  \bibinfo{author}{J.~N. Kutz},
\newblock \bibinfo{title}{On dynamic mode decomposition: Theory and
  applications},
\newblock \bibinfo{journal}{Journal of Computational Dynamics}
  \bibinfo{volume}{1} (\bibinfo{year}{2014}) \bibinfo{pages}{391--421}.
\bibitem[{Williams et~al.(2015{\natexlab{a}})Williams, Kevrekidis, and
  Rowley}]{Williams-et-al-15}
\bibinfo{author}{M.~O. Williams}, \bibinfo{author}{I.~G. Kevrekidis},
  \bibinfo{author}{C.~W. Rowley},
\newblock \bibinfo{title}{A data--driven approximation of the {K}oopman
  operator: Extending dynamic mode decomposition,},
\newblock \bibinfo{journal}{Journal of Nonlinear Science} \bibinfo{volume}{25}
  (\bibinfo{year}{2015}{\natexlab{a}}) \bibinfo{pages}{1307--1346}.
\bibitem[{Williams et~al.(2015{\natexlab{b}})Williams, Rowley, and
  Kevrekidis}]{williams2016}
\bibinfo{author}{M.~O. Williams}, \bibinfo{author}{C.~W. Rowley},
  \bibinfo{author}{I.~G. Kevrekidis},
\newblock \bibinfo{title}{A kernel-based method for data-driven {K}oopman
  spectral analysis},
\newblock \bibinfo{journal}{Journal of Computational Dynamics}
  \bibinfo{volume}{2} (\bibinfo{year}{2015}{\natexlab{b}})
  \bibinfo{pages}{247--265}.
\bibitem[{Kutz et~al.(2016)Kutz, Fu, and Brunton}]{Kutz-16}
\bibinfo{author}{N.~Kutz}, \bibinfo{author}{X.~Fu},
  \bibinfo{author}{S.~Brunton},
\newblock \bibinfo{title}{Multiresolution dynamic mode decomposition},
\newblock \bibinfo{journal}{SIAM Journal on Applied Dynamical Systems}
  \bibinfo{volume}{15} (\bibinfo{year}{2016}) \bibinfo{pages}{713--735}.
\bibitem[{DeGennaro and Urban(2019)}]{degennaro2019}
\bibinfo{author}{A.~M. DeGennaro}, \bibinfo{author}{N.~M. Urban},
\newblock \bibinfo{title}{Scalable extended dynamic mode decomposition using
  random kernel approximation},
\newblock \bibinfo{journal}{SIAM Journal on Scientific Computing}
  \bibinfo{volume}{41} (\bibinfo{year}{2019}) \bibinfo{pages}{A1482--A1499}.
\bibitem[{Buza et~al.(2021)Buza, Jain, and Haller}]{buza2021}
\bibinfo{author}{G.~Buza}, \bibinfo{author}{S.~Jain},
  \bibinfo{author}{G.~Haller},
\newblock \bibinfo{title}{Using spectral submanifolds for optimal mode
  selection in nonlinear model reduction},
\newblock \bibinfo{journal}{Proceedings of the Royal Society A: Mathematical,
  Physical and Engineering Sciences} \bibinfo{volume}{477}
  (\bibinfo{year}{2021}) \bibinfo{pages}{20200725}.
\bibitem[{Pan et~al.(2021)Pan, Arnold-Medabalimi, and Duraisamy}]{pan2021}
\bibinfo{author}{S.~Pan}, \bibinfo{author}{N.~Arnold-Medabalimi},
  \bibinfo{author}{K.~Duraisamy},
\newblock \bibinfo{title}{Sparsity-promoting algorithms for the discovery of
  informative {K}oopman-invariant subspaces},
\newblock \bibinfo{journal}{Journal of Fluid Mechanics} \bibinfo{volume}{917}
  (\bibinfo{year}{2021}) \bibinfo{pages}{A18}.
\bibitem[{Colbrook and Townsend(2021)}]{colbrok2021}
\bibinfo{author}{M.~J. Colbrook}, \bibinfo{author}{A.~Townsend},
\newblock \bibinfo{title}{Rigorous data-driven computation of spectral
  properties of {K}oopman operators for dynamical systems},
\newblock \bibinfo{journal}{arXiv preprint https://arxiv.org/abs/2111.14889}
  (\bibinfo{year}{2021}).
\bibitem[{Baddoo et~al.(2022)Baddoo, Herrmann, McKeon, and
  Brunton}]{baddoo2022}
\bibinfo{author}{P.~J. Baddoo}, \bibinfo{author}{B.~Herrmann},
  \bibinfo{author}{B.~J. McKeon}, \bibinfo{author}{S.~L. Brunton},
\newblock \bibinfo{title}{Kernel learning for robust dynamic mode
  decomposition: linear and nonlinear disambiguation optimization},
\newblock \bibinfo{journal}{Proceedings of the Royal Society A}
  \bibinfo{volume}{478} (\bibinfo{year}{2022}) \bibinfo{pages}{20210830}.
\bibitem[{Brunton et~al.(2017)Brunton, Brunton, Proctor, Kaiser, and
  Kutz}]{Brunton17}
\bibinfo{author}{S.~Brunton}, \bibinfo{author}{B.~Brunton},
  \bibinfo{author}{J.~Proctor}, \bibinfo{author}{E.~Kaiser},
  \bibinfo{author}{J.~Kutz},
\newblock \bibinfo{title}{Chaos as an intermittently forced linear system},
\newblock \bibinfo{journal}{Nat. Commun.} \bibinfo{volume}{8}
  (\bibinfo{year}{2017}).
\bibitem[{Giannakis et~al.(2015)Giannakis, Slawinska, and Zhao}]{Giannakis-15}
\bibinfo{author}{D.~Giannakis}, \bibinfo{author}{J.~Slawinska},
  \bibinfo{author}{Z.~Zhao},
\newblock \bibinfo{title}{Spatiotemporal feature extraction with data-driven
  {K}oopman operators},
\newblock in: \bibinfo{booktitle}{Journ. of Mach. Learn. Res.: Workshop and
  Conference Proceedings}, volume~\bibinfo{volume}{44}, \bibinfo{year}{2015},
  pp. \bibinfo{pages}{103--115}.
\bibitem[{Giannakis(2017)}]{Giannakis17}
\bibinfo{author}{D.~Giannakis},
\newblock \bibinfo{title}{Data-driven spectral decomposition and forecasting of
  ergodic dynamical systems},
\newblock \bibinfo{journal}{Appl. Comput. Harmon. Anal.}
  \bibinfo{volume}{https://doi.org/10.1016/j.acha.2017.09.001}
  (\bibinfo{year}{2017}).
\bibitem[{Vautard and Ghil(1989)}]{Vautard-89}
\bibinfo{author}{R.~Vautard}, \bibinfo{author}{M.~Ghil},
\newblock \bibinfo{title}{Singular spectrum analysis in nonlinear dynamics,
  with applications to paleoclimatic time series},
\newblock \bibinfo{journal}{Physica D: Nonlinear Phenomena}
  \bibinfo{volume}{35} (\bibinfo{year}{1989}) \bibinfo{pages}{395--424}.
\bibitem[{Kondrashov et~al.(2020)Kondrashov, Ryzhov, and
  Berloff}]{Kondrashov-et-al-20}
\bibinfo{author}{D.~Kondrashov}, \bibinfo{author}{E.~A. Ryzhov},
  \bibinfo{author}{P.~Berloff},
\newblock \bibinfo{title}{{Data-adaptive harmonic analysis of oceanic waves and
  turbulent flows}},
\newblock \bibinfo{journal}{Chaos: An Interdisciplinary Journal of Nonlinear
  Science} \bibinfo{volume}{30} (\bibinfo{year}{2020}). \bibinfo{note}{061105}.
\bibitem[{Zhen et~al.(2022)Zhen, Chapron, M\'emin, and Peng}]{Zhen-et-al-2022}
\bibinfo{author}{Y.~Zhen}, \bibinfo{author}{B.~Chapron},
  \bibinfo{author}{E.~M\'emin}, \bibinfo{author}{L.~Peng},
\newblock \bibinfo{title}{Eigenvalues of autocovariance matrix: A practical
  method to identify the {K}oopman eigenfrequencies},
\newblock \bibinfo{journal}{Phys. Rev. E} \bibinfo{volume}{105}
  (\bibinfo{year}{2022}) \bibinfo{pages}{034205}.
\bibitem[{Zhen et~al.(2023)Zhen, Chapron, and M{\'e}min}]{Zhen-et-al-2023}
\bibinfo{author}{Y.~Zhen}, \bibinfo{author}{B.~Chapron},
  \bibinfo{author}{E.~M{\'e}min},
\newblock \bibinfo{title}{Bridging {K}oopman operator and time-series
  auto-correlation based {H}ilbert--{S}chmidt operator},
\newblock in: \bibinfo{editor}{B.~Chapron}, \bibinfo{editor}{D.~Crisan},
  \bibinfo{editor}{D.~Holm}, \bibinfo{editor}{E.~M{\'e}min},
  \bibinfo{editor}{A.~Radomska} (Eds.), \bibinfo{booktitle}{Stochastic
  Transport in Upper Ocean Dynamics}, \bibinfo{publisher}{Springer
  International Publishing}, \bibinfo{year}{2023}, pp.
  \bibinfo{pages}{301--316}.
\bibitem[{Mezi\'c(2005)}]{mezic2005}
\bibinfo{author}{I.~Mezi\'c},
\newblock \bibinfo{title}{Spectral properties of dynamical systems, model
  reduction and decompositions},
\newblock \bibinfo{journal}{Nonlinear Dynamics} \bibinfo{volume}{41}
  (\bibinfo{year}{2005}) \bibinfo{pages}{309--325}.
\bibitem[{Das et~al.(2021)Das, Giannakis, and Slawinska}]{Das-Giannakis-21}
\bibinfo{author}{S.~Das}, \bibinfo{author}{D.~Giannakis},
  \bibinfo{author}{J.~Slawinska},
\newblock \bibinfo{title}{Reproducing kernel {H}ilbert space compactification
  of unitary evolution groups},
\newblock \bibinfo{journal}{Applied and Computational Harmonic Analysis}
  \bibinfo{volume}{54} (\bibinfo{year}{2021}) \bibinfo{pages}{75--136}.
\bibitem[{Gilani et~al.(2021)Gilani, Giannakis, and Harlim}]{gilani2021}
\bibinfo{author}{F.~Gilani}, \bibinfo{author}{D.~Giannakis},
  \bibinfo{author}{J.~Harlim},
\newblock \bibinfo{title}{Kernel-based prediction of non-{M}arkovian time
  series},
\newblock \bibinfo{journal}{Physica D: Nonlinear Phenomena}
  \bibinfo{volume}{418} (\bibinfo{year}{2021}) \bibinfo{pages}{132829}.
\bibitem[{Alexander and Giannakis(2020)}]{alexander2020}
\bibinfo{author}{R.~Alexander}, \bibinfo{author}{D.~Giannakis},
\newblock \bibinfo{title}{Operator-theoretic framework for forecasting
  nonlinear time series with kernel analog techniques},
\newblock \bibinfo{journal}{Physica D: Nonlinear Phenomena}
  \bibinfo{volume}{409} (\bibinfo{year}{2020}) \bibinfo{pages}{132520}.
\bibitem[{Burov et~al.(2021)Burov, Giannakis, Manohar, and Stuart}]{burov2021}
\bibinfo{author}{D.~Burov}, \bibinfo{author}{D.~Giannakis},
  \bibinfo{author}{K.~Manohar}, \bibinfo{author}{A.~Stuart},
\newblock \bibinfo{title}{Kernel analog forecasting: Multiscale test problems},
\newblock \bibinfo{journal}{Multiscale Model. Simul.} \bibinfo{volume}{19}
  (\bibinfo{year}{2021}) \bibinfo{pages}{1011--1040}.
\bibitem[{Klus et~al.(2020)Klus, Schuster, and Muandet}]{klus2020}
\bibinfo{author}{S.~Klus}, \bibinfo{author}{I.~Schuster},
  \bibinfo{author}{K.~Muandet},
\newblock \bibinfo{title}{Eigendecompositions of transfer operators in
  reproducing kernel {H}ilbert spaces},
\newblock \bibinfo{journal}{Journal of Nonlinear Science} \bibinfo{volume}{30}
  (\bibinfo{year}{2020}) \bibinfo{pages}{283--315}.
\bibitem[{Santitissadeekorn and Bollt(2020)}]{santitissadeekorn2020}
\bibinfo{author}{N.~Santitissadeekorn}, \bibinfo{author}{E.~M. Bollt},
\newblock \bibinfo{title}{Ensemble-based method for the inverse
  {F}robenius--{P}erron operator problem: Data-driven global analysis from
  spatiotemporal ``movie'' data},
\newblock \bibinfo{journal}{Physica D: Nonlinear Phenomena}
  \bibinfo{volume}{411} (\bibinfo{year}{2020}) \bibinfo{pages}{132603}.
\bibitem[{Lee et~al.(2023)Lee, {De Brouwer}, Hamzi, and Owhadi}]{Lee-Phys-D-23}
\bibinfo{author}{J.~Lee}, \bibinfo{author}{E.~{De Brouwer}},
  \bibinfo{author}{B.~Hamzi}, \bibinfo{author}{H.~Owhadi},
\newblock \bibinfo{title}{Learning dynamical systems from data: A simple
  cross-validation perspective, part iii: Irregularly-sampled time series},
\newblock \bibinfo{journal}{Physica D: Nonlinear Phenomena}
  \bibinfo{volume}{443} (\bibinfo{year}{2023}) \bibinfo{pages}{133546}.
\bibitem[{Hamzi et~al.(2021)Hamzi, Maulik, and Owhadi}]{Hamzi-Proc-RSA-21}
\bibinfo{author}{B.~Hamzi}, \bibinfo{author}{R.~Maulik},
  \bibinfo{author}{H.~Owhadi},
\newblock \bibinfo{title}{Simple, low-cost and accurate data-driven geophysical
  forecasting with learned kernels},
\newblock \bibinfo{journal}{Proceedings of the Royal Society A}
  \bibinfo{volume}{477} (\bibinfo{year}{2021}).
\bibitem[{Arbabi and Mezic(2017)}]{Arbabi-Mezic17}
\bibinfo{author}{H.~Arbabi}, \bibinfo{author}{I.~Mezic},
\newblock \bibinfo{title}{Ergodic theory, dynamic mode decomposition and
  computation of spectral properties of the {K}oopman operator},
\newblock \bibinfo{journal}{SIAM Journal on Applied Dynamical Systems}
  \bibinfo{volume}{16} (\bibinfo{year}{2017}) \bibinfo{pages}{2096--2126}.
\bibitem[{Klus et~al.(2020)Klus, N{\"u}ske, Peitz, Niemann, Clementi, and
  Sch{\"u}tte}]{Klus-et-al-Physica-D-2021}
\bibinfo{author}{S.~Klus}, \bibinfo{author}{F.~N{\"u}ske},
  \bibinfo{author}{S.~Peitz}, \bibinfo{author}{J.-H. Niemann},
  \bibinfo{author}{C.~Clementi}, \bibinfo{author}{C.~Sch{\"u}tte},
\newblock \bibinfo{title}{Data-driven approximation of the {K}oopman generator:
  Model reduction, system identification, and control},
\newblock \bibinfo{journal}{Physica D: Nonlinear Phenomena}
  \bibinfo{volume}{406} (\bibinfo{year}{2020}) \bibinfo{pages}{132416}.
\bibitem[{Brunton et~al.(2022)Brunton, Budi\v{s}i\'{c}, Kaiser, and
  Kutz}]{Brunton-et-al-22}
\bibinfo{author}{S.~L. Brunton}, \bibinfo{author}{M.~Budi\v{s}i\'{c}},
  \bibinfo{author}{E.~Kaiser}, \bibinfo{author}{J.~N. Kutz},
\newblock \bibinfo{title}{Modern {K}oopman theory for dynamical systems},
\newblock \bibinfo{journal}{SIAM Review} \bibinfo{volume}{64}
  (\bibinfo{year}{2022}) \bibinfo{pages}{229--340}.
\bibitem[{Bevanda et~al.(2021)Bevanda, Sosnowski, and Hirche}]{BEVANDA2021197}
\bibinfo{author}{P.~Bevanda}, \bibinfo{author}{S.~Sosnowski},
  \bibinfo{author}{S.~Hirche},
\newblock \bibinfo{title}{{K}oopman operator dynamical models: Learning,
  analysis and control},
\newblock \bibinfo{journal}{Annual Reviews in Control} \bibinfo{volume}{52}
  (\bibinfo{year}{2021}) \bibinfo{pages}{197--212}.
\bibitem[{S.~Otto(2021)}]{Otto-Rowley-21}
\bibinfo{author}{C.~R. S.~Otto},
\newblock \bibinfo{title}{{K}oopman operators for estimation and control of
  dynamical systems},
\newblock \bibinfo{journal}{Annual Reviews of Control, Robotics, and Autonomous
  systems} \bibinfo{volume}{4} (\bibinfo{year}{2021}) \bibinfo{pages}{59--87}.
\bibitem[{Kostic et~al.(2022)Kostic, Novelli, Maurer, Ciliberto, Rosasco, and
  Pontil}]{Kostic-22-NIPS}
\bibinfo{author}{V.~Kostic}, \bibinfo{author}{P.~Novelli},
  \bibinfo{author}{A.~Maurer}, \bibinfo{author}{C.~Ciliberto},
  \bibinfo{author}{L.~Rosasco}, \bibinfo{author}{M.~Pontil},
\newblock \bibinfo{title}{Learning dynamical systems via {K}oopman operator
  regression in reproducing kernel {H}ilbert spaces.},
\newblock in: \bibinfo{booktitle}{Advances in Neural Information Processing
  Systems}, \bibinfo{year}{2022}.
\bibitem[{Kostic et~al.(2023)Kostic, Lounici, Novelli, and Pontil}]{Kostic-23}
\bibinfo{author}{V.~Kostic}, \bibinfo{author}{K.~Lounici},
  \bibinfo{author}{P.~Novelli}, \bibinfo{author}{M.~Pontil},
  \bibinfo{title}{{K}oopman operator learning: Sharp spectral rates and
  spurious eigenvalues}, \bibinfo{howpublished}{arXiv,
  https://arxiv.org/abs/2302.02004}, \bibinfo{year}{2023}.
\bibitem[{Klus et~al.(2020)Klus, N{\"u}ske, and Hamzi}]{Klus-et-al-2020}
\bibinfo{author}{S.~Klus}, \bibinfo{author}{F.~N{\"u}ske},
  \bibinfo{author}{B.~Hamzi},
\newblock \bibinfo{title}{Kernel-based approximation of the {K}oopman generator
  and schr{\"o}dinger operator},
\newblock \bibinfo{journal}{Entropy} \bibinfo{volume}{22}
  (\bibinfo{year}{2020}).
\bibitem[{Berry and Sauer(2016)}]{berry2016}
\bibinfo{author}{T.~Berry}, \bibinfo{author}{T.~Sauer},
\newblock \bibinfo{title}{Local kernels and the geometric structure of data},
\newblock \bibinfo{journal}{Applied and Computational Harmonic Analysis}
  \bibinfo{volume}{40} (\bibinfo{year}{2016}) \bibinfo{pages}{439--469}.
\bibitem[{Banisch and Koltai(2017)}]{banish2017}
\bibinfo{author}{R.~Banisch}, \bibinfo{author}{P.~Koltai},
\newblock \bibinfo{title}{Understanding the geometry of transport: Diffusion
  maps for {L}agrangian trajectory data unravel coherent sets},
\newblock \bibinfo{journal}{Chaos: An Interdisciplinary Journal of Nonlinear
  Science} \bibinfo{volume}{27} (\bibinfo{year}{2017}) \bibinfo{pages}{035804}.
\bibitem[{Coifman and Hirn(2014)}]{coifman2014}
\bibinfo{author}{R.~R. Coifman}, \bibinfo{author}{M.~J. Hirn},
\newblock \bibinfo{title}{Diffusion maps for changing data},
\newblock \bibinfo{journal}{Applied and Computational Harmonic Analysis}
  \bibinfo{volume}{36} (\bibinfo{year}{2014}) \bibinfo{pages}{79--107}.
\bibitem[{Zhou(2008)}]{Zhou-2008}
\bibinfo{author}{D.-X. Zhou},
\newblock \bibinfo{title}{Derivative reproducing properties for kernel methods
  in learning theory},
\newblock \bibinfo{journal}{Journal of Computational and Applied Mathematics}
  \bibinfo{volume}{220} (\bibinfo{year}{2008}) \bibinfo{pages}{456--463}.
\bibitem[{Berlinet and Thomas-Agnan(2001)}]{Berlinet-Thomas-Agnan}
\bibinfo{author}{T.~Berlinet}, \bibinfo{author}{C.~Thomas-Agnan},
  \bibinfo{title}{Reproducing kernel {H}ilbert spaces in Probability and
  Statistics}, \bibinfo{publisher}{Kluwer Academic Publishers},
  \bibinfo{year}{2001}.
\bibitem[{Cucker and Smale(2001)}]{Cucker-Smale-01}
\bibinfo{author}{F.~Cucker}, \bibinfo{author}{S.~Smale},
\newblock \bibinfo{title}{On the mathematical foundation of learning},
\newblock \bibinfo{journal}{Bul. of the Amer. Math. Soc.} \bibinfo{volume}{39}
  (\bibinfo{year}{2001}) \bibinfo{pages}{1--49}.
\bibitem[{Aronszajn(1950)}]{Aronszajn-50}
\bibinfo{author}{N.~Aronszajn},
\newblock \bibinfo{title}{Theory of reproducing kernels},
\newblock \bibinfo{journal}{Trans. of the American Math. Society}
  \bibinfo{volume}{68} (\bibinfo{year}{1950}).
\bibitem[{K\"{o}nig(1986)}]{Konig86}
\bibinfo{author}{H.~K\"{o}nig},
\newblock \bibinfo{title}{Eigenvalue distribution of compact operators with
  application to integral operators},
\newblock \bibinfo{journal}{Linear Algebra and its application}
  \bibinfo{volume}{84} (\bibinfo{year}{1986}) \bibinfo{pages}{111--122}.
\bibitem[{{Le Dimet} and Talagrand(1986)}]{Ledimet86}
\bibinfo{author}{F.-X. {Le Dimet}}, \bibinfo{author}{O.~Talagrand},
\newblock \bibinfo{title}{Variational algorithms for analysis and assimilation
  of meteorological observations: theoretical aspects},
\newblock \bibinfo{journal}{Tellus} \bibinfo{volume}{38A}
  (\bibinfo{year}{1986}) \bibinfo{pages}{97--110}.
\bibitem[{Lions(1971)}]{Lions71}
\bibinfo{author}{J.~L. Lions}, \bibinfo{title}{Optimal control of systems
  governed by PDEs}, \bibinfo{publisher}{Springer-Verlag},
  \bibinfo{address}{New-York}, \bibinfo{year}{1971}.
\bibitem[{Carmeli et~al.(2006)Carmeli, De~Vito, and
  Toigo}]{Carmelli-et-al-2006}
\bibinfo{author}{C.~Carmeli}, \bibinfo{author}{E.~De~Vito},
  \bibinfo{author}{A.~Toigo},
\newblock \bibinfo{title}{Vector valued reproducing kernel {H}ilbert spaces of
  integrable functions and {M}ercer theorem},
\newblock \bibinfo{journal}{Analysis and Applications} \bibinfo{volume}{04}
  (\bibinfo{year}{2006}) \bibinfo{pages}{377--408}.
\bibitem[{Steinwart and Scovel(2012)}]{Steinward-Scovel-2012}
\bibinfo{author}{I.~Steinwart}, \bibinfo{author}{C.~Scovel},
\newblock \bibinfo{title}{Mercer's theorem on general domains: On the
  interaction between measures, kernels, and rkhss},
\newblock \bibinfo{journal}{Constructive Approximation} \bibinfo{volume}{35}
  (\bibinfo{year}{2012}) \bibinfo{pages}{363--417}.
\bibitem[{Ladyzhenskaya(1982)}]{ladyzhenskaya1982finite}
\bibinfo{author}{O.~A. Ladyzhenskaya},
\newblock \bibinfo{title}{Finite-dimensionality of bounded invariant sets for
  navier--stokes systems and other dissipative systems},
\newblock \bibinfo{journal}{Zapiski Nauchnykh Seminarov POMI}
  \bibinfo{volume}{115} (\bibinfo{year}{1982}) \bibinfo{pages}{137--155}.
\bibitem[{Constantin et~al.(1985)Constantin, Foia{\c{s}}, and
  Temam}]{Constantin-Fois-Temam-85}
\bibinfo{author}{P.~Constantin}, \bibinfo{author}{C.~Foia{\c{s}}},
  \bibinfo{author}{R.~Temam}, \bibinfo{title}{Attractors representing turbulent
  flows}, volume \bibinfo{volume}{314}, \bibinfo{publisher}{American
  Mathematical Soc.}, \bibinfo{year}{1985}.
\bibitem[{Lax(1968)}]{Lax-68}
\bibinfo{author}{P.~Lax},
\newblock \bibinfo{title}{Integrals of nonlinear equations of evolution and
  solitary waves},
\newblock \bibinfo{journal}{Communications on pure and applied mathematics}
  \bibinfo{volume}{21} (\bibinfo{year}{1968}) \bibinfo{pages}{467--490}.
\bibitem[{Li(2005)}]{Li-PAMS-05}
\bibinfo{author}{Y.~C. Li},
\newblock \bibinfo{title}{Ergodic isospectral theory of the {L}ax pairs of
  {E}uler equations with harmonic analysis flavor},
\newblock \bibinfo{journal}{Proc. of the Amer. Math. Soc.}
  \bibinfo{volume}{133} (\bibinfo{year}{2005}) \bibinfo{pages}{2681--2687}.
\bibitem[{Greatbatch and Nadiga(2000)}]{greatbatch2000}
\bibinfo{author}{R.~J. Greatbatch}, \bibinfo{author}{B.~T. Nadiga},
\newblock \bibinfo{title}{Four-gyre circulation in a barotropic model with
  double-gyre wind forcing},
\newblock \bibinfo{journal}{Journal of Physical Oceanography}
  \bibinfo{volume}{30} (\bibinfo{year}{2000}) \bibinfo{pages}{1461--1471}.
\bibitem[{Brannan et~al.(1998)Brannan, Duan, and Wanner}]{Brannan-Duan-98}
\bibinfo{author}{J.~R. Brannan}, \bibinfo{author}{J.~Duan},
  \bibinfo{author}{T.~Wanner},
\newblock \bibinfo{title}{Dissipative quasi-geostrophic dynamics under random
  forcing},
\newblock \bibinfo{journal}{Journal of Mathematical Analysis and Applications}
  \bibinfo{volume}{228} (\bibinfo{year}{1998}) \bibinfo{pages}{221--233}.
\bibitem[{Duan and Goldys(2001)}]{Duan-2001}
\bibinfo{author}{J.~Duan}, \bibinfo{author}{B.~Goldys},
\newblock \bibinfo{title}{Ergodicity of stochastically forced large scale
  geophysical flows},
\newblock \bibinfo{journal}{International Journal of Mathematics and
  Mathematical Sciences} \bibinfo{volume}{28} (\bibinfo{year}{2001})
  \bibinfo{pages}{313--320}.
\bibitem[{Yang and Pu(2017)}]{Yang-Pu-17}
\bibinfo{author}{L.~Yang}, \bibinfo{author}{X.~Pu},
\newblock \bibinfo{title}{Ergodicity of large scale stochastic geophysical
  flows with degenerate gaussian noise},
\newblock \bibinfo{journal}{Applied Mathematics Letters} \bibinfo{volume}{64}
  (\bibinfo{year}{2017}) \bibinfo{pages}{27--33}.
\bibitem[{Hamill et~al.(2001)Hamill, Whitaker, and Snyder}]{hamill2001}
\bibinfo{author}{T.~M. Hamill}, \bibinfo{author}{J.~S. Whitaker},
  \bibinfo{author}{C.~Snyder},
\newblock \bibinfo{title}{Distance-dependent filtering of background error
  covariance estimates in an ensemble {K}alman filter},
\newblock \bibinfo{journal}{Monthly Weather Review} \bibinfo{volume}{129}
  (\bibinfo{year}{2001}) \bibinfo{pages}{2776--2790}.
\bibitem[{Minh(2010)}]{minh2010}
\bibinfo{author}{H.~Q. Minh},
\newblock \bibinfo{title}{Some properties of {G}aussian reproducing kernel
  {H}ilbert spaces and their implications for function approximation and
  learning theory},
\newblock \bibinfo{journal}{Constructive Approximation} \bibinfo{volume}{32}
  (\bibinfo{year}{2010}) \bibinfo{pages}{307--338}.
\bibitem[{Le~Traon et~al.(1998)Le~Traon, Nadal, and Ducet}]{traon1998}
\bibinfo{author}{P.~Y. Le~Traon}, \bibinfo{author}{F.~Nadal},
  \bibinfo{author}{N.~Ducet},
\newblock \bibinfo{title}{An improved mapping method of multisatellite
  altimeter data},
\newblock \bibinfo{journal}{Journal of Atmospheric and Oceanic Technology}
  \bibinfo{volume}{15} (\bibinfo{year}{1998}) \bibinfo{pages}{522 -- 534}.
\bibitem[{Ducet et~al.(2000)Ducet, Le~Traon, and Reverdin}]{ducet2000}
\bibinfo{author}{N.~Ducet}, \bibinfo{author}{P.~Y. Le~Traon},
  \bibinfo{author}{G.~Reverdin},
\newblock \bibinfo{title}{Global high-resolution mapping of ocean circulation
  from topex/poseidon and ers-1 and -2},
\newblock \bibinfo{journal}{Journal of Geophysical Research: Oceans}
  \bibinfo{volume}{105} (\bibinfo{year}{2000}) \bibinfo{pages}{19477--19498}.
\bibitem[{AVISO(2016)}]{aviso}
\bibinfo{author}{AVISO}, \bibinfo{title}{SSALTO/DUACS User Handbook: MSLA and
  (M)ADT Near-Real Time and Delayed Time Products}, \bibinfo{type}{Technical
  Report} \bibinfo{number}{CLS-DOS-NT-06-034, SALP-MU-P-EA-21065-CLS, 5rev0},
  CNES, \bibinfo{year}{2016}. \bibinfo{note}{Available online at
  https://www.aviso.altimetry.fr/fileadmin/documents/data/tools/hdbk\_duacs.pdf}.
\bibitem[{Ubelmann et~al.(2016)Ubelmann, Cornuelle, and Fu}]{ubelmann2016}
\bibinfo{author}{C.~Ubelmann}, \bibinfo{author}{B.~Cornuelle},
  \bibinfo{author}{L.-L. Fu},
\newblock \bibinfo{title}{Dynamic mapping of along-track ocean altimetry:
  Method and performance from observing system simulation experiments},
\newblock \bibinfo{journal}{Journal of Atmospheric and Oceanic Technology}
  \bibinfo{volume}{33} (\bibinfo{year}{2016}) \bibinfo{pages}{1691--1699}.
\bibitem[{Guillou et~al.(2021)Guillou, Metref, Cosme, Ubelmann, Ballarotta,
  Le~Sommer, and Verron}]{leguillou2021}
\bibinfo{author}{F.~Guillou}, \bibinfo{author}{S.~Metref},
  \bibinfo{author}{E.~Cosme}, \bibinfo{author}{C.~Ubelmann},
  \bibinfo{author}{M.~Ballarotta}, \bibinfo{author}{J.~Le~Sommer},
  \bibinfo{author}{J.~Verron},
\newblock \bibinfo{title}{Mapping altimetry in the forthcoming swot era by
  back-and-forth nudging a one-layer quasigeostrophic model},
\newblock \bibinfo{journal}{Journal of Atmospheric and Oceanic Technology}
  \bibinfo{volume}{38} (\bibinfo{year}{2021}) \bibinfo{pages}{697 -- 710}.
\bibitem[{Chaturantabut and Sorensen(2010)}]{chaturantabut2010}
\bibinfo{author}{S.~Chaturantabut}, \bibinfo{author}{D.~C. Sorensen},
\newblock \bibinfo{title}{Nonlinear model reduction via discrete empirical
  interpolation},
\newblock \bibinfo{journal}{SIAM Journal of Scientific Computing}
  \bibinfo{volume}{32} (\bibinfo{year}{2010}) \bibinfo{pages}{2737--2764}.
\bibitem[{M\'emin(2014)}]{Memin14}
\bibinfo{author}{E.~M\'emin},
\newblock \bibinfo{title}{Fluid flow dynamics under location uncertainty},
\newblock \bibinfo{journal}{Geophys. \& Astro. Fluid Dyn.}
  \bibinfo{volume}{108} (\bibinfo{year}{2014}) \bibinfo{pages}{119--146}.
\bibitem[{Bauer et~al.(2020)Bauer, Chandramouli, Chapron, Li, and
  M{\'e}min}]{Bauer-et-al-JPO-20}
\bibinfo{author}{W.~Bauer}, \bibinfo{author}{P.~Chandramouli},
  \bibinfo{author}{B.~Chapron}, \bibinfo{author}{L.~Li},
  \bibinfo{author}{E.~M{\'e}min},
\newblock \bibinfo{title}{Deciphering the role of small-scale inhomogeneity on
  geophysical flow structuration: A stochastic approach},
\newblock \bibinfo{journal}{Journal of Physical Oceanography}
  \bibinfo{volume}{50} (\bibinfo{year}{01 Apr. 2020}) \bibinfo{pages}{983 --
  1003}.
\bibitem[{Chapron et~al.(2018)Chapron, D{\'e}rian, M{\'e}min, and
  Resseguier}]{Chapron-18}
\bibinfo{author}{B.~Chapron}, \bibinfo{author}{P.~D{\'e}rian},
  \bibinfo{author}{E.~M{\'e}min}, \bibinfo{author}{V.~Resseguier},
\newblock \bibinfo{title}{Large-scale flows under location uncertainty: a
  consistent stochastic framework},
\newblock \bibinfo{journal}{QJRMS} \bibinfo{volume}{144} (\bibinfo{year}{2018})
  \bibinfo{pages}{251--260}.
\bibitem[{Debussche et~al.(2023)Debussche, Hug, and M{\'e}min}]{Debussche-2023}
\bibinfo{author}{A.~Debussche}, \bibinfo{author}{B.~Hug},
  \bibinfo{author}{E.~M{\'e}min},
\newblock \bibinfo{title}{A consistent stochastic large-scale representation of
  the {N}avier--{S}tokes equations},
\newblock \bibinfo{journal}{Journal of Mathematical Fluid Mechanics}
  \bibinfo{volume}{25} (\bibinfo{year}{2023}) \bibinfo{pages}{19}.
\bibitem[{Duf{\'e}e et~al.(2023)Duf{\'e}e, M{\'e}min, and
  Crisan}]{Dufee-STUOD-22}
\bibinfo{author}{B.~Duf{\'e}e}, \bibinfo{author}{E.~M{\'e}min},
  \bibinfo{author}{D.~Crisan},
\newblock \bibinfo{title}{{Observation-Based Noise Calibration: An Efficient
  Dynamics for the Ensemble Kalman Filter}},
\newblock in: \bibinfo{booktitle}{{Stochastic Transport in Upper Ocean
  Dynamics}}, volume~\bibinfo{volume}{10} of
  \textit{\bibinfo{series}{Mathematics of Planet Earth}},
  \bibinfo{publisher}{{Springer International Publishing}},
  \bibinfo{year}{2023}, pp. \bibinfo{pages}{43--56}. \URLprefix
  \url{https://hal.science/hal-03910764}.
  \DOIprefix\doi{10.1007/978-3-031-18988-3\_4}.
\bibitem[{Duf{\'e}e et~al.(2022)Duf{\'e}e, M{\'e}min, and
  Crisan}]{Dufee-QJRMS-22}
\bibinfo{author}{B.~Duf{\'e}e}, \bibinfo{author}{E.~M{\'e}min},
  \bibinfo{author}{D.~Crisan},
\newblock \bibinfo{title}{Stochastic parametrization: An alternative to
  inflation in ensemble {K}alman filters},
\newblock \bibinfo{journal}{Quarterly Journal of the Royal Meteorological
  Society} \bibinfo{volume}{148} (\bibinfo{year}{2022})
  \bibinfo{pages}{1075--1091}.
\bibitem[{Simon-Gabriel and Sch{{\"o}}lkopf(2018)}]{Simon-Gabriel-Scholkopf-18}
\bibinfo{author}{C.-J. Simon-Gabriel}, \bibinfo{author}{B.~Sch{{\"o}}lkopf},
\newblock \bibinfo{title}{Kernel distribution embeddings: Universal kernels,
  characteristic kernels and kernel metrics on distributions},
\newblock \bibinfo{journal}{Journal of Machine Learning Research}
  \bibinfo{volume}{19} (\bibinfo{year}{2018}) \bibinfo{pages}{1--29}.
\bibitem[{Steinwart and Christmann(2008)}]{Steinward-Christmann}
\bibinfo{author}{I.~Steinwart}, \bibinfo{author}{A.~Christmann},
  \bibinfo{title}{Support Vector Machines. Information Science and Statistics},
  \bibinfo{publisher}{Springer}, \bibinfo{year}{2008}.
\bibitem[{Arakawa(1966)}]{arakawa1966}
\bibinfo{author}{A.~Arakawa},
\newblock \bibinfo{title}{{Computational design for long-term numerical
  integration of the equations of fluid motion: Two-dimensional incompressible
  flow. Part I}},
\newblock \bibinfo{journal}{Journal of Computational Physics}
  \bibinfo{volume}{1} (\bibinfo{year}{1966}) \bibinfo{pages}{119--143}.
\bibitem[{Lumley(1967)}]{lumley1967}
\bibinfo{author}{J.~L. Lumley},
\newblock \bibinfo{title}{The structure of inhomogeneous turbulent flows},
\newblock \bibinfo{journal}{Atmospheric turbulence and radio wave propagation}
  \bibinfo{volume}{1} (\bibinfo{year}{1967}) \bibinfo{pages}{166--178}.
\bibitem[{San et~al.(2011)San, Staples, Wang, and Iliescu}]{san2011}
\bibinfo{author}{O.~San}, \bibinfo{author}{A.~E. Staples},
  \bibinfo{author}{Z.~Wang}, \bibinfo{author}{T.~Iliescu},
\newblock \bibinfo{title}{Approximate deconvolution large eddy simulation of a
  barotropic ocean circulation model},
\newblock \bibinfo{journal}{Ocean Modelling} \bibinfo{volume}{40}
  (\bibinfo{year}{2011}) \bibinfo{pages}{120--132}.

\end{thebibliography}

\end{document}